\documentclass[11pt]{article}
\usepackage{amssymb,amsmath, amsthm, amsopn, amsfonts,bm,mathtools}
\usepackage[english]{babel}
\usepackage{color}\definecolor{rd}{rgb}{1,0.3,0.35}
\usepackage{graphicx}
\usepackage{stmaryrd}
\usepackage{mathrsfs}

\newcommand{\hil}{\mathscr{Z}}
\DeclareMathOperator*{\card}{card}

\usepackage[hmargin=2.5cm,vmargin=3cm,centering]{geometry}
\newtheorem{thm}{Theorem}[section]

\newtheorem{lem}[thm]{Lemma}
\newtheorem{definition}[thm]{Definition}
\newtheorem{prop}[thm]{Proposition}
\newtheorem{remark}[thm]{Remark}
\newtheorem{hyp}[thm]{Hypothesis}

\def\Z{{\mathscr{Z}}}

\def\P{{\mathcal{P}}}
\def\L{{\mathcal{L}}}
\def\H{{\mathcal{H}}}

\def\N{{\bf N}}
\newcommand{\field}[1]{\mathbb{#1}}
\newcommand{\rz}{\field{R}}
\newcommand{\cz}{\field{C}}
\newcommand{\nz}{\field{N}}

\newcommand{\sz}{\field{S}}
\newcommand{\tz}{\field{T}}
\newcommand{\chid}{\chi}

\def\11{{\rm 1~\hspace{-1.2ex}l} }

\def\Tr{{\rm{Tr}}}

\def\Real{{\mathrm{Re~}}}
\def\Imag{\mathrm{Im~}}

\def\Id{\mathrm{Id}}
\begin{document}
\title{Quantum mean-field asymptotics 
and multiscale analysis}
\author{
Z.~Ammari\thanks{zied.ammari@univ-rennes1.fr,  IRMAR, Universit{\'e} de Rennes I, UMR-CNRS 6625,
campus de Beaulieu, 35042 Rennes Cedex, France.},
\quad{S.~Breteaux\thanks{sebastien.breteaux@ens-cachan.org, BCAM - Basque Center for Applied Mathematics, Alameda de Mazarredo~14, 48009 Bilbao, Spain.}} \quad and
\quad
F.~Nier\thanks{francis.nier@math.univ-paris13.fr, 
LAGA, UMR-CNRS 9345, Universit{\'e} de Paris~13, av.~J.B.~Cl{\'e}ment, 93430
Villetaneuse, France.}
}
\maketitle

\begin{abstract}
We study, via multiscale analysis, some defect of compactness phenomena which occur in bosonic and fermionic quantum mean-field problems. The approach relies on a combination of mean-field asymptotics and second microlocalized semiclassical measures. The phase space geometric description is illustrated by various examples.
\end{abstract}

\bigskip
\noindent
{\it Keywords and phrases: }{ Semiclassical and multiscale measures, reduced density matrices, second quantization, 
microlocal analysis.}\\
{\it 2010 Mathematics subject classification: }{ 81S30, 81Q20, 81V70, 81S05.}

\section{Introduction}
\paragraph{Motivations:}
Over the past two decades, it becomes clear that microlocal analysis provides interesting mathematical tools for the study of quantum field theories and quantum many-body theory, see for instance \cite{AmNi1, Rad}.
In particular, in the analysis of general bosonic mean-field problems, as done in \cite{AmNi1,AmNi2,AmNi3,AmNi4}, the following defect of compactness problem arises. 
If $\gamma_\varepsilon^{(p)}$ denotes the $p$-particles reduced density matrix, 
one may have 
\begin{equation}\label{eq:defect}
\lim_{\varepsilon\to 0}\Tr[\gamma_\varepsilon^{(p)} \tilde{b}]=\Tr[\gamma_0^{(p)} \tilde b]
\end{equation} 
for any $p$-particle \underline{compact} observable $\tilde b$\,, 
while it is not true for a general bounded $\tilde{b}$\,, e.g. 
$$
\lim_{\varepsilon\to 0}\Tr[\gamma_\varepsilon^{(p)}] > \Tr[\gamma_0^{(p)}]\,.
$$
In the fermionic case, it is even worse, because mean-field asymptotics cannot be described in terms of finitely many quantum states and the right-hand side of \eqref{eq:defect} is usually $0$ while $\lim_{\varepsilon\to 0}\Tr[\gamma_\varepsilon^{(p)}] >0$\,.
From the analysis of finite dimensional partial differential equations, it is known that such defect of compactness can be localized geometrically with accurate quantitative information by introducing scales and small parameters within semiclassical techniques (e.g. \cite{Ger,GMMP}).
We are thus led to introduce two small parameters $\varepsilon>0$ for the mean-field asymptotics and $h>0$ for the semiclassical quantization of finite dimensional $p$-particles phase space. The small parameter $\varepsilon$ stands for $1/n$\,, where $n\to \infty$ is the typical number of particles, while $h$ is the \underline{rescaled} Planck constant measuring the proximity of quantum mechanics to classical mechanics.
The combined analysis of this article is concerned with the general situation when $\varepsilon=\varepsilon(h)$ with $\lim_{h\to0}\varepsilon(h)=0$\,.
In order to keep track of the information at the quantum level especially in the bosonic case we also introduce finite dimensional multiscale observables in spirit of \cite{Bon,FeGe,Fer,Nie}.

\paragraph{Framework:}
The one particle space $\Z$ is a separable complex Hilbert space
endowed with the scalar product $\langle~\,,\,~\rangle$ (anti-linear
in the left-hand side). 
For a Hilbert space $\mathfrak{h}$ the set of bounded operator is denoted by
$\mathcal{L}(\mathfrak{h})$\,, while the Schatten class is denoted by 
$\mathcal{L}^{p}(\mathfrak{h})$\,, $1\leq p \leq \infty$\,, the case
$p=\infty$ corresponding to the space of compact operators. 
Let $\Gamma_{\pm}(\Z)$  be the bosonic ($+$ sign) or fermionic
($-$ sign) Fock-space built on the separable Hilbert space $\Z$~:
$$
\Gamma_{\pm}(\Z)=\bigoplus^{\perp}_{n\in\nz}{\cal
  S}^{n}_{\pm}{\Z}^{\otimes n}\,,
$$
where tensor products and direct sum are Hilbert completed. The
operator ${\cal S}^{n}_{\pm}$ is the orthogonal projection given by
$$
{\cal S}^{n}_{\pm}(f_{1}\otimes \cdots\otimes f_{n})
=\frac{1}{n!}\sum_{\sigma\in
  \mathfrak{S}_{n}}s_{\pm}(\sigma)f_{\sigma(1)}\otimes
\cdots\otimes f_{\sigma(n)}\,,
$$
where $s_{+}(\sigma)$ equals $1$ while $s_{-}(\sigma)$ denotes the
signature of the permutation $\sigma$ and $ \mathfrak{S}_{n}$ is the $n$-symmetric group.\\
The dense set of many-body states with a finite number of particles is
$$
\Gamma^{fin}_{\pm}(\Z)=\bigoplus^{\perp, alg}_{n\in\nz}{\cal
  S}^{n}_{\pm}\Z^{\otimes n}\,,
$$
where the ${~}^{\perp, alg}$ superscript stands for the algebraic orthogonal
direct sum.\\
We shall use the notations $[A,B]_{+}=[A,B]=\textrm{ad}_{A}B=AB-BA$ 
for the commutator of two
operators and the notation $[A,B]_{-}=AB+BA$ for the anticommutator.\\
One way to investigate the mean-field asymptotics relies on
parameter-dependent $CCR$ (resp. $CAR$)\,. The small parameter
$\varepsilon>0$ has to be thought of as the inverse of the typical
number of particles and the Canonical Commutation
(resp. Anticommutation) Relations are given by
$$
[a_{\pm}(g),a_{\pm}(f)]_{\pm}=[a_{\pm}^{*}(g),a_{\pm}^{*}(f)]_{\pm}=0\quad,\quad
[a_{\pm}(g), a_{\pm}^{*}(f)]_{\pm}=\varepsilon\langle g\,,\,f\rangle\,.
$$
Let $(\varrho_{\varepsilon})_{\varepsilon>0}$ be a family of
normal states (i.e., non-negative and normalized  trace-class operators) on the Fock space $\Gamma_{\pm}(\Z)$\,,
depending on $\varepsilon>0$\,, we
want to investigate the asymptotic behaviour of reduced density matrices,
defined below,
 as $\varepsilon\to 0$\,, by possibly introducing another scale
 $h>0$ on the $p$-particles phase-space, with  $\varepsilon=\varepsilon(h)$  and
 $\lim_{h\to 0}\varepsilon(h)=0$\,.

\paragraph{Outline:}
In Section~\ref{se.wickobs}, we recall how Wick observables are used to define the reduced density matrices $\gamma^{(p)}_{\varepsilon}$\,. Note that it is much more convenient here, in the general grand canonical framework, to work with non normalized reduced density matrices. Some symmetrization formulas are also recalled in this section.
In Section~\ref{sec:classic}, we present the geometry of the classical $p$-particles phase space and introduce the formalism of double scale semiclassical measures, after \cite{Fer,FeGe}.
In Section~\ref{sec:meanfield}, we combine the mean-field asymptotics with semiclassical analysis, the two parameters $\varepsilon$ and $h$ being related through $\varepsilon=\varepsilon(h)$ with $\lim_{h\to 0}\varepsilon(h)=0$\,. Instead of studying the collection of non normalized reduced density matrices  $(\gamma^{(p)}_{\varepsilon(h)})_{p\in\mathbb{N}}$\,, it is more convenient to associate generating functions $z\mapsto \Tr\big[\varrho_{\varepsilon(h)} \,e^{z\,d\Gamma_\pm(a^{Q,h})}\big]$\,, and to use holomorphy arguments presented there.
In Section~\ref{sec:examples}, some classical examples with various asymptotics illustrate the general framework: coherent states in the bosonic setting; simple Gibbs states in the fermionic case; more involved Gibbs states in the bosonic case, which make explicit the separation of condensate and non condensate phases for rather general non interacting steady Bose gases.
The Appendices collect or revisit known things about multiscale semiclassical measures, the (PI)-condition of bosonic mean-field problems, Wick composition formulas, and traces of non self-adjoint second quantized contractions.

\section{Wick observables and reduced density matrices}
\label{se.wickobs}

\subsection{Wick observables}
\label{sec:wick-observables}

\textbf{Notation:} For $n\in\nz$\,, the operator $\mathcal{S}^{n}_{\pm}$ is an orthogonal
projection in $\Z^{\otimes n}$  so that
$(\mathcal{S}^{n}_{\pm})^{*}=\mathcal{S}^{n}_{\pm}$\,. However, we
consider $\mathcal{S}^{n}_{\pm}$ as a bounded operator from
$\Z^{\otimes n}$ onto $\mathcal{S}^{n}_{\pm}\Z^{\otimes n}$ and its
adjoint, 
denoted by $\mathcal{S}^{n,*}_{\pm}:\mathcal{S}^{n}_{\pm}\Z^{\otimes
  n}\to \Z^{\otimes n}$\,, is nothing but the natural embedding.\\
Let $\tilde{b}\in {\cal L}({\cal S}_{\pm}^{p}\Z^{\otimes p}\,;\, {\cal
  S}^{q}_{\pm}\Z^{\otimes q})$\,, the Wick quantization of $\tilde{b}$ is the operator on $\Gamma^{fin}_{\pm}(\Z)$ defined by
$$
\tilde{b}^{Wick}\big|_{{\cal S}^{n+p}_{\pm}\Z^{\otimes(n+p)}}
=\varepsilon^{\frac{p+q}{2}}\frac{\sqrt{(n+p)!(n+q)!}}{n!}{\cal
  S}^{n+q}_{\pm}(
\tilde{b}\otimes \Id_{\Z^{\otimes n}}){\cal S}^{n+p,*}_{\pm}\,.
$$
In the bosonic case, an element $\tilde{b}\in {\cal L}({\cal
  S}_{+}^{p}\Z^{\otimes p};{\cal S}_{+}^{q}\Z^{\otimes q})$ is
determined by the symbol $\Z\ni z\mapsto b(z)=\langle z^{\otimes q}\,,\,
\tilde{b}z^{\otimes p}\rangle$ owing to the relation
$\tilde{b}=\frac{1}{q!p!}\partial_{\overline{z}}^{q}\partial_{z}^{p}b$\,. Observe that   
$b(z)$ admits higher Fr\'echet derivatives with the natural identification of 
$\partial_z ^{k} b(z)$ as a continuous form on $ {\cal
  S}_{+}^{k}\Z^{\otimes k}$ and $\partial_{\bar z} ^{k} b(z)$ as a vector  in 
   $ {\cal   S}_{+}^{k}\Z^{\otimes k}$\,. We shall use also the
 notation $b^{Wick}=\tilde{b}^{Wick}$\,.
 
\bigskip\noindent 
Examples: 
\begin{description}
\item[a)] The annihilation operator $a_{\pm}(f)$\,, $f\in\Z$\,, is the Wick quantization of 
$\tilde{b}=\langle f|: \Z^{\otimes 1}=\Z\ni \varphi\mapsto \langle
f\,,\,\varphi\rangle\in \Z^{\otimes 0}=\cz$\,.
\item[b)] The creation operator $a_{\pm}^{*}(f)$\,, $f\in\Z$\,, is the Wick quantization of  
$\tilde{b}=|f\rangle: \Z^{\otimes 0}=\cz\ni \lambda\mapsto \lambda
f\in \Z^{\otimes 1}=\Z$\,.
\item[c)] For $\tilde{b}\in {\cal L}(\Z)$ its Wick quantization 
$\tilde{b}^{Wick}$ is nothing but
$$
d\Gamma_{\pm}(\tilde{b})\big|_{{\cal S}^{n}_{\pm}\Z^{\otimes n}}=\varepsilon
\left[\tilde{b}\otimes \Id_{\Z}\otimes\cdots\otimes \Id_{\Z}
+\cdots +\Id_{\Z}\otimes\cdots\otimes \Id_{\Z}\otimes\tilde{b} \right]\,.
$$
\end{description}
When $\tilde{b}$ is self-adjoint one has
$$
d\Gamma_{\pm}(\tilde{b})=i\partial_{t}e^{-itd\Gamma_{\pm}(\tilde{b})}\big|_{t=0}
=i\partial_{t}\Gamma_{\pm}(e^{-i\varepsilon t \tilde{b}})\big|_{t=0}\,,
$$
while for a contraction $C\in {\cal L}(\Z;\Z)$\,,
$$
\Gamma_{\pm}(C)\big|_{{\cal S}^{n}_{\pm}\Z^{\otimes
    n}}=C\otimes\cdots\otimes C\,.
$$
A particular case is $\tilde{b}=\Id_{\Z}$ associated with the scaled
number operator ($\mathbf{N}_{\pm,\varepsilon=1}$ stands for the usual
$\varepsilon$-independent number operator):
$$
\tilde{b}^{Wick}=d\Gamma_{\pm}(\Id_{\Z})=\mathbf{N}_{\pm}=
\varepsilon\mathbf{N}_{\pm,\varepsilon=1}\,.
$$
From the definition of the Wick quantization one easily checks the following properties.
\begin{prop}
\label{pr:Wickpm}
  For $\tilde{b}\in {\cal L}({\cal S}_{\pm}^{p}\Z^{\otimes p}\,;\, {\cal
  S}^{q}_{\pm}\Z^{\otimes q})$\,:
  \begin{itemize}
  \item $\left[\tilde{b}^{Wick}\right]^{*}=[\tilde{b}^{*}]^{Wick}$\,.
  \item The operator
    $(1+\mathbf{N}_{\pm})^{-m/2}\tilde{b}^{Wick}(1+\mathbf{N}_{\pm})^{-m'/2}$
 extends to a bounded operator on $\Gamma_{\pm}(\Z)$ as soon as $m+m'\geq p+q$ with 
\begin{equation}
      \label{eq:numbEst}
\|(1+\mathbf{N}_{\pm})^{-m/2}\tilde{b}^{Wick}(1+\mathbf{N}_{\pm})^{-m'/2}
\|_{{\cal L}(\Gamma_{\pm}(\Z))}\leq C_{m,m'}
\|\tilde{b}\|_{{\cal L}({\cal S}^{p}_{\pm}\Z;{\cal S}^{q}_{\pm}\Z)}\,,
\end{equation}
with $C_{m,m'}$ independent of $\tilde{b}$ and of 
$\varepsilon\in (0,\varepsilon_{0})$\,.
\item $(\tilde{b}\geq 0)\Leftrightarrow (\tilde{b}^{Wick}\geq 0)$\,,
  while this makes sense only for $q=p$\,.
  \end{itemize}
\end{prop} 
The Wick quantized operators generally are  unbounded operators on
$\Gamma_{\pm}(\Z)$ (e.g. $\mathbf{N}_{\pm}$) but they are well defined on
the dense set $\Gamma^{fin}_{\pm}(\Z)$ which is preserved by their
action. Hence $\tilde{b}_{1}^{Wick}\circ\tilde{b}_{2}^{Wick}$ makes
sense at least on $\Gamma^{fin}_{\pm}(\Z)$ and the following  composition law holds true. 


\begin{prop}[Composition of Wick operators]
\label{pro:composition_wick}Let $\tilde{b}_{j}\in\mathcal{L}(\mathcal{S}_{\pm}^{p_{j}}\hil^{\otimes p_{j}};\mathcal{S}_{\pm}^{q_{j}}\hil^{\otimes q_{j}})$,
$j=1,2$, then
\begin{equation}
\tilde{b}_{1}^{Wick}\circ\tilde{b}_{2}^{Wick}=\sum_{k=0}^{\min\{p_{1},q_{2}\}}(\pm1)^{(p_{1}-k)(p_{2}+q_{2})}\,\frac{\varepsilon^{k}}{k!}\,(\tilde{b}_{1}\sharp^{k}\tilde{b}_{2})^{Wick}\,,\label{eq:Composition_Wick}
\end{equation}
where $\tilde{b}_{1}\sharp^{k}\tilde{b}_{2}:=\tfrac{p_{1}!}{(p_{1}-k)!}\tfrac{q_{2}!}{(q_{2}-k)!}\,\mathcal{S}_{\pm}^{q_{1}+q_{2}-k}\,(\tilde{b}_{1}\otimes \Id^{\otimes q_{2}-k})\,(\Id^{\otimes p_{1}-k}\otimes\tilde{b}_{2})\,\mathcal{S}_{\pm}^{p_{1}+p_{2}-k,*}$.
\end{prop}
\noindent
For reader's convenience, the proof of Prop.~\ref{pro:composition_wick} is provided in 
Appendix~\ref{sec:Proof_Composition_Wick}.

In the bosonic case the symbols $b(z)=\langle z^{\otimes q}\,,\,
\tilde{b} z^{\otimes p}\rangle$ are convenient for writing the
composition of Wick quantized operators.  If $b_{1}\sharp^{Wick}b_{2}$ denotes the symbol of
${\tilde{b}}_{1}^{Wick}\circ \tilde{b}_{2}^{Wick}$\,, the composition law is  
summarized below (see \cite{AmNi1}). 

\begin{prop}[Composition of Wick symbols in the bosonic case]
\label{pr:Wbos}
$$
b_{1}\sharp^{Wick}b_{2}(z)=e^{\varepsilon\partial_{z_{1}}\cdot\partial_{\overline{z_{2}}}}b_{1}(z_{1})b_{2}(z_{2})\big|_{z_{1}=z_{2}=z}
=\sum_{k=0}^{\min\left\{p_{1},q_{2}\right\}}\frac{\varepsilon^{k}}{k!}\partial_{z}^{k}b_{1}(z)\cdot\partial_{\overline{z}}^{k}b_{2}(z)\,.
$$
The commutator of Wick operators in the bosonic case:
$$
\left[b_{1}^{Wick}\,,\, b_{2}^{Wick}\right]
=\left(
\sum_{k=1}^{\max\{\min\left\{p_{1},q_{2}\right\},\min\left\{p_{2},q_{1}\right\}\}}
\frac{\varepsilon^{k}}{k!}
\left\{b_{1},b_{2}\right\}^{(k)}
\right)^{Wick}\,,
$$
where the $k$-th order Poisson bracket is given by $\left\{b_{1},b_{2}\right\}^{(k)}=\partial_{z}^{k}b_{1}(z)\cdot\partial_{\overline{z}}^{k}b_{2}(z)-\partial_{z}^{k}b_{2}(z)\cdot\partial_{\overline{z}}^{k}b_{1}(z)$\,.
\end{prop}

\begin{prop}\label{cor:numberest}
Let $p$, $m$, $m'\in\mathbb{N}$, such that $m+m'\geq2p-2$. Then, there exist coefficients $C_{j_{1},\dots, j_{k}}\geq0$ such that, for any $\tilde{b}\in\mathcal{L}(\mathcal{Z};\mathcal{Z})$\,,
\begin{equation}\label{eqn:dGamma-bp-bpWick}
d\Gamma_{\pm}(\tilde{b})^{p}-(\tilde{b}^{\otimes p})^{Wick}=\sum_{k=1}^{p-1}\varepsilon^{p-k} \!\!\! \sum_{\substack{0\leq j_{1}\leq\cdots\leq j_{k}\\
j_{1}+\cdots+j_{k}=p
}
} \!\!\! C_{j_{1},\dots ,j_{k}}\, \left(\mathcal{S}_{\pm}^{k}\,\tilde{b}^{j_{1}}\otimes\cdots\otimes\tilde{b}^{j_{k}}\,\mathcal{S}_{\pm}^{k,*}\right)^{Wick}
\end{equation}
and the estimate
\[
\|(1+\mathbf{N}_{\pm})^{-m/2} 
\big( d\Gamma_{\pm}(\tilde{b})^{p}-(\tilde{b}^{\otimes p})^{Wick} \big) 
(1+\mathbf{N}_{\pm})^{-m'/2}\|_{\mathcal{L}(\Gamma_{\pm}(\mathcal{Z}))}\leq \varepsilon \, B_{p}\, \|\tilde{b}\|_{\mathcal{L}(\mathcal{Z})}^{p}
\]
holds in both the bosonic case and the fermionic case, with $B_{p}$ the $p$-th Bell number.
\end{prop}
\begin{remark}
The $p$-th Bell number $B_p$ can be defined as the number of partitions of a set with $p$ elements and satisfies $B_{p}<\Big(\frac{0.792p}{\ln(p+1)}\Big)^{p}$ (see \cite{BeTa}), and hence it grows much slower than $p!$\,.
\end{remark}
\begin{proof}
We first prove Formula~\eqref{eqn:dGamma-bp-bpWick} by induction on $p\in\mathbb{N}^{*}$. 

For $p=1$, Formula~\eqref{eqn:dGamma-bp-bpWick} holds because $d\Gamma_{\pm}(\tilde{b})=(\tilde{b})^{Wick}$\,.

We then set $r_{p}(\tilde{b}):=d\Gamma_{\pm}(\tilde{b})^{p}-(\tilde{b}^{\otimes p})^{Wick}$. 
Assuming the result holds for some $p\in\mathbb{N}^{*}$, one can
compute 
\[
d\Gamma_{\pm}(\tilde{b})^{p+1}=(\tilde{b}^{\otimes p})^{Wick}(\tilde{b})^{Wick}+r_{p}(\tilde{b})^{Wick}(\tilde{b})^{Wick}
\]
using the composition formula \eqref{eq:Composition_Wick} for
\[
(\tilde{b}^{\otimes p})^{Wick}(\tilde{b})^{Wick}=(\tilde{b}^{\otimes p+1})^{Wick}+p\varepsilon\,\big(\mathcal{S}_{\pm}^{p}\,\,\tilde{b}^{\otimes p-1}\otimes\tilde{b}^{2}\,\,\mathcal{S}_{\pm}^{p,*}\big)^{Wick}
\]
and for
\begin{multline*}
\varepsilon^{p-k}\big(\mathcal{S}_{\pm}^{k}\,\tilde{b}^{j_{1}}\otimes\cdots\otimes\tilde{b}^{j_{k}}\,\mathcal{S}_{\pm}^{k,*}\big)^{Wick}(\tilde{b})^{Wick}\\
=\varepsilon^{p+1-(k+1)}\big(\mathcal{S}_{\pm}^{k+1}\,\tilde{b}\otimes\tilde{b}^{j_{1}}\otimes\cdots\otimes\tilde{b}^{j_{k}}\,\mathcal{S}_{\pm}^{k+1,*}\big)^{Wick}\\
+k\varepsilon^{p+1-k}\Big(\mathcal{S}_{\pm}^{k}\big(\tilde{b}^{j_{1}}\otimes\cdots\otimes\tilde{b}^{j_{k}}\big)\mathcal{S}_{\pm}^{k,*}\,\mathcal{S}_{\pm}^{k}\,\big(\tilde{b}\otimes\Id_{\mathcal{Z}}^{\otimes j_{1}+\cdots+j_{k}-1}\big)\mathcal{S}_{\pm}^{k,*}\Big)^{Wick}\,,
\end{multline*}
which yields the expected form for $r_{p+1}(\tilde{b})$, and achieves
the induction.

We then remark that the sum of coefficients of order $k$, 
\[
S_{2}(p,k)=\sum_{\substack{0\leq j_{1}\leq\cdots\leq j_{k}\\
j_{1}+\cdots+j_{k}=p
}
}C_{j_{1},\dots, j_{k}}\,,
\]
satisfies the  recurrence relation $S_{2}(p,k)=kS_{2}(p-1,k)+S_{2}(p-1,k-1)$,
with $S_{2}(p,1)=1=S_{2}(1,k)$ for all $p,k\in\mathbb{N}^{*}$, where the $S_2(p,k)$ are the Stirling numbers of the second kind. Observe
that, for $M/2\geq k$, and for any $\tilde{c}\in\mathcal{L}(\mathcal{S}_{\pm}^{k}\mathcal{Z}^{\otimes k})$,
\[
\|\tilde{c}^{Wick}(1+\mathbf{N}_{\pm})^{-M/2}\|_{\mathcal{L}(\Gamma_{\pm}(\mathcal{Z}))}\leq\|\tilde{c}\|_{\mathcal{L}(\mathcal{S}_{\pm}^{k}\mathcal{Z}^{\otimes k};\mathcal{S}_{\pm}^{k}\mathcal{Z}^{\otimes k})}\,.
\]
We thus get, 
\[
\|(1+\mathbf{N}_{\pm})^{-m/2}\Big(d\Gamma_{\pm}(\tilde{b})^{p}-(\tilde{b}^{\otimes p})^{Wick}\Big)(1+\mathbf{N}_{\pm})^{-m'/2}\|_{\mathcal{L}(\Gamma_{\pm}(\mathcal{Z}))}\leq\sum_{k=1}^{p-1}\varepsilon^{p-k}S_{2}(p,k)\:\|\tilde{b}\|_{\mathcal{L}(\mathcal{Z})}^{p}
\]
 and the estimate then follows from
$
\sum_{k=1}^{p-1}\varepsilon^{p-k}S_{2}(p,k)
\leq\varepsilon\sum_{k=1}^{p}S_{2}(p,k)
=\varepsilon B_{p}
$ 
with $B_{p}$ the $p$-th Bell number.
\end{proof}

\subsection{Reduced density matrices}
\label{sec:reduc-dens-matr}
Reduced density matrices emerge naturally in the study of correlation functions of quantum gases 
and in particular in the quantum mean-field theory they are the main quantities to be analysed.  
We shall work with \underline{non normalized} reduced density matrices, which are
easier to handle. Going back to the more natural reduced density
matrices with trace equal to $1$\,, requires attention when
normalizing and taking the limits.
\begin{definition}
\label{de:gammap}
  Let $\varrho_{\varepsilon}\in {\cal L}^{1}(\Gamma_{\pm}(\Z))$
  ($\varepsilon>0$ is fixed here) be such that
  $\varrho_\varepsilon\geq 0$\,, $\Tr\left[\varrho_{\varepsilon}\right]=1$ and $\Tr\left(\varrho_{\varepsilon}
  e^{c\mathbf{N}_{\pm}}\right)<\infty$ for some $c>0$\,. The  non normalized
reduced density
matrix of order $p\in \nz$\,, $\gamma^{(p)}_{\varepsilon}\in
\mathcal{L}^{1}(\mathcal{S}^{p}_{\pm}\Z^{\otimes p})$\,, 
 is defined by duality according to
$$
\forall \tilde{b}\in {\cal L}({\cal S}_{\pm}^{p}\Z^{\otimes p};{\cal
  S}_{\pm}^{p}\Z^{\otimes p})\,,\quad
\Tr\left[\gamma_{\varepsilon}^{(p)}\tilde{b}\right]
=\Tr\left[\varrho_{\varepsilon} \tilde{b}^{Wick}\right]\,.
$$
\end{definition} 
\noindent The definition makes sense owing to  the number estimate
(\ref{eq:numbEst}) and to
$(1+\mathbf{N}_{\pm})^{k}e^{-c\mathbf{N}_{\pm}}\in {\cal
  L}(\Gamma_{\pm}(\Z))$\,. 
When $\Tr\left[\gamma_{\varepsilon}^{(p)}\right]\neq 0$\,, the
normalized density matrix $\bar\gamma^{(p)}_{\varepsilon}$ is defined by 
$\bar\gamma^{(p)}_{\varepsilon}=\frac{\gamma^{(p)}_{\varepsilon}}{\Tr\left[\gamma^{(p)}_{\varepsilon}\right]}$\,, 
that is
\begin{align*}
\forall \tilde{b}\in {\cal L}({\cal S}_{\pm}^{p}\Z^{\otimes p}),\quad
\Tr\left[\bar\gamma^{(p)}_{\varepsilon}\tilde{b}\right]
& =\frac{\Tr\left[\varrho_{\varepsilon}\tilde{b}^{Wick}\right]}{\Tr\left[\varrho_{\varepsilon}(\Id_{{\cal
        S}_{\pm}^{p}\Z^{\otimes p}})^{Wick}\right]} \\
        &=\frac{\Tr\left[\varrho_{\varepsilon}\tilde{b}^{Wick}\right]}{\Tr\left[\varrho_{\varepsilon}N_{\pm}(N_{\pm}-\varepsilon)\ldots
  (N_{\pm}-\varepsilon (p-1))\right]} \,.
\end{align*}
These normalized reduced density matrices $\bar\gamma^{(p)}_{\varepsilon}$ are
commonly used, especially when $\varrho_\varepsilon\in \mathcal L^1(\mathcal S_\pm \Z^{\otimes n} )$, with $n\varepsilon\sim 1$ (see \cite{BGM,BEGMY,BPS,KnPi,LNR}), 
for the following reason:
When
$\varrho_{\varepsilon}\in {\cal L}^{1}({\cal S}_{\pm}^{n}\Z^{\otimes
  n})\subset {\cal L}^{1}(\Z^{\otimes n})$ lies in
the $n$-particles sector (with $n\varepsilon\to 1$ for the mean-field
regime) it is simply given by the partial trace of
$\varrho_{\varepsilon}$ when $n>p$\,.
Actually
$$
\tilde{b}^{Wick}\big|_{{\cal S}_{\pm}^{n}\Z^{\otimes
    n}}=\varepsilon^{p}\frac{n!}{(n-p)!}{\cal
  S}_{\pm}^{n}(\tilde{b}\otimes \Id_{\Z^{\otimes(n-p)}}){\cal S}_{\pm}^{n,*}
$$
and, as $\varepsilon^p \, n(n-1)\cdots(n-p+1)\to 1$ if $n\varepsilon \to 1$, it follows that
$$
\lim_{\substack{n\varepsilon\sim 1\\ \varepsilon\to 0}}\Tr\left[\gamma^{(p)}_{\varepsilon}\tilde{b}\right]
        =\lim_{\substack{n\varepsilon\sim 1\\ \varepsilon\to 0}} \Tr\left[\varrho_{\varepsilon}(\tilde{b}\otimes
  \Id_{\Z^{\otimes (n-p)})})\right]\,.
$$
When $\Z=L^{2}(M;dv)$ one thus often considers:
$$
\tilde{\gamma}_{\varepsilon}^{(p)}(x_{1},\ldots,x_{p};x_{1}',\ldots,x_{p}')=
\int_{M^{n-p}}\varrho_{\varepsilon}(x_{1},\ldots,x_{p},x;x_{1}',\ldots,x_{p}',
x)~dv^{\otimes (n-p)}(x)\,.
$$ 
But, if the states $\varrho_\varepsilon$ are not localized on the $n$-particles sector, such an alternative definition does not coincide with $\gamma_\varepsilon^{(p)}$,
even asymptotically in the mean-field regime.

As well there is no general relation between the non normalized density matrices
$\gamma^{(p+1)}_{\varepsilon}$ and $\gamma^{(p)}_{\varepsilon}$\,.
Actually we have 
\begin{align*}
\left({\cal S}_{\pm}^{p+1}(\tilde{b}\otimes \Id_{\Z}){\cal
   S}_{\pm}^{p+1,*}\right)^{Wick}\big|_{{\cal
   S}_{\pm}^{n+p+1}\Z}
&=\varepsilon^{p+1}\frac{(n+p+1)!}{n!}
{\cal S}_{\pm}^{n+p+1}(\tilde{b}\otimes \Id_{\Z^{\otimes n+1}}){\cal
   S}_{\pm}^{n+p+1,*}\\
&=\varepsilon(n+1)\tilde{b}^{Wick}\big|_{{\cal S}_{\pm}^{n}\Z^{n+p+1}}\,,
\end{align*}
from which we deduce
\begin{eqnarray*}
  &&\Tr\left[\gamma_{\varepsilon}^{(p+1)}(\tilde{b}\otimes \Id_{\Z})\right]
=\Tr\left[\varrho_{\varepsilon}\N_{\pm}\tilde{b}^{Wick}\right]+{\cal
     O}(\varepsilon)\,,\\
\text{while}&&
\Tr\left[\gamma_\varepsilon^{(p)}\tilde{b}\right]=\Tr\left[\varrho_{\varepsilon}\tilde{b}^{Wick}\right]\,,
\end{eqnarray*}
where we have again identified $\gamma^{(p+1)}_{\varepsilon}$ as an
element of ${\cal L}^{1}(\Z^{\otimes (p+1)})$\,.
We thus conclude with the following important remark.
\begin{remark}
Simple asymptotic relation between the $\gamma^{(p)}_{\varepsilon}$ and
$\gamma^{(p')}_{\varepsilon}$ (or the normalized version
$\bar\gamma^{(p)}_{\varepsilon}$ and $\bar\gamma^{(p')}_{\varepsilon}$) 
can be expected when
$\varrho_{\varepsilon}=\varrho_{\varepsilon}
1_{[\nu-\delta(\varepsilon),\nu+\delta(\varepsilon)]}(\mathbf{N}_{\pm})$
with $\nu>0$ and $\lim_{\varepsilon\to 0}\delta(\varepsilon)=0$ but
not otherwise (of course the condition above is sufficient but not necessary). 
\end{remark}
We shall use recurrently with variations the following lemma, with
the following notations
$$
\tilde{b}_{1}\odot\cdots\odot \tilde{b}_{p}
=
\frac{1}{p!}\sum_{\sigma\in
  \mathfrak{S}_{p}}\tilde{b}_{\sigma(1)}\otimes\cdots\otimes \tilde{b}_{\sigma(p)}\,,
$$  
for $\tilde{b}_{1},\ldots,\tilde{b}_{p}\in \mathcal{L}(\Z)$\,.\\
We also write shortly $(\tilde{b}_{1}\odot\ldots\odot
\tilde{b}_{p})^{Wick}$ and $(\tilde{b}^{\otimes p})^{Wick}$ instead of 
$\left({\cal S}_{\pm}^{p}(\tilde{b}_{1}\odot\ldots\odot
\tilde{b}_{p}){\cal S}_{\pm}^{p,*}\right)^{Wick}$ and $\left({\cal
S}_{\pm}^{p}
(\tilde{b}^{\otimes p}){\cal S}_{\pm}^{p,*}\right)^{Wick}$.
\begin{lem}
\label{le:symm}
  \textbf{Quantum symmetrization lemma:}
In the bosonic and fermionic cases for any $p\in \nz$\,, the equality
\begin{equation}\label{eq:symm_lemma}
\mathcal{S}^{p}_{\pm}(\tilde{b}_{1}\otimes\cdots\otimes \tilde{b}_{p})\mathcal{S}^{p,*}_{\pm}=
\mathcal{S}^{p}_{\pm}(\tilde{b}_{1}\odot\cdots\odot
\tilde{b}_{\sigma(p)})\mathcal{S}^{p,*}_{\pm}\,,
\end{equation}
holds in $\mathcal{L}(\mathcal{S}^{p}_{\pm}\Z^{\otimes
  p};\mathcal{S}^{p}_{\pm}\Z^{\otimes
  p})$
for all $\tilde{b}_{1},\ldots,\tilde{b}_{p}\in \mathcal{L}(\Z;\Z)$\,.

As a consequence, under the assumptions of Definition~\ref{de:gammap},
the
non normalized (resp. normalized if possible) reduced density
matrix $\gamma^{(p)}_{\varepsilon}$
(resp. $\bar\gamma^{(p)}_{\varepsilon}$)\,, 
$p\in\nz$\,, is
 completely determined by the set of quantities $\left\{\Tr\left[\varrho_\varepsilon
     (\tilde{b}^{\otimes p})^{Wick}\right]\,,
   \tilde{b}\in\mathcal{B}\right\}$ when $\mathcal{B}$
is any dense subset of ${\cal L}^{\infty}(\Z;\Z)$\,.
 \end{lem}
 \begin{remark}
   While computing $\Tr\left[\gamma^{(p)}_{\varepsilon}\right]$ or
   studying $\bar\gamma^{(p)}_{\varepsilon}$ one can simply add to $\mathcal{B}$ the element $\Id_{\Z}$ owing to 
${\cal S}_{\pm}^{p}\Id_{\Z}^{\otimes p}{\cal S}_{\pm}^{p,*}=\Id_{{\cal
       S}_{\pm}^{p}\Z^{\otimes p}}$\,. 
For $\varepsilon>0$ fixed it is not necessary because compact observables are sufficient to determine the total trace owing to 
$\Tr\left[\gamma^{(p)}_{\varepsilon}\right]=\sup_{B\in {\cal
    L}^{\infty}(\mathcal S_\pm^p\Z^{\otimes p}),\, 0\leq B \leq \Id} 
\Tr\left[\gamma_{\varepsilon}^{(p)}B\right]$\,. 
However, while
considering weak$^*$-limits as $\varepsilon\to 0$\,, adding the identity operator 
$\Id_{{\cal S}_{\pm}^{p}\Z^{\otimes p}}$ to the set of compact observables, or possibly replacing $\mathcal B$ by the Calkin algebra $\cz\Id(\Z)\oplus \mathcal L^\infty(\Z)$, is useful in order to control the asymptotic total mass.
\end{remark}
\begin{proof}
For $\tilde{b}_{1},\ldots, \tilde{b}_{p}\in \mathcal{L}(\Z)$\,, we decompose $\mathcal{S}^{p}_{\pm}\left(\tilde{b}_{1}\otimes
  \cdots\otimes\tilde{b}_{p}\right)\mathcal{S}^{p,*}_{\pm}
\mathcal{S}^{p}_{\pm}(\psi_{1}\otimes\cdots\otimes \psi_{p})$ as 
\begin{align*}
\mathcal{S}^{p}_{\pm}&\left[\frac{1}{p!}\sum_{\sigma'\in
    \mathfrak{S}_{p}}s_{\pm}(\sigma')(\tilde{b}_{1}\psi_{\sigma'(1)})\otimes
\cdots\otimes (\tilde{b}_{p}\psi_{\sigma'(p)})\right]
\\
&\qquad
=\frac{1}{p!p!}\left[
\sum_{\sigma\in \mathfrak{S}_{p}}\sum_{\sigma'\in \mathfrak{S}_{p}}
s_{\pm}(\sigma)s_{\pm}(\sigma') 
(\tilde{b}_{\sigma(1)}\psi_{\sigma\circ\sigma'(1)})
\otimes \cdots\otimes
(\tilde{b}_{\sigma(p)}\psi_{\sigma\circ\sigma'(p)})
\right]\,.
\end{align*}
Setting $\sigma''=\sigma\circ\sigma'$\,, with
$s_{\pm}(\sigma'')=s_{\pm}(\sigma)s_{\pm}(\sigma')$ yields the Eq.~\eqref{eq:symm_lemma},
after noting that
$\tilde{b}_{1}\odot \cdots\odot \tilde{b}_{p}=\frac{1}{p!}\sum_{\sigma\in\mathfrak{S}_{p}}\tilde{b}_{\sigma(1)}\otimes
\cdots\otimes \tilde{b}_{\sigma(p)}$ commutes with
$\mathcal{S}^{p}_{\pm}$ in both the bosonic case and the fermionic case.\\
Now the non normalized reduced density matrix is determined by
$$
\Tr\left[\gamma^{(p)}_{\varepsilon}\tilde{B}\right]
=\Tr\left[\varrho_\varepsilon \tilde{B}^{Wick}\right]
$$
for $\tilde{B}\in
\mathcal{L}^{\infty}(\mathcal{S}^{p}_{\pm}\Z^{\otimes p})$ as
$\mathcal{L}^{1}(\mathcal{S}^{p}_{\pm}\Z^{\otimes p})$ is the dual of
$\mathcal{L}^{\infty}(\mathcal{S}^{p}_{\pm}\Z^{\otimes p})$\,.
But $\tilde{B}\in {\cal L}^{\infty}({\cal S}_{\pm}^{p}\Z^{\otimes p})$
means $\tilde{B}={\cal S}_{\pm}^{p}\tilde{B}'{\cal S}_{\pm}^{p,*}$
with $\tilde{B}'\in {\cal L}^{\infty}(\Z^{\otimes p})$\,, while the
algebraic tensor product ${\cal L}^{\infty}(\Z)^{\otimes^{alg}p}$ is dense
in ${\cal L}^{\infty}(\Z^{\otimes p})$\,.\\
With the estimate 
\begin{align*}
\left|\Tr\left[\varrho_\varepsilon \tilde{B}^{Wick}\right]\right| 
& =
\left|\Tr\left[e^{c/2\mathbf{N}}\varrho_\varepsilon
      e^{c/2\mathbf{N}}e^{-c/2\mathbf{N}}\tilde{B}^{Wick}e^{-c/2\mathbf{N}}\right]\right|
      \\
& \leq C\Tr\left[\varrho_\varepsilon
  e^{c\mathbf{N}}\right]\|\tilde{B}\|_{\mathcal{L}(\mathcal{S}^{p}_{\pm}\Z^{\otimes
  p}; \mathcal{S}^{p}_{\pm}\Z^{\otimes p})}\,,
\end{align*}
it suffices to consider $\tilde{B}={\cal S}_{\pm}^{p}\tilde{B}'{\cal
  S}_{\pm}^{p,*}$ with $\tilde{B}'\in {\cal
  L}^{\infty}(\Z)^{\otimes^{alg}  p}$. By linearity and density, $\gamma_\varepsilon^{(p)}$ is determined by the quantities $\Tr[\varrho_\varepsilon \tilde B^{Wick}]$ with 
$\tilde{B}'=\tilde{b}_{1}\otimes \cdots\otimes \tilde{b}_{p}$\,,
$\tilde{b}_{i}\in {\cal B}$\,.
We conclude with
$$
\mathcal{S}^{p}_{\pm}\left(\tilde{b}_{1}\otimes \cdots
  \otimes
\tilde{b}_{p}\right)\mathcal{S}^{p,*}_{\pm}
=\mathcal{S}^{p}_{\pm}\left(\tilde{b}_{1}\odot\cdots\odot \tilde{b}_{p}\right)\mathcal{S}^{p,*}_{\pm}\,,
$$
and the polarization identity
$$
\tilde{b}_{1}\odot\cdots\odot\tilde{b}_{p}=\frac{1}{2^{p}p!}\sum_{\varepsilon_{i}=\pm 1}
\varepsilon_{1}\cdots\varepsilon_{p}(\sum_{i=1}^{p}\varepsilon_{i}\tilde{b_{i}})^{\otimes
p}\,.
$$
\end{proof}
\begin{remark}
In the bosonic case, the non normalized reduced density matrices $\gamma^{(p)}_\varepsilon$ are also characterized by the values of $\Tr[\gamma^{(p)}_\varepsilon B]$ for $B$ in
  $\mathcal{B}
=\left\{|\psi^{\otimes p}\rangle \langle \psi^{\otimes p}|\,, \psi\in
  \Z\right\}\cup\left\{\Id_{\Z}^{\otimes p}\right\}$\,. This does not hold in the fermionic case. 
\end{remark}
The rest of the article is devoted to the asymptotic analysis of
$\gamma^{(p)}_{\varepsilon}$ as $\varepsilon\to 0$\,.
In particular we shall study their concentration at the quantum level
while testing with fixed observable $\tilde{b}$  (with $\tilde{b}$
compact) and their semiclassical behaviour after taking
semiclassically quantized observables,
e.g. $a(x,hD_{x})$ 
with some relation $\varepsilon=\varepsilon(h)$ between $\varepsilon$
and $h$\,.

\section{Classical phase-space and $h$-quantizations}
\label{sec:classic}

When $\Z=L^{2}(M^{1},dx)$\,, with $M^{1}=M$
 a manifold with volume measure
 $dx$\,,
 the classical one particle phase-space is
${\cal X}^{1}={\cal X}=T^{*}M^{1}$ and we will focus on the $h$-dependent quantizations which
associates with a symbol $a(x,\xi)=a(X)$\,, $X\in {\cal X}^{1}$ an operator
 $a^{Q,h}=a(x,hD_{x})$ with the standard semiclassical quantization
or when $M^{1}=\rz^{d}$\,, $a^{Q,h}=a^{W,h}=a^{W}(h^{t}x,h^{1-t}D_{x})$
by using the Weyl quantization, $t\in \rz$ being fixed.

Note that in later sections the parameters $\varepsilon$ and $h$ will be linked through $\varepsilon=\varepsilon(h)$ with $\lim_{h\to
  0}\varepsilon(h)=0$\,.
In relation with the symmetrization Lemma~\ref{le:symm}, we introduce
the adapted $p$-particles phase-space 
 which was also considered in
\cite{Der1}, and the corresponding semiclassical observables.

\subsection{Classical $p$-particles phase-space}
\label{sec:classp}

A fundamental principle of quantum mechanics is that identical
particles are indistinguishable. The classical description is thus
concerned with indistiguishable classical particles. If one classical particle is characterized by
its position-momentum $(x,\xi)\in \mathcal{X}^{1}=T^{*}M^{1}$\,, $x\in M$ being
the position coordinate and $\xi$ the momentum coordinates,
 $p$ indistinguishable particles will be
described by their position-momentum coordinates
$(X_{1},\ldots ,X_{p})=(x_{1},\xi_{1},\ldots,x_{p},\xi_{p})\in
\mathcal{X}^{p}/\mathfrak{S}_{p}=(T^{*}M)^{p}/\mathfrak{S}_{p}=T^{*}(M^{p})/\mathfrak{S}_{p}$\,,
where the quotient by $\mathfrak{S}_{p}$ simply implements the
identification
$$
\forall \sigma\in \mathfrak{S}_{p}\,,\quad
(X_{\sigma(1)},\ldots, X_{\sigma(p)})
\equiv(X_{1},\ldots,X_{p})\,.
$$
The grand canonical description of a classical particles system then
takes place in the disjoint union 
$$
\mathop{\sqcup}_{p\in\nz}\mathcal{X}^{p}/\mathfrak{S}_{p}=
\mathop{\sqcup}_{p\in\nz}(T^{*}M)^{p}/\mathfrak{S}_{p}\,.
$$ 
A $p$-particles classical observable will be a function on
$\mathcal{X}^{p}/\mathfrak{S}_{p}$ and when the number of particles is
not fixed a collection of functions $(a^{(p)})_{p\in\nz}$ each
$a^{(p)}$ being a function on $\mathcal{X}^{p}/\mathfrak{S}_{p}$\,.
The situation is presented in this way in \cite{Der1}.
A $p$-particles
observable is a function $a^{(p)}$ on ${\cal X}^{p}/\mathfrak{S}_{p}$
and a $p$-particles classical state is a probability measure (and when
the normalization is forgotten a non negative measure) on ${\cal X}^{p}/\mathfrak{S}_{p}$\,.\\ 
However while quantizing a classical observable, it is better to work in
$\mathcal{X}^{p}$ which equals $T^{*}(M^{p})$\,,
 a function $a^{(p)}$ on
$\mathcal{X}^{p}/\mathfrak{S}_{p}$ being nothing but a function on
$\mathcal{X}^{p}$ which satisfies
\begin{eqnarray*}
&&
  \forall \sigma\in \mathfrak{S}_{p}\,,\quad
\sigma^{*}a^{(p)}=a^{(p)}=\frac{1}{p!}\sum_{\tilde\sigma\in \mathfrak{S}_{p}}\tilde\sigma^{*}a^{(p)}\,,\\
\text{with}
&&
\forall X_{1}\ldots X_{p}\in{\cal X}^{p}\,, \quad
   \sigma^{*}a^{(p)}(X_{1},\ldots,X_{p})
=a^{(p)}(X_{\sigma(1)},\ldots, X_{\sigma(p)})\,.
\end{eqnarray*}
In the same way, we define for a Borel measure $\nu$ on ${\cal
  X}^{p}$ and $\sigma \in \mathfrak{S}_{p}$\,,
the measure  $\sigma_{*}\nu$ by
$\int_{X}\sigma^{*}a^{(p)}d\nu=\int_{{\cal X}^{p}}a^{(p)}\, d(\sigma_{*}\nu)$
for all $a^{(p)}\in {\cal C}^{0}_{c}({\cal X}^{p})$\,, or
alternatively
$\sigma_{*}\nu(E)=\nu(\sigma^{-1}E)$ for all Borel subset $E$ of
${\cal X}^{p}$\,.  
A non-negative measure on
${\cal X}^{p}/\mathfrak{S}_{p}$ is identified with a non-negative
measure $\nu$ on ${\cal X}^{p}$ such that
\begin{equation}
  \label{eq:invsigmanu}
\forall \sigma\in \mathfrak{S}_{p}\,,\quad
\sigma_{*}\nu=\nu=\frac{1}{p!}\sum_{\tilde\sigma\in \mathfrak{S}_{p}}\tilde\sigma_{*}\nu\,.
\end{equation}
\begin{lem}[Classical symmetrization lemma]
\label{le:clsymlem}
Any Borel measure $\mu^{(p)}$ on ${\cal X}^{p}/\mathfrak{S}_{p}$ is characterized
by the quantities $\left\{\int_{{\cal X}^{p}}a^{\otimes p}d\mu^{(p)}\,,\quad
a\in{\cal C}\right\}$ where the tensor power $a^{\otimes p}$ means 
$a^{\otimes
p}(X_{1},\ldots, X_{p})=\prod_{i=1}^{p}a(X_{i})$ and ${\cal C}$ is any
dense set  in ${\cal C}^{0}_{\infty}({\cal X}^{1})=\left\{f\in {\cal
    C}^{0}({\cal X}^{1})\,, \lim_{X\to \infty}f(X)=0\right\}$\,.
\end{lem}
\begin{proof}
By Stone-Weierstrass Theorem the subalgebra generated by the  algebraic tensor product ${\cal C}^{\otimes^{alg}p}$ is dense in ${\cal
    C}^{0}_{\infty}({\cal X}^{p})$\,. Hence it suffices to consider
$$
a_{1}\odot\cdots\odot a_{p}=\frac{1}{p!}\sum_{\sigma\in
  \mathfrak{S}_{p}}
a_{\sigma(1)}\otimes \cdots\otimes a_{\sigma(p)}\,,\quad a_{i}\in 
{\cal C}\,.
$$
We conclude again with the polarization identity
$$
a_{1}\odot\cdots\odot a_{p}=\frac{1}{2^{p}p!}\sum_{\varepsilon_{i}=\pm 1}
\varepsilon_{1}\ldots\varepsilon_{p}\left(\sum_{i=1}^{p}\varepsilon_{i}a_{i}\right)^{\otimes
p}\,.
$$
\end{proof}
We will work essentially with $M=\rz^{d}$ and
$\mathcal{X}=T^{*}\rz^{d}$ and therefore on
$\mathcal{X}^{p}=T^{*}\rz^{dp}\sim\rz^{2dp}$ and recall the invariance properties,
if possible by a change of variable in order to extend it to the
general case. Remember that
on   $\rz^{dp}$, the standard and Weyl semiclassical
quantization are asymptotically equivalent
$a(x,hD_{x})-a^{W}(x,hD_{x})=O(h)$ when $a\in S(1,dX^{2})$ 
($\sup_{X\in T^{*}\rz^{dp}}|\partial _{X}^{\alpha}a(X)|<\infty$ for
all $\alpha\in\nz^{2d}$)\,. Moreover on $\rz^{dp}$\,,
$a^{W}(x,hD_{x})$ is unitary equivalent to $a^{W}(h^{t}x,
h^{1-t}D_{x})$ for any fixed $t\in \rz$ so that result can be adapted to
different scalings.
\subsection{Semiclassical and multiscale measures}
\label{sec:semmultmeas}
We recall the notions
 of semiclassical (or Wigner) measures and multiscale measures
in the finite dimensional case. We start with the results 
 on $M=\rz^{D}$ (think of
$D=d\,p$) and review the invariance properties for 
applications
to some more general manifolds  $M$\,.
\subsubsection{In the Euclidean Space}
On $\rz^{D}$ the semiclassical Weyl quantization of a symbol $a\in {\cal
  S}'(\rz^{2D})$ will be written $a^{W,h}=a^{W}(h^{t}x,h^{1-t}D_{x})$
with $t>0$ fixed and a kernel given by
$$
[a^{W}(x,D_{x})](x,y)=\int_{\rz^{d}}e^{i\xi\cdot(x-y)}a(\frac{x+y}{2},\xi)~\frac{d\xi}{(2\pi)^{d}}\,.
$$ 
\begin{definition}
\label{de:adapted}
Let $(\gamma_{h})_{h\in {\cal E}}$ with $0\in \overline{{\cal E}}$\,,
${\cal E}\subset (0,+\infty)$\,, be a family of trace-class non-negative
operators on $L^{2}(\rz^{D})$ such that
 $\lim_{h\to 0}\Tr\left[\gamma_{h}\right]<+\infty$\,.
 The semiclassical quantization $a\mapsto
 a^{W,h}=a^{W}(h^{t}x,h^{1-t}D_{x})$  is said \emph{adapted} to the
 family $(\gamma_{h})_{h\in {\cal E}}$ if 
$$
\lim_{\delta\to 0^{+}}\limsup_{h\in {\cal E},\, h\to
  0}\Real \Tr\left[(1-\chi(\delta\,\cdot )^{W,h})\gamma_{h}\right]=0
$$
for some $\chi\in {\cal C}^{\infty}_{0}(T^{*}\rz^{D})$ such that
$\chi\equiv 1$ in a neighborhood of $0$\,.

The set of \emph{Wigner measures} ${\cal M}(\gamma_{h}, h\in {\cal E})$ is
the set of non-negative measures $\nu$ on $T^{*}\rz^{D}$ such that
there exists ${\cal E}'\subset {\cal E}$\,, $0\in \overline{{\cal
    E}'}$ such that
$$
\forall a\in {\cal C}^{\infty}_{0}(T^{*}\rz^{D})\,, \quad
\lim_{\substack{h\in {\cal E'}\\h\to 0}}\Tr\left[\gamma_{h}a^{W,h}\right]=\int_{T^{*}\rz^{D}}a(X)~d\nu(X)\,.
$$
\end{definition}
The following well known statement (see
\cite{CdV,HMR,Ger,GMMP,LiPa,Sch})  results from the asymptotic
positivity of the semiclassical quantization and it is actually the
finite dimensional version of bosonic mean-field Wigner measures (with
the change of parameter $\varepsilon=2h$) (see \cite{AmNi1}--Section~3.1).
\begin{prop}
\label{pr:scmeas}
Let $(\gamma_{h})_{h\in {\cal E}}$ with $0\in \overline{{\cal E}}$\,,
${\cal E}\subset (0,+\infty)$\,, such that $\gamma_h\geq 0$ and
 $\lim_{h\to 0}\Tr\left[\gamma_{h}\right]<+\infty$\,. The set of semiclassical
measures ${\cal M}(\gamma_{h}, h\in {\cal E})$ is non-empty. 
The semiclassical quantization $a^{W,h}$ is adapted to the family $(\gamma_{h})_{h\in {\cal E}}$\,, 
if, and only if, any $\nu\in {\cal   M}(\gamma_{h}, h\in {\cal E})$ satisfies
  $\nu(\rz^{2D})=\lim_{h\to 0}\Tr\left[\gamma_{h}\right]$\,.
\end{prop}
\begin{remark}
  The manifold version, with $a^{Q,h}=a(x,hD_{x})$ instead of $a^{W,h}$\,, 
 results from the semiclassical Egorov theorem.
 
By reducing ${\cal E}$ to some subset ${\cal E}'$ (think of subsequence
extraction), one can always assume that there is a unique semiclassical
measure. While considering a time evolution problem, or adding another
non  countable parameter, $(\gamma_{t,h})_{h\in {\cal E}, t\in\rz}$
finding simultaneously the subset ${\cal E}'$ for all $t\in \rz$
requires some compactness argument w.r.t the parameter $t\in\rz$\,,
usually obtained by equicontinuity properties.
\end{remark}

We now review the multiscale measures introduced in \cite{FeGe,Fer}. For reader's convenience, details are given in Appendix~\ref{sec:multiscale-measures}, concerning the relationship between Prop.~\ref{pr:multisc_simplified} below and the more general statement of \cite{Fer}.

The class of symbols $S^{(2)}$ is defined as the set of $a \in {\cal C}^{\infty}(\rz^{2D}\times \rz^{2D})$\,,  
such that 
\begin{itemize}
\item there exists $C>0$ such that $\forall Y\in \rz^{2D}$\,,
$a(\cdot,Y)\in {\cal C}^{\infty}_{0}(B(0,C))$\,;
\item there exists  a function 
$a_{\infty}\in {\cal C}^{\infty}_{0}(\rz^{2D}\times \sz^{2D-1})$ 
such that
$a(X,R\omega)\stackrel{R\to \infty}{\to}a_{\infty}(X,\omega)$ in
${\cal C}^{\infty}(\rz^{2D}\times \sz^{2D-1})$\,.
\end{itemize}
Those symbols are quantized according to 
$$
a^{(2),h}=a_{h}^{W,h}\quad, \quad
a_{h}(X)=a(X,\frac{X}{h^{1/2}})\,.
$$
A geometrical interpretation of those double scale symbols can be given by matching the compactified quantum phase space with the blow-up at $r=0$ of the macroscopic phase space, see Figure~\ref{figblow}.

\begin{figure}[h] 
\centering
\includegraphics[scale=0.7]{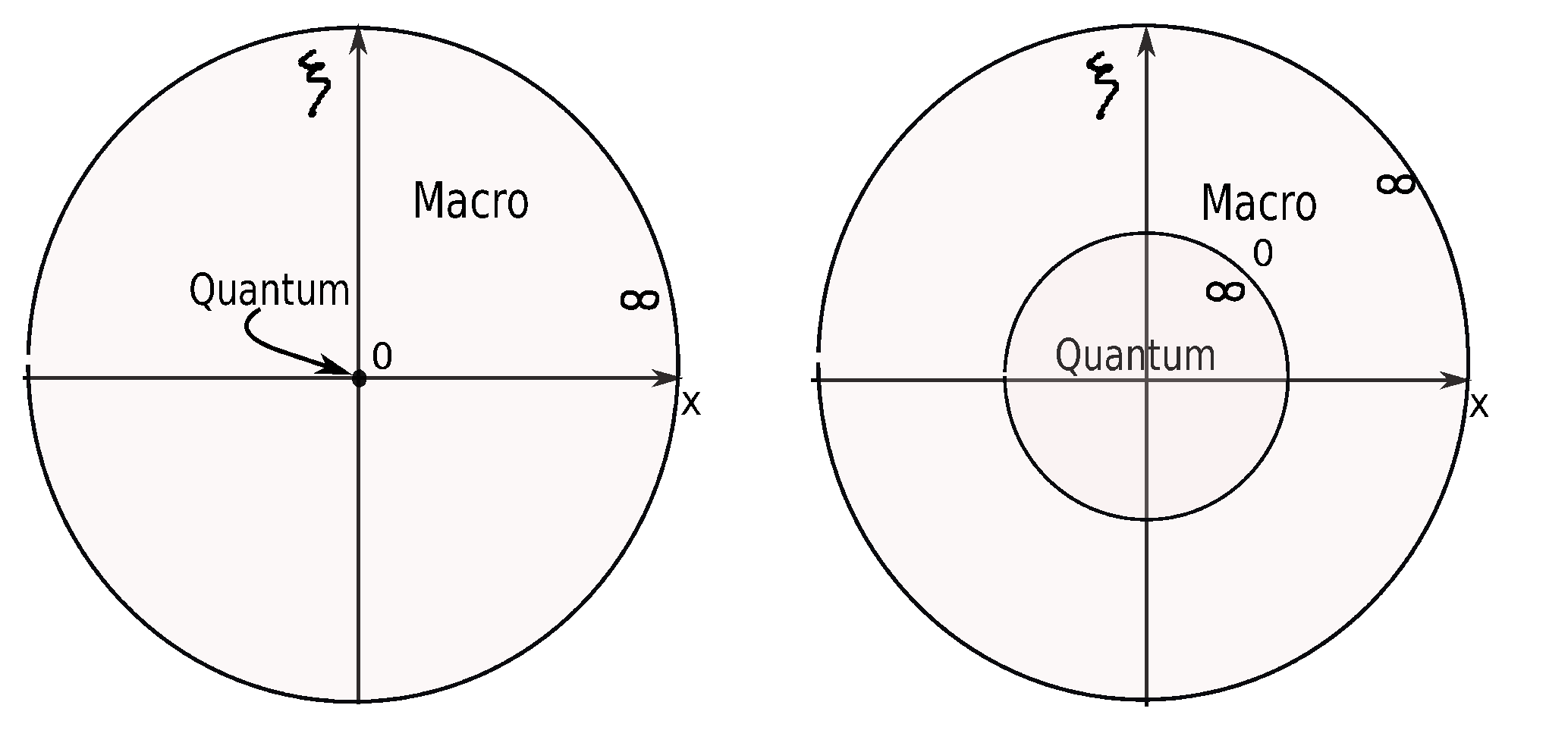}
\caption{On the left-hand side, the macroscopic phase space with its sphere at infinity. On the right-hand side the matched quantum and macroscopic phase spaces for which the quantum sphere at infinity and the $r=0$ macroscopic sphere coincide.}\label{figblow}
\end{figure}

\begin{prop}
\label{pr:multisc_simplified}
Let $(\gamma_{h})_{h\in {\cal E}}$ be a bounded family of non-negative
trace-class operators on $L^{2}(\rz^{D})$ with
$\lim_{h\to 0}\Tr\left[\gamma_{h}\right]<+\infty$\,. There exist ${\cal E}'\subset {\cal E}$\,, $0\in \overline{{\cal E}'}$\,, non-negative measures
$\nu$ and $\nu_{(I)}$ 
on $\mathbb R^{2D}$ and $\sz^{2D-1}$, 
and a $\gamma_0 \in {\cal L}^{1}(L^{2}(\rz^{D}))$\,,  such that 
${\cal M}(\gamma_{h}, h\in {\cal E}') = \left\{\nu\right\}$ and, for all $a\in S^{(2)}$,
\begin{multline*}
\lim_{\substack{h\in {\cal E}'\\ h\to 0}}\Tr\left[\gamma_{h}a^{(2),h}\right]
=\int_{\rz^{2D}\setminus \{0\}}
a_{\infty}(X,\frac{X}{|X|})~d\nu(X)
\\
+\int_{\sz^{2D-1}}a_{\infty}(0,\omega)~d\nu_{(I)}(\omega)
+\Tr\left[a(0,x,D_{x})\gamma_0\right]\,.
\end{multline*}

\end{prop}
\begin{definition}
  \label{de:M2_simplified} 
${\cal M}^{(2)}(\gamma_{h}, h\in {\cal E})$ denotes the set of all triples $(\nu,\nu_{(I)},\gamma_{0})$  
which can be obtained in Prop.~\ref{pr:multisc_simplified} for suitable choices of ${\cal E}'\subset {\cal E}$\,, $0\in \overline{{\cal E}'}$\,.
\end{definition}

\begin{remark}\label{gamma0-is-the-weak-limit-of-gammah}
Actually when
$a^{W,h}=a^{W}(\sqrt{h}x,\sqrt{h}D_{x})$\,, this trace class operator 
$\gamma_{0}$ is nothing but the weak${}^{*}$-limit of
$\gamma_{h}$\,.  Take simply $ \tilde a(X,Y)=\chi(X)\alpha(Y)$ with
$\chi\,,\,\alpha\in {\cal C}^{\infty}_{0}(\rz^{2D})$\,, $\chi\equiv 1$
in a neighborhood of $0$\,, for which
$\lim_{h\to 0}\| \tilde a^{(2),h}-\alpha^{W}(x,D_{x})\|_{{\cal
    L}(L^{2})}=0$\,. 
The above results
says 
$\lim_{h\to 0}\Tr\left[\gamma_{h}\alpha^{W}(x,D_{x})\right]=\Tr\left[\gamma_{0}\alpha^{W}(x,D_{x})\right]$\, for all 
$\alpha\in {\cal C}^{\infty}_{0}(\rz^{2D})\subset
L^{2}(\rz^{2D},dX)$\,,
 and by the density of the embeddings
 ${\cal C}^{\infty}_{0}(\rz^{2D})\subset L^{2}(\rz^{2D},dx)\sim {\cal
   L}^{2}(L^{2}(\rz^{D}))\subset {\cal L}^{\infty}(L^{2}(\rz^{D}))$\,,
 the test observable $\alpha^{W}(x,D_{x})$\,, 
can be replaced by any compact operator $K\in {\cal
  L}^{\infty}(L^{2}(\rz^{D},dx))$\,.
Moreover the relationship between $\nu$ and the triple
$(1_{(0,+\infty)}(|X|)\nu, \nu_{(I)},\gamma_{0})$ can be completed in this case by
\begin{equation}\label{eqn:nu(0)=}
\nu(\left\{0 \right\})=\int_{\sz^{2D-1}}d\nu_{(I)}(\omega)
+\Tr\left[\gamma_{0}\right]\,,
\end{equation}
and $\nu_{(I)}\equiv 0$ is equivalent to
$\nu(\left\{0\right\})=\Tr\left[\gamma_{0}\right]$\,. 
\end{remark}

Because products of spheres are not spheres, handling the par
$\nu_{(I)}$ in the $p$-particles space, $D=dp$\,, 
is not straitghtforward within a tensorization procedure, see Figure~\ref{figprod}.
\begin{figure}[h] 
\centering
\includegraphics[scale=0.7]{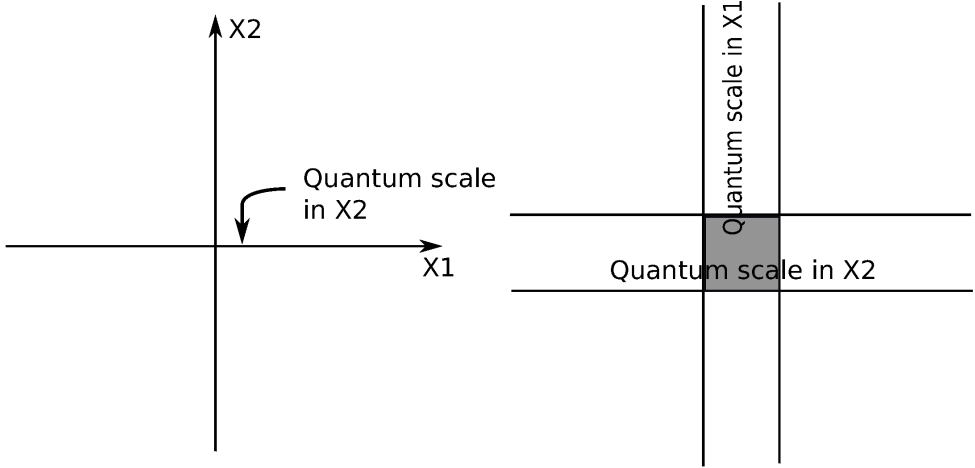}
\caption{Tensor product of two blow-ups. The product of the two matching spheres is not a sphere: the corners of the grey square correspond to the case when the quantum variables $|X_1|$ and $|X_2|$ go to infinity without any proportionality rule.}\label{figprod}
\end{figure}
 Actually we
expect in the applications that a well chosen quantization leads to
$\nu_{(I)}=0$\,. This leads to the following definition.
\begin{definition}
\label{de:sepRd}
Assume that the quantization $a^{W,h}=a^{W}(\sqrt{h}x,\sqrt{h}D_{x})$
 is adapted to the family
$(\gamma_{h})_{h\in {\cal E}}$\,, $\gamma_{h}\geq 0$\,,
$\Tr\left[\gamma_{h}\right]=1$\,. 
We say that the quantization $a^{W,h}=a^{W}(\sqrt{h}x,\sqrt{h}D_{x})$
is \emph{separating} for the family $(\gamma_{h})_{h\in {\cal E}}$ if one of the three following (equivalent) conditions is satisfied:
\begin{enumerate}
\item \label{ite:nu_I=0}For any triple $(\nu,\nu_{(I)},\gamma_{0}) \in {\cal M}^{(2)}(\gamma_{h}, h\in {\cal E})$, $\nu_{(I)} = 0$\,.
\item \label{ite:conv-faible-gamma_h}$
\left.
  \begin{array}[c]{r}
    {\cal M}(\gamma_{h}, h\in {\cal E}')=\left\{\nu\right\}\,,\\
   \displaystyle
\mathop{\mathrm{w^{*}-lim}}_{h\in {\cal E}', h\to 0}
\gamma_{h}=\gamma_{0}\quad\text{in}~{\cal L}^{1}(L^{2}(\rz^{D}))
  \end{array}
\right\}
\Rightarrow \nu(\left\{0\right\})=\Tr\left[\gamma_{0}\right]\,.
$
\item \label{ite:nu(0)=Tr[gamma0]}For any triple $(\nu,\nu_{(I)},\gamma_{0}) \in {\cal M}^{(2)}(\gamma_{h}, h\in {\cal E})$, $\nu(\left\{0\right\})=\Tr\left[\gamma_{0}\right]$\,.
\end{enumerate}
\end{definition}
\begin{remark}This terminology expresses the fact that the mass localized at any intermediate scale vanishes asymptotically when $\nu_{(I)}\equiv 0$\,. Accordingly, the microscopic quantum scale and the macroscopic scale are well identified and separated.
\end{remark} 
Hence we can get all the information by computing the weak$^{*}$-limit
of $\gamma_{h}$ and the semiclassical measure $\nu$ and then  by checking a
posteriori the equality
$\nu(\left\{0\right\})=\Tr\left[\gamma_{0}\right]$\,.\\
This will suffice when the quantum part corresponds within a
macroscopic scale, to a point in the phase-space. When $M=\rz^{d}$\,,
 we have enough flexibility by choosing the small parameter $h>0$ and
 using some dilation in $\rz^{D}$ in order to reduce many problems to
 such a case. On a manifold $M$ if we can first localize the analysis
 around a point $x_{0}\in M$\,, the problem can be transferred to
 $\rz^{D}$ and then analyzed with the suitable scaling.
 
\subsubsection{On a Compact Manifold}

We now consider another interesting case of a
compact manifold $M$ with the semiclassical calculus $a^{Q,h}=a(x,hD_{x})$. This case is not completely treated in \cite{Fer}
because the geometric invariance properties do not follow only from the
microlocal equivariance of semiclassical calculus. 
We assume $\Z=L^{2}(M,dx)$ to be defined globally on the compact manifold $M$ 
(e.g. by introducing a metric, $dx$ being the associated volume measure).
\begin{remark}
When $M$ is a general manifold, 
replace $a^{W,h}$ in Def.~\ref{de:adapted} by  $a^{Q,h}=a(x,hD_{x})$\,, 
and $\chi(\delta \, \cdot)$ with $\delta\to 0$ by some increasing sequence of comptacly supported cut-off functions $(\chi_{n})_{n\in\nz}$\,, 
such that $\cup_{n\in\nz}\chi_{n}^{-1}(\left\{1\right\})=T^{*}M$\,.
\end{remark}
To adapt Prop.~\ref{pr:multisc_simplified} to the case of a compact manifold, we consider another notion instead of the symbols $S^{(2)}$. For the observables we shall consider the pair $(K,a)$ where
 $K\in {\cal
  L}^{\infty}(L^{2}(M,dx))$ and
 $a\in {\cal C}^{\infty}_{0}(S^{*}M\sqcup (T^{*}M\setminus M))$
with $S^{*}M\sqcup (T^{*}M\setminus M)$ being described in local
coordinates through the identification
\begin{align*}
M\times[0,\infty)\times\mathbb{S}^{D-1}\ni(x,r,\omega) & \mapsto\begin{cases}
(x,\omega)\in S^{*}M & \mbox{if}\,\, r=0\,,\\
(x,\xi=r\omega)\in T^{*}M\setminus M & \mbox{otherwise}\,.
\end{cases}
\end{align*}
We have identified the $0$-section of the cotangent bundle $T^{*}M$
with $M$\,. 
After introducing an additional parameter $\delta>0$\,, $\delta\geq h$\,, 
and a ${\cal C}^{\infty}$ partition of unity 
$(1-\chi)+\chi \equiv 1$ on $T^{*}M$ with $1-\chi\in
{\cal C}^{\infty}_{0}(T^{*}M)$\,, $1-\chi\equiv 1$ in a neighborhood
of $M$\,, we can quantize $a$ as
\begin{eqnarray*}
a^{(2)\,Q,\delta,h}=[\chid(x,\xi)a(x,h\delta^{-1}\xi)]^{Q,\delta}\,.
\end{eqnarray*}
Note that $K$ and the quantization of $a$ are geometrically defined modulo ${\cal
  O}(\delta)$ when $h\leq \delta$ in ${\cal L}(L^{2}(M,dx))$: Use
local charts for the semiclassical calculus with parameter $\delta$ 
 while ${\cal L}^{\infty}(L^{2}(M,dx))$ is globally defined like
all natural spaces associated with $L^{2}(M,dx)$\,.
Actually in local
coordinates the seminorms of the symbol
$\chid(x,\xi)a(x,h\delta^{-1}\xi)$ in $S(1,dx^{2}+d\xi^{2})$
are uniformly bounded w.r.t.~$h\in (0,\delta]$ by seminorms of $a$ in
${\cal C}^{\infty}_{0}((T^{*}M\setminus M)\sqcup S^{*}M)$\,. When $a\geq 0$ one also
has
\begin{eqnarray}
\label{eq.realaa}
&&\|(\chid a(\cdot,h\delta^{-1}\cdot)^{Q,\delta}-\Real\left[(\chid a(\cdot,h\delta^{-1}\cdot)^{Q,\delta}\right]\|\leq
  C_{a}\delta\\
\label{eq.posa}
&& \|a\|_{L^{\infty}}+C_{a}\delta\geq
\Real\left[(\chid a(\cdot,h\delta^{-1}\cdot)^{Q,\delta}\right]\geq -C_{a}\delta   
\end{eqnarray}
uniformly w.r.t to $h\in(0,\delta]$\,.
\begin{prop}
\label{pr:multivar}
Let  $(\gamma_{h})_{h\in {\cal E}}$ be a family of  non-negative trace class operators on $L^{2}(M,dx)$\,, such that 
$\lim_{h\to 0}\Tr\left[\gamma_{h}\right]<+\infty$\, . 
Then there exist 
${\cal E}'\subset {\cal E}$\,, $0\in \overline{{\cal E}'}$\,, with
${\cal M}(\gamma_{h}, h\in {\cal E}')=\left\{\nu\right\}$\,, 
a non-negative  measure $\nu_{(I)}$ on $S^{*}M$ and 
a non-negative $\gamma_{0}\in {\cal L}^{1}(L^{2}(M,dx))$  such that, for any $K\in {\cal L}^{\infty}(L^{2}(M,dx))$,
$$
\lim_{\substack{h\in {\cal E}'\\ h\to 0}}
\Tr\left[\gamma_{h}K\right]
= \Tr\left[\gamma_{0}K\right]
$$
and, for any 
$a\in {\cal C}^{\infty}_{0}(S^{*}M\sqcup(T^{*}M\setminus M))$, 
and any partition of unity $(1-\chi)+\chid\equiv 1$ 
with $1-\chi\in {\cal C}^{\infty}_{0}(T^{*}M)$\,, 
$1-\chi\equiv 1$ in a neighborhood of $M$,
$$
\lim_{\delta\to 0}\lim_{\substack{h\in {\cal E}'\\ h\to 0}}
\Tr\left[\gamma_{h}\,a^{(2)\,Q,\delta,h}\right]
=\int_{T^{*}M \setminus M}a(X)~d\nu(X)+
\int_{S^{*}M}a(X)d\nu_{(I)}(X)\,.
$$
Additionally 
$(\nu_{(I)},\gamma_{0})$ is
related to $\nu$ by
$$
\nu(E)=\nu_{(I)}(\pi^{-1}(E))+\nu_{0}(E)\,,
$$
for any Borel set $E\subset M$ identified with $E\times
\left\{0\right\}$\,,  when $\pi:S^{*}M\to M$ is the
natural projection
and $\nu_{0}$ is defined by $\int_{M}\varphi(x)d\nu_{0}(x)=\Tr\left[\gamma_{0}\varphi\right]$\,, where
$\varphi\in{\cal C}^\infty(M)$ is identified with the multiplication operator by the function $\varphi$.
\end{prop}
\begin{proof}
When $\gamma_{h}$ is bounded in ${\cal L}^{1}(L^{2}(M,dx))$\,, 
after extraction of a sequence $h_{n}\to 0$ from $\cal E$\,, 
${\cal M}((\gamma_{h_n})_{n\in\mathbb N})=\left\{\nu\right\}$\,, and the weak${}^{*}$ limit $\gamma_{0}$ of $(\gamma_{h_n})$,  and the associated measure $\nu_{0}$ are  well-defined objects on the manifold $M$.

Let us construct a measure $\tilde{\nu}$ on $$(T^{*}M\setminus M)\sqcup
S^{*}M=\left\{(x,r\omega)\,, x\in M, \omega\in S^{d-1}\,, r\in
  [0,\infty)\right\}$$
and a subset ${\cal E}'\subset {\cal E}$\,, $0\in
\overline{{\cal E}'}$, such that
\begin{equation}
\label{eq:mestnu}
\lim_{\delta\to 0}\lim_{\substack{h\in {\cal E}'\\h\to 0}} 
\Tr\left[\gamma_{h}(\chid a(\cdot,h\delta^{-1}\cdot))^{Q,\delta}\right]
=\int_{(T^{*}M\setminus M)\sqcup
S^{*}M}a~d\tilde{\nu}\,
\end{equation}
holds for all $a\in {\cal C}^{\infty}_{0}((T^{*}M\setminus M)\sqcup S^{*}M)$\,.

Fix first the partition of unity $(1-\chid)+\chid\equiv 1$, $1-\chi \in \mathcal C^\infty_0(T^*M)$, $1-\chi\equiv 1$ in a neighborhood of $M$, 
and $\delta=\delta_{0}>0$\,. For a given $a\in
{\cal C}^{\infty}_{0}((T^{*}M\setminus M)\sqcup
S^{*}M)$\,, the inequalities (\ref{eq.realaa}) and (\ref{eq.posa})
imply that one can find a subsequence $(h_{k,\chi,\delta_{0},a})_{k\in \nz}$ of $(h_n)_{n\in\mathbb N}$\,, such that
\begin{equation}
  \label{eq:cva}
\lim_{k\to
  \infty}\Tr\left[\gamma_{h_{k,\chi,\delta_{0},a}}(\chid a(\cdot,h_{k,\chi,\delta_{0},a}\delta_{0}^{-1}\cdot)^{Q,\delta_{0}}\right]=\ell_{\chi,\delta_{0},a}\in
\cz\,.
\end{equation}
For a different partition of unity
$(1-\tilde{\chid})+\tilde{\chid}\equiv 1$ the symbol
$
[\chid-\tilde{\chid}] a(x,h\delta_{0}^{-1}\xi)$ is supported
in $ C_{\chi,\tilde{\chi},\delta_{0}}^{-1}\leq|\xi|\leq
C_{\chi,\tilde{\chi},\delta_{0}}$ and equals
$$
[\chid-\tilde{\chid}] a(x,h\delta_{0}^{-1}\xi)=
[\chid-\tilde{\chid}] a_{0}(x,\frac{\xi}{|\xi|})+hr_{\chi,\tilde{\chi},\delta_{0},h}(x,\xi)\,,
$$
where $a_{0}=a\big|_{S^{*}M}$ and
with $r_{\chi,\tilde{\chi},\delta_{0},h}$ uniformly bounded in
$S(1,dx^{2}+d\xi^{2})$\,. 
For $\delta_{0}>0$ fixed,
the operator 
$[(\chid-\tilde{\chid})a_{0}]^{Q,\delta_{0}}$ is a
compact operator and we obtain
$$
\lim_{h\to0}\Tr\left[\gamma_{h}(\chid a(\cdot,h\delta_{0}^{-1}\cdot))^{Q,\delta_{0}}\right]
-
\Tr\left[\gamma_{h}(\tilde{\chid} a(\cdot,h\delta_{0}^{-1}\cdot))^{Q,\delta_{0}}\right]
=\Tr\left[\gamma_{0}((\chid -\tilde{\chid})a_{0})^{Q,\delta_{0}}\right]\,.
$$
Therefore the subsequence extraction, which ensures the convergence
(\ref{eq:cva}) can be done independently of the choice of $\tilde \chi$ and by taking $\tilde{\chi}(x,\xi)=\chi(x,\delta\delta_{0}^{-1}\xi)$
independently of $\delta>0$\,. For ${\cal E}_{a}=(h_{k,a})_{k\in\nz}$
such a sequence of parameters, the limits can be compared by
\begin{eqnarray}
\nonumber 
\ell_{\tilde{\chi},\delta,a}-\ell_{\chi,\delta_{0},a}
&=&\lim_{\substack{h\in {\cal E}_{a}\\h\to 0}}
\Tr\left[(\gamma_{h}\tilde{\chid}a(.,h\delta^{-1}.))^{Q,\delta}\right]
-
\Tr\left[(\gamma_{h}\chid a(\cdot,h\delta_{0}^{-1}\cdot))^{Q,\delta_{0}}\right]\\
\label{eq:ellchi}
&=&\Tr\left[((\tilde{\chid}(\delta\delta_{0}^{-1})-\chid)a_{0})^{Q,\delta_{0}}\gamma_{0}\right]\,.
\end{eqnarray}
By choosing $\tilde{\chi}=\chi$ above, the inequality $0\leq (\chi -\chi(\delta \delta_0^{-1})) a_0\leq \chi a_0$, for $a_0\geq 0$ and $\delta\leq \delta_0$, and the $\delta_0$-Garding inequality implies
$$|\Tr\left[((\tilde{\chid}(\delta\delta_{0}^{-1})-\chid)a_{0})^{Q,\delta_{0}}\gamma_{0}\right]|\leq \Tr[(\chi a_0)^{Q,\delta_0}\gamma_0]+\mathcal O(\delta_0)$$
uniformly with respect to $\delta\leq \delta_0$\,. 
Thus the quantity $\ell_{\chi,\delta,a}$ thus satisfies the Cauchy criterion as $\delta\to 0$ because  $\mathop{\mathrm{s-lim}}_{\delta_0\to 0}(\chid a_{0})^{Q,\delta_0}=0$ and $\gamma_{0}$ is fixed in
${\cal L}^{1}(L^{2}(M,dx))$\,. Hence the limit
$$
\ell_{\chi,a} = \lim_{\delta\to 0}\ell_{\chi,\delta,a}=\lim_{\delta\to 0}\lim_{\substack{h\in
  {\cal E}_{a}\\h\to 0}}
\Tr\left[(\gamma_{h}\chid a(\cdot,h\delta^{-1}\cdot))^{Q,\delta}\right]
$$
exists for any fixed $a\in {\cal C}^{\infty}_{0}((T^{*}M\setminus
M)\sqcup S^{*}M)$\,. 
Using \eqref{eq:ellchi} with $\delta=\delta_0$ but a general pair $(\chi,\tilde \chi)$ and taking the limit as $\delta \to 0$ shows $\ell_{\tilde \chi,a}=\ell_{ \chi,a}=\ell_a$\,.
The inequalities (\ref{eq.realaa})
and~(\ref{eq.posa}) give $0\leq \ell_{a}\leq\|a\|_{L^{\infty}}$\,.
By the usual  diagonal extraction process
according to a countable set ${\cal N}\subset{\cal
  C}^{\infty}_{0}((T^{*}M\setminus M)\sqcup S^{*}M)$ dense in
the set of continuous functions with limit $0$ at infinity\,,
 we have found a subset
${\cal E}'\subset {\cal E}$\,, $0\in \overline{{\cal E}'}$\,, and a
non-negative measure $\tilde{\nu}$ such that (\ref{eq:mestnu})
holds. Note that we have also proved 
\begin{eqnarray*}
\int_{(T^{*}M\setminus M)\sqcup
S^{*}M}a~d\tilde{\nu}
&=&\lim_{\delta\to 0}\lim_{\substack{h\in {\cal E}'\\h\to 0}} 
\Tr\left[\gamma_{h}(\chid a(.,h\delta^{-1}.))^{Q,\delta}\right]
\\
&=&\lim_{\delta\to 0}\lim_{\substack{h\in {\cal E}'\\h\to 0}} 
\Tr\left[(\gamma_{h}-\gamma_{0})(\chid a(.,h\delta^{-1}.))^{Q,\delta}\right]\,
\end{eqnarray*}
where both limits do not depend on the partition of unity
$(1-\chid)+\chid\equiv 1$ with $1-\chid\in {\cal
  C}^{\infty}_{0}(T^{*}M)$ equal to $1$ in a neighborhood of $M$\,.
  
We still have to compare $\tilde{\nu}$ and $\nu$\,.
For this take $a\in {\cal C}^{\infty}_{0}(T^{*}M)$ and set
$a_{0}(x,\omega)=\varphi(x)=a(x,0)$\,. The symbol identity
$$
a(x,h\delta^{-1}\xi)=a(x,h\delta^{-1}\xi)(1-\chid)+a(x,h\delta^{-1}\xi)\chid
=\varphi(x)(1-\chid)+a(x,h\delta^{-1}\xi)\chid+hr_{a,\chi,\delta,h}
$$
with $r_{a,\delta,\chi,h}$ uniformly bounded in $S(1,dx^{2}+d\xi^{2})$
w.r.t.~$h$\,, leads after $\delta$-quantization to 
\begin{align*}
\int_{T^{*}M}a~d\nu
& =
\lim_{\substack{h\in {\cal E'}\\h\to   0}}\Tr\left[\gamma_{h}a^{Q,h}\right]
\\
& =
\lim_{\substack{h\in {\cal E'}\\h\to 0}}\Tr\left[\gamma_{h}(\varphi(x)(1-\chid))^{Q,\delta}\right]
+
\lim_{\substack{h\in {\cal E'}\\h\to   0}}\Tr\left[\gamma_{h}(\chid a(\cdot,h\delta^{-1}\cdot))^{Q,\delta}\right]\,.
\end{align*}
For $\delta>0$ fixed $(\varphi(x)(1-\chid))^{Q,\delta}$ is a fixed
compact operator so that the first limit is
$$
\lim_{\substack{h\in {\cal E'}\\h\to 0}}\Tr\left[\gamma_{h}(\varphi(x)(1-\chid))^{Q,\delta}\right]
=
\Tr\left[\gamma_{0}(\varphi(x)(1-\chid))^{Q,\delta}\right]\,,
$$
while the second one is exactly the quantity occuring in the
definition of $\tilde{\nu}$\,.
Taking the limit as $\delta\to 0$ with
$\mathop{\mathrm{s-lim}}_{\delta\to 0}(\varphi(x)(1-\chid))^{Q,\delta}=\varphi(x)$\,, yields 
$\nu\big|_{T^{*}M\setminus M}=\tilde{\nu}\big|_{T^{*}M\setminus M}$. 
Finally setting $\nu_{(I)}=\tilde\nu\big|_{S^{*}M}$ yields, for any $a\in {\cal C}^{\infty}_{0}(T^{*}M)$,
$$
\int_{T^*M}ad\nu
=\int_{T^*M\setminus M}ad\nu
 +\int_{S^*M}a_0 d\nu_{(I)}
 +\int_{M}\varphi d\nu_0
$$
which imply the relation for the measures.
\end{proof}
 
\begin{definition}
\label{de:M2}
${\cal M}^{(2)}(\gamma_{h}, h\in {\cal E})$ denotes the set of all triples $(\nu,\nu_{(I)},\gamma_{0})$  
which can be obtained in Prop.~\ref{pr:multivar} for suitable choices of ${\cal E}'\subset {\cal E}$\,, $0\in \overline{{\cal E}'}$\,. 
\end{definition}

We note that the equality $\nu(M)=\Tr\left[\gamma_{0}\right]$ implies
$\nu_{(I)}\equiv 0$ and this leads like in the previous case to the
following definition.
\begin{definition}
\label{de:sepM} 
  On a compact manifold $M$\,, assume that the quantization
  $a^{Q,h}=a(x,hD_{x})$ is adapted to the family $(\gamma_{h})_{h\in {\cal
      E}}$\,, with $\gamma_{h}\in {\cal L}^{1}(L^{2}(M))$\,,
  $\gamma_{h}\geq 0$ and $\lim_{h\to
    0}\Tr\left[\gamma_{h}\right]<\infty$\,. 
We say that  the quantization is \emph{separating} if for any ${\cal E}'\subset {\cal
    E}$\,, $0\in \overline{{\cal E}'}$\,, 
$$
\left.
 \begin{array}[c]{r}
    {\cal M}(\gamma_{h}, h\in {\cal E}')=\left\{\nu\right\}\,,\\
   \displaystyle
\mathop{\mathrm{w^{*}-lim}}_{h\in {\cal E}', h\to 0}
\gamma_{h}=\gamma_{0}\quad\text{in}~{\cal L}^{1}(L^{2}(M))
 \end{array}
\right\}
\Rightarrow \nu(\left\{\xi=0\right\})=\Tr\left[\gamma_{0}\right]\,.
$$
\end{definition}
While doing the double scale analysis of the non normalized reduced
density matrices $\bar \gamma^{(p)}_{h}$ especially with the help of
tensorization arguments, we will simply study their weak$^{*}$ limit in
${\cal L}^{1}$ and their semiclassical measures. The equality of
Definition~\ref{de:sepRd} or Definition~\ref{de:sepM} will be checked
a posteriori in order to ensure $\nu_{(I)}\equiv 0$\,.\\

\section{Mean-field asymptotics with $h$-dependent observables}
\label{sec:meanfield}
We now combine the mean-field asymptotics with semiclassically quantized observables. This means that the parameter $\varepsilon$
 appearing in CCR (resp. CAR) relations in Section~\ref{se.wickobs} is bound to the semiclassical parameter $h$ of Section~\ref{sec:classic}  parametrizing 
observables $a^{W,h}$ (or $a^{Q,h}$) :
$$
\varepsilon=\varepsilon(h)>0
\quad\text{with}\quad \lim_{h\to 0}\varepsilon(h)=0\,.
$$
Firstly, we give a sufficient condition in terms of semiclassical
$1$-particle observables and of the family
$(\varrho_{\varepsilon(h)})_{h\in {\cal E}}$ so that a quantization
$a^{W,h}$ defined on the $p$-particles phase-space ${\cal X}^{p}$ is
adapted to the non normalized reduced density matrix 
$\gamma_{\varepsilon(h)}^{(p)}$ for  all $p\in \nz$\,. If
$\lim_{h\to 0}\Tr\left[\gamma_{\varepsilon(h)}^{(p)}\right]=\lim_{h\to
0}\Tr\left[\varrho_{\varepsilon(h)}\mathbf{N}_{\pm}^{p}\right]=T^{(p)}$ then
the semiclassical measures $\nu^{(p)}\in
\mathcal{M} \big(\gamma_{\varepsilon(h)}^{(p)}, h\in {\cal E}\big)$
 (or multiscale asymptotic triples
 $(\nu^{(p)},\nu_{(I)}^{(p)},\gamma_{0}^{(p)})$)  have a total mass equal
to $T^{(p)}$\,.\\
After this, the quantum and classical symmetrization Lemmas
\ref{le:symm}  and \ref{le:clsymlem} then provide simple ways to
identify the weak limits $\gamma_{0}^{(p)}$ or the semiclassical
measures associated with the family
$(\gamma^{(p)}_{\varepsilon(h)})_{h\in{\cal E}}$ for all $p\in\nz$\,.
According to 
the discussion in Section~2, about Definitions~\ref{de:sepRd} and
\ref{de:sepM},  a simple mass argument allows to check that all the
multiscale information has been classified.
\\
Remember 
that the non normalized reduced density matrices $\gamma^{(p)}_{\varepsilon(h)}$ are
defined for $h>0$ by
$$
\forall \tilde{b}\in \mathcal{L}(\mathcal{S}_{\pm}^{p}\Z^{\otimes
  p})\,,\quad
\Tr\left[\gamma^{(p)}_{\varepsilon(h)}\tilde{b}\right]=
\Tr\left[\varrho_{\varepsilon(h)}\tilde{b}^{Wick}\right]\,.
$$
They are  well defined and uniformly bounded  trace-class operators
w.r.t  $h\in
{\cal E}$\,, as soon as 
$\Tr\left[\varrho_{\varepsilon(h)}\mathbf{N}^{p}\right]$ is bounded
uniformly
 w.r.t $h\in
{\cal E}$,  for every $p\in \nz$\,. Actually, it is more convenient 
in many cases, and not so restrictive,  to work with exponential weights 
in terms of the number operator $\mathbf{N}_{\pm}$\,.
\begin{hyp}
\label{hyp:borne}
The family $(\varrho_{\varepsilon(h)})_{h\in \mathcal{E}}$ in
$\mathcal{L}^{1}(\Gamma_{\pm}(\Z))$ satisfies
\begin{description}
\item[i)] For all $h\in \mathcal{E}$\,, $\varrho_{\varepsilon(h)}\geq
  0$ and $\Tr\left[\varrho_{\varepsilon(h)}\right]=1$\,;
\item[ii)] There exists $c,C>0$ such that 
$\Tr\left[\varrho_{\varepsilon(h)}e^{c\mathbf{N}_{\pm}}\right]\leq
C$\,, for all $h\in {\cal E}$\,.
\end{description}
\end{hyp}
When the one particle phase-space is ${\cal X}^{1}=T^{*}\rz^{d}$ we use
the Weyl quantization on ${\cal X}^{p}=T^{*}\rz^{dp }$\,,
$a^{Q,h}=a^{W,h}=a^{W}(h^{t}x, h^{1-t}D_{x})$\,, $x\in \rz^{dp}$\,, 
 and when $M^{1}$ is a compact
manifold, ${\cal X}^{p}=T^{*}M^{p}$\,, we use $a^{Q,h}=a(x,hD_{x})$\,, $x\in M^{p}$\,.
\begin{prop}
\label{pr:adaptgamp}
Assume Hypothesis~\ref{hyp:borne}. Let $\chi\in
\mathcal{C}^{\infty}_{0}(T^{*}M^{1})$ satisfy $0\leq \chi\leq 1$ and
$\chi\equiv 1$ in a neighborhood of $0$ (resp. in a neighborhood of the
null section $\left\{(x,\xi)\in T^{*}M\,, \xi=0\right\}=M$)  when
$M=\rz^{d}$
(resp. $M$ compact manifold) and let $\chi_{\delta}(X)=\chi(\delta X)$
(resp. $\chi_{\delta}(x,\xi)=\chi(x,\delta \xi)$)\,.
For $c'<c$\,, $c$ given by Hypothesis~\ref{hyp:borne}-ii)\,, If
\begin{equation}
  \label{eq:defsdelt}
s_{c',\chi}(\delta)=\limsup_{h\to 0}\Real\Tr\left[\varrho_{\varepsilon(h)}
  (e^{c'\mathbf{N}_{\pm}}
-e^{c'd\Gamma_{\pm}(\chi_{\delta}^{Q,h})})\right] 
\to 0 \quad \text{as}\quad \delta \to 0\,,
\end{equation}
then for all
$p\in\nz$\,, the quantization $a^{Q,h}$ is adapted to the family
$\gamma^{(p)}_{\varepsilon(h)}$\,.
\end{prop}
\begin{lem}
\label{le:estimexp}
Let $A\in \mathcal{L}(\Z)$ and $\alpha\geq \|A\|$. For  $z$ in the open disc $D(0,\frac{\alpha}{\|A\|})\subset \mathbb C$, 
the operator
  $e^{zd\Gamma_{\pm}(A)}e^{-\alpha\mathbf{N}_{\pm}}=e^{d\Gamma_{\pm}(zA-\alpha\Id_{\Z})}$
  is a contraction in  $\Gamma_{\pm}(\Z)$\,. 
\begin{enumerate}
\item The function $z\mapsto
  e^{d\Gamma_{\pm}(zA-\alpha \Id_{\Z})}$ is holomorphic in $D(0,\frac{\alpha}{\|A\|})$
 and
$$
\frac{1}{p!}d\Gamma_{\pm}(A)^{p}e^{-\alpha \mathbf{N}_{\pm}}=
e^{-\alpha \mathbf{N}_{\pm}}\frac{1}{p!}d\Gamma_{\pm}(A)^{p}
=\frac{1}{2i\pi}\int_{|z|=r}e^{d\Gamma_{\pm}(zA-\alpha \Id_{\Z})}\frac{dz}{z^{p+1}}
$$
holds in $\mathcal{L}(\Gamma_{\pm}(\Z))$ for all $p\in\nz$ and all
$r\in (0,\frac{\alpha}{\|A\|})$\,.
\end{enumerate}
Assume moreover that $A, B\in \mathcal{L}(\Z)$,
and $\alpha>\alpha_{0}=
\max\left\{\left\|A\right\|\,, \left\|B\right\|\right\}$, then:
\begin{enumerate}
\item[2.] For all $z\in D(0,\frac{\alpha}{\alpha_{0}})$\,,
$$
\left\|
(e^{zd\Gamma_{\pm}(B)}-e^{zd\Gamma_{\pm}(A)})e^{-\alpha\mathbf{N}_{\pm}}
\right\|_{\mathcal{L}(\Gamma_{\pm}(\Z))}\leq \frac{\alpha\|B-A\|_{\mathcal{L}(\Z)}}{\alpha_{0}(\alpha-\alpha_{0})e}\,.
$$
\item[3.]  This contains, for all $p\in\nz$ and $r\in (0,\frac{\alpha}{\alpha_{0}})$\,,
$$
\left\|\left(d\Gamma_{\pm}(B)^{p}-d\Gamma_{\pm}(A)^{p}\right)
e^{-\alpha \mathbf{N}_{\pm}}\right\|_{{\cal L}(\Gamma_{\pm}(\Z))}\leq
\frac{\alpha p!\|B-A\|_{{\cal L}(\Z)}}{\alpha_{0}(\alpha-\alpha_{0})er^{p}}\,.
$$
\end{enumerate}
\end{lem}
\begin{proof}[Proof of Lemma \ref{le:estimexp}]
  After setting $A'=zA$ with $|z|< \frac{\alpha}{\|A\|}$ so that $\|A'\|<\alpha$\,, notice that $\alpha\Id_{\Z}-A'=\alpha-A'$
  is a bounded accretive operator so that $(e^{-t\varepsilon (\alpha-A')})_{t\geq 0}$ is a
  strongly continuous semigroup of contractions on $\Z$\,, the same
  holds for $\Gamma_{\pm}(e^{-t\varepsilon (\alpha-A')})=e^{-\alpha
    \mathbf{N}_{\pm}}
e^{td\Gamma(A')}=e^{td\Gamma(A')}e^{-t\alpha \mathbf{N}_{\pm}}$
 in $\Gamma_{\pm}(\Z)$\,.  The holomorphy and the Cauchy
  formula are then standard.
  
For the second inequality, set $B'=zB$ and $A'=zA$\,,
$|z|<\frac{\alpha}{\alpha_{0}}$\,, and use Duhamel's
formula 
\begin{multline*}
e^{-d\Gamma_{\pm}(\alpha-B')}-e^{-d\Gamma_{\pm}(\alpha-A'
  )}
\\
=\int_{0}^{1}e^{-(1-t)d\Gamma_{\pm}(\alpha_{0}-A')}
d\Gamma_{\pm}(B'-A')
e^{-(\alpha-\alpha_{0})\mathbf{N}_{\pm}}
e^{-td\Gamma_{\pm}(\alpha_{0}-B')}~dt\,.
\end{multline*}
Since $e^{-(1-t)d\Gamma_{\pm}(\alpha_{0}-A')}$ and
$e^{-td\Gamma_{\pm}(\alpha_{0}-A')}$ are contractions\,, the inequality
$$
\|d\Gamma_{\pm}(B'-A')e^{-(\alpha-\alpha_{0})\mathbf{N}_{\pm}}\|\leq
\frac{\alpha}{\alpha_{0}}\|B-A\|\sup_{n\in\nz}\varepsilon n e^{-(\alpha-\alpha_{0})\varepsilon
  n}
\leq \frac{\alpha\|B-A\|}{\alpha_{0}(\alpha-\alpha_{0})e}\,,
$$
yields Point 2.

Point (3) follows from Point (1) and Point (2).
\end{proof}
\begin{proof}[Proof of Proposition~\ref{pr:adaptgamp}]
Fix $p\in \nz$\,.
We want  to find $\tilde{\chi}\in {\cal
  C}^{\infty}_{0}(T^{*}M^{p})$\,, $0\leq\tilde{\chi}\leq 1$
and $\tilde{\chi}\equiv 1$ in a
neighborhood of $\left\{X\in \rz^{2dp}\,, X=0\right\}$
(resp. $\left\{(x,\xi)\in T^{*}M^{p}\,, \xi =0\right\}=M^p$) when $M^{p}=\rz^{dp}$
(resp. when $M$ is a compact manifold)\,, such that
\begin{eqnarray*}
 &&\lim_{\delta\to 0}\limsup_{h\to 0}
{\cal T}(\delta,h)=0\\
\text{with}
&&
{\cal T}(\delta,h):=
\Real \Tr\left[\gamma_{\varepsilon(h)}^{(p)}(\Id_{{\cal
      S}_{\pm}^{p}\Z^{\otimes p}}-\tilde{\chi}_{\delta}^{Q,h})\right]
=
\Real \Tr\left[\varrho_{\varepsilon(h)}
(\Id_{{\cal    S}_{\pm}^{p}\Z^{\otimes p}}-\tilde{\chi}_{\delta}^{Q,h})^{Wick}\right]\,.
\end{eqnarray*}
We know that $\chi^{\otimes p}\in {\cal
  C}^{\infty}_{0}(T^{*}M^{p})$\,, with $0\leq\chi^{\otimes p}\leq 1$\,.
Take $\tilde{\chi}$ such that $\chi^{\otimes p}\leq \tilde{\chi}\leq
1$\,. 
For a constant $\kappa_{\delta}>0$ to be fixed,
 the inequalities of symbols 
 \begin{eqnarray*}
   &&0 \leq \chi_{\delta}^{\otimes p}\leq \tilde{\chi}_{\delta}\leq
      1\\
&& 0\leq \chi_{\delta}+\kappa_{\delta}h\leq 1+\kappa_{\delta}h
 \end{eqnarray*}
and the semiclassical calculus imply
\begin{eqnarray*}
  && \|(1-\tilde{\chi}_{\delta})^{Q,h}-\Real
     \left[(1-\tilde\chi_{\delta})^{Q,h}\right]\|_{{\cal L}(\Z^{\otimes p})}\leq
     C_{\delta}h\quad,\quad
\|\chi_{\delta}^{Q,h}-\Real\left[\chi_{\delta}^{Q,h}\right]\|\leq C_{\delta}h\,,
\\
&&0\leq 
\Real\left[(1-\chi_{\delta}^{\otimes p})^{Q,h}\right]+C'_{\delta}h
=1-(\Real\left[\chi_{{\delta}}^{Q,h}\right])^{\otimes p}+C'_{\delta}h
\\
&&
\qquad
\leq
  (1+2\kappa_{\delta}h)^{p}-(\Real\left[(\chi_{\delta}+\kappa_{\delta}h)^{Q,h}\right])^{\otimes
  p} +C''_{\delta}h
\quad\text{in}~{\cal L}(\Z^{\otimes p})\,,
\end{eqnarray*}
for some  constants $C_{\delta},C_{\delta}',C_{\delta}''>0$\,, chosen
according to $p\in \nz$\,, $\delta>0$  and $\kappa_{\delta}>0$\,. Moreover for $\delta>0$ fixed, 
the constant $\kappa_{\delta}$ can be chosen so that
$$
0\leq \Real\left[(\chi_{\delta}+\kappa_{\delta}h)^{Q,h}\right]\leq 1+2\kappa_{\delta}h\,.
$$
With $\|(1+\mathbf{N}_{\pm})^p e^{-c'/2\mathbf{N}_{\pm}}\|_{{\cal
    L}(\Gamma_{\pm}(\Z))}\leq C_{p,c'}$\,, the number
estimate~(\ref{eq:numbEst}) and the positivity property $(\tilde{b}\geq
0)\Rightarrow (\tilde{b}^{Wick}\geq 0)$\,,
 writing 
$$\varrho_{\varepsilon(h)}
=e^{-c/2\mathbf{N}_{\pm}}e^{c/2\mathbf{N}_{\pm}}\varrho_{\varepsilon(h)}
e^{c/2\mathbf{N}_{\pm}}
e^{-c/2\mathbf{N}_{\pm}}
$$ leads to
\begin{eqnarray*}
{\cal T}(\delta,h)&:=&\Real \Tr\left[\varrho_{\varepsilon(h)}(\Id_{{\cal
  S}_{\pm}^{p}\Z^{\otimes
  p}}-\tilde{\chi}_{\delta}^{Q,h})^{Wick}\right]
\\
&=&
\Tr\left[\varrho_{\varepsilon(h)}
(\Id_{{\cal
  S}_{\pm}^{p}\Z^{\otimes
  p}}-\Real\left[\tilde{\chi}_{\delta}^{Q,h}\right])^{Wick}\right]
+{\cal O}_{\delta}(h)
\\
&\leq&
\Tr\left[\varrho_{\varepsilon(h)} \left((1+2\kappa_{\delta}h)^{p}-(\Real\left[(\chi_{\delta}+\kappa_{\delta}h)^{Q,h}\right])^{\otimes
  p}\right)^{Wick}\right]+{\cal O}_{\delta}(h)\,.
\end{eqnarray*}
 We now use Propostion~\ref{cor:numberest} for
$$
{\cal T}(\delta,h)
\leq \Tr\Big[\varrho_{\varepsilon(h)}\Big(d\Gamma_{\pm}(1+2\kappa_{\delta}
  h)^{p}-d\Gamma_{\pm}\big(\Real[(\chi_{\delta}+\kappa_{\delta}h)^{Q,h}]\big)^{p}\Big)\Big]
+{\cal O}_{\delta}(h+\varepsilon(h))\,.
$$
The two operators ${\cal A}=d\Gamma_{\pm}(1+2\kappa_{\delta}h)$ and
${\cal
  B}=d\Gamma_{\pm}(\Real[(\chi_{\delta}+\kappa_{\delta}h)^{Q,h}])$ are
 commuting self-adjoint operators such that $0\leq {\cal
  B}\leq {\cal A}$\,, so that
$0\leq {\cal A}^{p}-{\cal B}^{p}\leq C_{p,c'}[e^{c'{\cal A}}-e^{c'{\cal
    B}}]$\,. We deduce
\begin{multline*}
{\cal T}(\delta,h)
\\
\leq
C_{p,c'}\Tr\left[\varrho_{\varepsilon(h)}e^{cN_{\pm}}e^{-cN_{\pm}}
\left(e^{d\Gamma_\pm(c'(1+2\kappa_{\delta}h))}-e^{d\Gamma_\pm(c'\Real 
[(\chi_{\delta}+\kappa_{\delta}h
    )^{Q,h}])}\right)\right]+{\cal O}_{\delta}(h+\varepsilon(h))\,.
\end{multline*}
We apply Lemma~\ref{le:estimexp} with $z=1$\,,
$A=c'(1+2\kappa_{\delta}h)$ and $B=c'$\,, or
$A=c'\Real[(\chi_{\delta}+\kappa_{\delta}h)^{Q,h}]$ and
$B=c'\chi_{\delta}^{Q,h}$\,,
and finally
$$
\alpha=c>\alpha_{0}=\frac{c+c'}{2}\geq
c'\max\left\{1+2\kappa_{\delta}h,
  \|(\chi_{\delta}+\kappa_{\delta}h)^{Q,h}\|,
\|\chi_{\delta}^{Q,h}\|\right\}
\quad \text{for}~h\leq h_{\delta,c,c'}\,
$$
and we get
$$
{\cal T}(\delta,h)\leq
\Real
\Tr\left[\varrho_{\varepsilon(h)}\left(e^{c'\mathbf{N}_{\pm}}-
e^{c'd\Gamma_{\pm}(\chi_{\delta}^{Q,h})}\right)\right]
+{\cal O}_{\delta}(h+\varepsilon(h))\,.
$$
We thus obtain
$$
\limsup_{h\to 0}{\cal T}(\delta,h)\leq s_{c',\chi}(\delta)
$$
and our assumption $\lim_{\delta\to 0}s_{c',\chi}(\delta)=0$ allows to conclude.
\end{proof}

\paragraph{Notation} For any open set $\Omega\subseteq\mathbb C$ the Hardy space $H^{\infty}(\Omega)$ is the space of bounded holomorphic functions on $\Omega$.
\begin{prop}
\label{pr:Phia}
Assume Hypothesis~\ref{hyp:borne}.
The set  $\mathcal{E}$ can be reduced to $\mathcal{E}'$ so that
$\mathcal{M}(\gamma^{(p)}_{\varepsilon(h)}, h\in
\mathcal{E}')=\left\{\nu^{(p)}\right\}$\,, where $\nu^{(p)}$ is a non-negative measure on $ T^{*}M^{p}/\mathfrak{S}_{p}$\,, i.e. a
measure on $(T^{*}M)^{p}$ with the invariance \eqref{eq:invsigmanu}.\\
When (\ref{eq:defsdelt}) is
satisfied,  this implies 
$\lim_{\substack{h\in {\cal E}'\\ h\to
    0}}\Tr\left[\gamma_{\varepsilon(h)}^{(p)}\right]=\int_{T^{*}M^{p}}d\nu^{(p)}(X)$
for all $p\in \nz$\,.\\
For any $a\in \mathcal{C}^{\infty}_{0}(\rz^{2d})$ there exists
$r_{a}>0$
 such that the function
$\Phi_{a,h}:s\mapsto \Tr\left[\varrho_{\varepsilon(h)}e^{s
    d\Gamma_{\pm}(a^{W,h})}\right]$ is uniformly bounded in 
$H^{\infty}(D(0,r_{a}))$ and
\begin{equation}
  \label{eq:holeq}
\lim_{\substack{h\in {\cal E}'\\ h\to 0}}\Phi_{a,h}(s)=\Phi_{a,0}(s)
:=\sum_{p=0}^{\infty}
\frac{s^{p}}{p!}\int_{T^{*}M^{p}}a^{\otimes
p}(X)~d\nu^{(p)}(X)\,.
\end{equation}
Reciprocally if $\Phi_{a,h}$ converges,  pointwise on the interval
$(-r_{a},r_{a})$ 
or in ${\cal D}'((-r_{a},r_{a}))$\,, to some function
$\Phi_{a,0}$,  as $h\to 0, h\in {\cal
  E}$\,, then  ${\cal M}(\gamma_{\varepsilon(h)}^{(p)}\,,
h\in {\cal E})=\left\{\nu^{(p)}\right\}$ for all $p\in\nz$ and
$\Phi_{a,0}$ is equal to (\ref{eq:holeq}).
\end{prop}
\begin{proof}
The uniform bound
$$
\Tr\left[\gamma^{(p)}_{\varepsilon(h)}\right]
\leq \Tr\left[\varrho_{\varepsilon(h)}\langle\mathbf{N}\rangle^{p}\right]\leq
C_{p,c}\Tr\left[\varrho_{\varepsilon(h)}e^{c\mathbf{N}}\right]
$$
and Hypothesis~\ref{hyp:borne}
ensure for each $p\in \nz$ the existence of ${\cal E}^{(p)}\subseteq {\cal E}^{(p-1)}\subseteq
{\cal E}$\,, $0\in \overline{{\cal E}^{(p)}}$\,, and 
${\cal M}(\gamma^{(p)}_{\varepsilon(h)}, h\in{\cal
  E}^{(p)})=\left\{\nu^{(p)}\right\}$\,.
A diagonal extraction w.r.t to $p$ determines ${\cal E}'\subset{\cal
  E}$\,, $0\in \overline{{\cal E}'}$\,, such that ${\cal
  M}(\gamma^{(p)}_{\varepsilon(h)}\,, h\in {\cal
  E}')=\left\{\nu^{(p)}\right\}$ for all $p\in \nz$\,.\\
The second statement  is a straightforward application of 
Lemma~\ref{le:estimexp}.\\
Thanks to the classical symmetrization Lemma~\ref{le:clsymlem}\,, the
measures $\nu^{(p)}$ are determined after integrating with all the test
functions $a^{\otimes p}$ with $a\in {\cal C}^{\infty}_{0}(T^{*}M)$\,.\\
 Hypothesis~\ref{hyp:borne} now combined with
$$
\|e^{-c\mathbf{N}}e^{zd\Gamma(\alpha^{W,h})}\|
=\| \Gamma(e^{\varepsilon( z\alpha^{W,h})-c})\|
\leq e^{|z|\|a^{W,h}\|-c}
\quad \text{and}\quad \|a^{W,h}\|\leq C_{a}\,,
$$
provides
the uniform boundedness w.r. t $h\in {\cal E}$ 
of $\Phi_{a,h}$ in
$H^{\infty}(D(0,2r_{a}))$\,.  In any ${\cal E}_{1}\subset{\cal
  E}$\,, $0\in \tilde{E}_{1}$\,, we can find a  subset 
${\cal E}_{2}$\,, $0\in \overline{{\cal E}_{2}}$\,, such that
$\Phi_{a,h}$\,,  locally uniformly in $D(0,2r_{a})$ and therefore in
$H^{\infty}(D(0,r_{a}))$\,, to some function $\Phi_{a,0}$\,. In particular when ${\cal E}_{2}\subset 
{\cal E}_{1}\subset {\cal E}'$\,,
Corollary~\ref{cor:numberest} implies
\begin{eqnarray*}
\int_{T^{*}M^{p}}a^{\otimes p}~d\nu^{(p)}&=&
\lim_{\substack{h\in {\cal E}_{2}\\h\to 0}}
\Tr\left[\gamma^{(p)}_{\varepsilon(h)}(a^{^{W,h}})^{\otimes  p}\right]
=\lim_{\substack{h\in {\cal E}_{2}\\h\to 0}}
\Tr\left[\varrho_{\varepsilon(h)}\left((a^{^{W,h}})^{\otimes
                                             p}\right)^{Wick}\right]
\\
&&=\lim_{\substack{h\in {\cal E}_{2}\\h\to 0}}
  \Tr\left[\varrho_{\varepsilon(h)}d\Gamma(a^{W,h})^{p}\right]
 =\lim_{h\to 0}\frac{d^{p}\Phi_{a,h}}{ds^{p}}(0)
 =\frac{d^{p}\Phi_{a,0}}{ds^{p}}(0)
\end{eqnarray*}
Hence the limit $\Phi_{a,0}\in H^{\infty}(D(0,r_{a}))$\,, as $h\in
{\cal E}_{2}$\,, $h\to 0$\,,  equals the right-hand side of
(\ref{eq:holeq}) and this uniqueness implies the convergence for the
whole family $(\Phi_{a,h})_{h\in {\cal E}'}$\,.\\
Reciprocally assume the convergence of $\Phi_{a,h}$ to $\Phi_{a,0}$
in a weak topology on the interval $(-r_{a},r_{a})$ as $h\in {\cal
  E}$\,. With the a uniform bound on $\Phi_{a,h}$
in $H^{\infty}(D(0,2r_{a}))$\,, $\Phi_{a,0}$ has an holomorphic
extension in $D(0,r_{a})$\,.  Additionally we can extract a subset ${\cal E}'\subset {\cal E}$
such that ${\cal M}(\gamma^{(p)}_{h}\,, h\in {\cal
  E}')=\left\{\nu^{(p)}\right\}$ and (\ref{eq:holeq}) hold.
Again the uniqueness of the limit $\Phi_{a,0}\big|_{(-r_{a},r_{a})}$
and of its holomorphic extension to $D(0,r_{a})$
  ends the proof.
\end{proof}
Replacing the semiclassical symmetrization Lemma~\ref{le:clsymlem} by
the quantum ones, Lemma~\ref{le:symm} in the above proof leads to the
following similar result for the quantum part.
\begin{prop}
 \label{pr:quantpart}
Assume Hypothesis~\ref{hyp:borne}. 
For all $K\in \mathcal{L}^{\infty}(\Z)$ there exists $r_{K}>0$ such that the set $\{\Psi_{K,h}, h\in \mathcal E\}$  of functions
$\Psi_{K,h}(s):=\Tr\left[\varrho_{\varepsilon(h)}e^{s d\Gamma_{\pm}(K)}\right]$
is bounded in $H^{\infty}(D(0,r_{K}))$\,.\\
The pointwise  or ${\cal D}'((-r_{K},r_{K}))$-convergence 
$\lim_{\substack{h\in {\cal E}\\h\to 0}}\Psi_{K,h}=\Psi_{K,0}$ is
equivalent to
$\mathop{\mathrm{w^{*}-lim}}_{\substack{h\in {\cal E},\\ h\to
  0}}\gamma^{(p)}_{h}=\gamma^{(p)}_{0}$ 
(remember $\mathcal{L}^{1}=(\mathcal{L}^{\infty})^{*}$) with
$$
\Psi_{K,0}(s)=\sum_{p=0}^{\infty}\Tr\left[\gamma_{0}^{(p)}
K^{\otimes p}\right]\frac{s^{p}}{p!}\,,
$$
\end{prop}
Let us consider the fermionic case:
\begin{prop}
\label{prop:gamma_0-vanishes-for-fermions}Let \textup{$(\varrho_{\varepsilon})_{\varepsilon\in\mathcal{E}}$
be a family of non-negative, trace $1$ operators in $\mathcal{L}^{1}(\Gamma_{-}(\mathcal{Z}))$.
Let} $\gamma_{\varepsilon}^{(p)}$ denote the corresponding non normalized
reduced density matrices of order $p$.\textup{ If $\gamma_{0}^{(p)}\in\mathcal{L}^{1}(\mathcal{S}_{-}^{p}\mathcal{Z}^{\otimes p})$
is such that 
\[
\forall K\in\mathcal{L}^{\infty}(\mathcal{S}_{-}^{p}\mathcal{Z}^{\otimes p})\,,\quad\lim_{\substack{\varepsilon\in\mathcal{E}\\
\varepsilon\to0
}
}\Tr[\gamma_{\varepsilon}^{(p)}K]=\Tr[\gamma_{0}^{(p)}K]\,,
\]
}then $\gamma_{0}^{(p)}=0$.

As a consequence, the weak limits $\gamma_0^{(p)}$ always vanish in the fermionic case.\end{prop}
\begin{proof}
First consider $K$ a non-negative finite rank operator. Then 
\[
\lim_{\substack{\varepsilon\in\mathcal{E}\\
\varepsilon\to0
}
}\Tr[\varrho_{\varepsilon}K^{Wick}]=\Tr[\gamma_{0}^{(p)}K]\,.
\]
For fermions, $K^{Wick}\leq\varepsilon^{p}\Tr[K]$, and hence $\Tr[\varrho_{\varepsilon}K^{Wick}]\leq\varepsilon(h)^{p}\to0$
as $\varepsilon\to0$. Any finite rank operator being of the form
$K=K_{1}-K_{2}+i(K_{3}-K_{4})$ for some non-negative finite rank
operators $K_{j}$, $j\in\{1,2,3,4\}$, the limit $\Tr[\varrho_{\varepsilon}K^{Wick}]\to0=\Tr[\gamma_{0}^{(p)}K]$
holds for any finite rank operator $K$. Hence, by density of the
finite rank operators in the compact operators for the operator norm,
$\Tr[\gamma_{0}^{(p)}K]=0$ for any $K\in\mathcal{L}^{\infty}(\mathcal{S}_{-}^{p}\mathcal{Z}^{\otimes p})$,
i.e., $\gamma_{0}^{(p)}=0$.\end{proof}

\section{Examples}
\label{sec:examples}

\subsection{$h$-dependent coherent states in the bosonic case}
\label{sec:coh}
We first recall our normalization for a coherent state. 
We need the notion of empty state: if we use the identification $\mathcal S^0_{\pm}\mathcal Z\equiv \mathbb C$, then the empty state is defined as $\Omega=(1,0,0,\dots)\in\Gamma_{\pm}(\mathcal Z)$\,. 
We then introduce the usual field operators  $\Phi(f)=\frac{1}{\sqrt{2}}(a^*(f)+a(f))$\,, 
with $f\in\mathcal Z$ and the Weyl operators are the $W(f)=\exp(\frac{i}{\sqrt{2}}\Phi(f))$\,. 
A coherent state is a pure state $E_z=W(\frac{\sqrt{2}z}{i\varepsilon})\Omega$, with $z\in\Z$\,. One then also speak of coherent state for the corresponding density matrix $|E_z\rangle\langle E_z|$\,. One of the useful properties of coherent states is that
\begin{equation}\label{eq:coherent-states-yields-symbol}
b(z)=\langle E(z),b^{Wick}E(z)\rangle\,.
\end{equation}
(See e.g. \cite[Prop.~2.10]{AmNi1} )
The case of coherent states is simple:
\begin{prop}
\label{prop:one-coherent-state}Let $(z_{\varepsilon})_{\varepsilon\in(0,1]}$
a bounded family of $\mathscr{Z}$\,, choose the semiclassical quantization
$a\mapsto a^{W,h}=a^{W}(\sqrt{h}x,\sqrt{h}D_{x})$\,, and fix a function
$\varepsilon=\varepsilon(h)\to0$ as $h\to0$\,. Up to an extraction,
$z_{\varepsilon(h)}\to z_{0}\in\mathcal{Z}$ weakly, and $\mathcal{M}(|z_{\varepsilon(h)}\rangle\langle z_{\varepsilon(h)}|,\, h\in\mathcal{E})=\{\nu\}$\,.
Assume that the semiclassical quantization $a^{W,h}=a^{W}(\sqrt{h}x,\sqrt{h}D_{x})$
is adapted to $(|z_{\varepsilon(h)}\rangle\langle z_{\varepsilon(h)}|)_{h}$
and separating for $(|z_{\varepsilon(h)}\rangle\langle z_{\varepsilon(h)}|)_{h}$\,.
Then the family $(\varrho_{\varepsilon(h)}=|E_{z_{\varepsilon(h)}}\rangle\langle E_{z_{\varepsilon(h)}}|)_{h\in\mathcal{E}}$
has $\gamma_{\varepsilon(h)}^{(p)}=|z_{\varepsilon(h)}^{\otimes p}\rangle\langle z_{\varepsilon(h)}^{\otimes p}|$
as (non normalized) reduced density matrices of order $p$\,, for which
the quantization is adapted and separating, and 
\[
\mathcal{M}^{(2)}(\gamma_{\varepsilon(h)}^{(p)},\, h\in\mathcal{E})=\{(\nu^{\otimes p},0,|z_{0}^{\otimes p}\rangle\langle z_{0}^{\otimes p}|)\}\,.
\]
\end{prop}
\begin{proof}
Formula~\eqref{eq:coherent-states-yields-symbol} yields, for $B\in\mathcal{L}(\mathcal{S}_{+}^{p}\mathcal{Z}^{\otimes p})$,
\[
\langle z_{\varepsilon(h)}^{\otimes p},Bz_{\varepsilon(h)}^{\otimes p}\rangle=\langle E_{z_{\varepsilon(h)}}|B^{Wick}|E_{z_{\varepsilon(h)}}\rangle=\Tr[\varrho_{\varepsilon(h)}B^{Wick}]=\Tr[\gamma_{\varepsilon(h)}^{(p)}B]\,,
\]
which implies the result.
\end{proof}
The case of coherent states, although simple, can already exhibit
interesting behaviors for some  families $(z_{\varepsilon})_{\varepsilon\in(0,1]}$.
Indeed,
\begin{remark}
Let $(z_{j,\varepsilon})_{\varepsilon\in(0,1]}$, $j\in\{1,2\}$,
be families of $\mathcal{Z}$ such that 
\begin{itemize}
\item $z_{1,\varepsilon}\xrightarrow[\varepsilon\to0]{}z_{1,0}\in\mathcal{Z}$,
and 
\item $(z_{2,\varepsilon})_{\varepsilon\in(0,1]}$ is bounded, converges
weakly to $0$, $\lim_{R\to\infty}\limsup_{\varepsilon\to0}\|z_{2,\varepsilon}1_{\complement B(0,R)}\|=0$
(no mass escaping at infinity), and $\mathcal{M}(|z_{2,\varepsilon(h)}\rangle\langle z_{2,\varepsilon(h)}|,\, h\in\mathcal{E})=\{\nu_{2}\}$,
with $\nu_{2}(\{0\})=0$.
\end{itemize}

Then 
$(|z_{1,\varepsilon}+z_{2,\varepsilon}\rangle\langle z_{1,\varepsilon}+z_{2,\varepsilon}|)_{\varepsilon\in(0,1]}$
satisfies the assumptions of Prop.~\ref{prop:one-coherent-state},
and $z_{0}=z_{1,0}$,  $\nu=\|z_{1,0}\|^{2}\delta_{0}+\nu_{2}$.

\end{remark}

\subsection{Gibbs states}
\label{sec:BEC}
For a given non-negative self-adjoint hamiltonian $H$ defined in
$\Z$ with domain $D(H)$\,, the Gibbs state at positive temperature
$\frac{1}{\beta}$ and with the chemical potential $\mu<0$ is given by 
$$
\omega_{\varepsilon}(A)=
\frac{\Tr\left[\Gamma_{\pm}(e^{-\beta(H-\mu)})A\right]}{
\Tr\left[\Gamma_{\pm}(e^{-\beta(H-\mu)})\right]}=\Tr\left[\varrho_{\varepsilon}A\right]\,.
$$
In general $\varrho_{\varepsilon}\in \mathcal{L}^{1}(\Gamma_{\pm}(\Z))$
as soon as $e^{-\beta(H-\mu)}\in \mathcal{L}^{1}(\Z)$  (in the bosonic case $H\geq 0$ and $\mu<0$ imply $\|e^{-\beta(H-\mu)}\|_{\mathcal L(\mathcal Z)}<1$, 
see 
Lemma~\ref{le:trbos}). Moreover the quasi-free state formula (see
\cite{BrRo2}) with $\varepsilon$-dependent quantization gives
\begin{eqnarray*}
\Tr\left[\varrho_{\varepsilon}\N\right]
&=&
\varepsilon\Tr\left[e^{-\beta(H-\mu)}(1\mp e^{-\beta(H-\mu)})^{-1}\right]\,
\end{eqnarray*}
and additionally, in the case of bosons,
\begin{eqnarray*}
  \Tr\left[\varrho_{\varepsilon}W(f)\right]
&=&
\exp\left[-\frac{\varepsilon}{4}\langle
    f\,,(1+e^{-\beta(H-\mu)})(1-e^{-\beta(H-\mu)})^{-1}f\rangle\right]\,.
\end{eqnarray*}

\subsubsection{In the fermionic case}
\label{sec:fermions}
We begin by the fermionic case, which is simpler than the bosonic case for two reasons: first because the quantum part vanishes (see Prop.~\ref{prop:gamma_0-vanishes-for-fermions}), and second because there is no singularity to handle. To fix the ideas we consider the simple case when $H$ is the harmonic oscillator. Actually one can treat more general pseudo differential operators, and we do that below in the more interesting case of bosons and Bose-Einstein condensation.
\begin{prop}
Let $\beta>0$, $H=\frac{1}{2}|X|^{2\, W,h}$, $\mu(\varepsilon)$
such that $\mu(\varepsilon)\geq C\varepsilon$ for some constant $C>0$,
and assume that $\varepsilon=\varepsilon(h)=h^{d}$. Let $\varrho_{\varepsilon(h)}=\frac{\Gamma_{-}(e^{-\beta(H-\mu(\varepsilon))})}{\Tr[\Gamma_{-}(e^{-\beta(H-\mu(\varepsilon))})]}$
and $\gamma_{\varepsilon(h)}^{(p)}$ its non normalized reduced density
matrix of order $p\geq 1$. Then $\mathcal{M}^{(2)}(\gamma_{\varepsilon(h)}^{(p)},h\in(0,1])=\{(\nu^{(p)},0,0)\}$,
where $d\nu^{(p)}=p! \, \Big(\frac{e^{-\beta|X|^{2}/2}}{1+e^{-\beta|X|^{2}/2}}\frac{dX}{(2\pi)^{d}}\Big)^{\otimes p}$.\end{prop}
\begin{proof}
From Rem.~\ref{gamma0-is-the-weak-limit-of-gammah} and Prop.~\ref{prop:gamma_0-vanishes-for-fermions},
any $(\nu^{(p)},\nu_{I}^{(p)},\gamma_{0}^{(p)})\in\mathcal{M}^{(2)}(\gamma_{\varepsilon(h)}^{(p)},h\in(0,1])$
satisfies $\gamma_{0}^{(p)}=0$.

Since we are considering a Gibbs state, the Wick formula yields
\[
\gamma_{\varepsilon(h)}^{(p)}=p! \, \varepsilon^{p} \, \mathcal{S}_{\pm}^{p} \, \gamma_{\varepsilon(h)}^{(1)\otimes p} \, \mathcal{S}_{\pm}^{p,*}
\,.
\]
Moreover, in the fermionic case, $\gamma_{\varepsilon(h)}^{(1)}=\frac{C}{1+C}$
for $\varrho_{\varepsilon(h)}=\frac{\Gamma_{-}(C)}{\Tr[\Gamma_{-}(C)]}$,
that is to say 
$\gamma_{\varepsilon(h)}^{(1)}
=h^d \frac{e^{-\beta(H-\mu)}}{1+e^{-\beta(H-\mu)}}$
in our case. The semiclassical calculus combined with Helffer-Sj\"ostrand functional calculus formula yields 
\[
\frac{e^{-\beta(|X|^{2\, W,h}/2-\mu)}}{1+e^{-\beta(|X|^{2\, W,h}/2-\mu)}}
=\Big(\frac{e^{-\beta|X|^{2}/2}}{1+e^{-\beta|X|^{2}/2}}\Big)^{W,h}
+\mathcal{O}(h) \quad \text{in}\quad \mathcal{L}(\mathcal{Z})\,.
\]
For details we refer the reader e.g. to \cite{DiSj,HeNi} or to the proof of Prop.~\ref{pr:scsing}. 
Again by the semiclassical calculus we know
 $h^d a^{W,h}$ is uniformly bounded in $\mathcal L^1(L^2(\mathbb{R}^d))$ for $a\in \mathcal C_0^\infty(\mathbb{R}^{2d})$. This leads to
\[
\Tr[(a^{W,h})^{\otimes p}\gamma_{\varepsilon(h)}^{(p)}]
=p!\, \Tr[a^{W,h}\gamma_{\varepsilon(h)}^{(1)}]^p
=\: p!\, \Tr\Big[(\frac{e^{-\beta|X|^{2}/2}}{1+e^{-\beta|X|^{2}/2}})^{W,h} \, h^{d} a^{W,h}\Big]^{p}+\mathcal{O}(h)\,.
\]

We finally use 
$h^{d} \Tr[a^{W,h}b^{W,h}]=\int_{\mathbb{R}^{2d}}a(X)b(X) \frac{dX}{(2\pi)^d}$
which implies 
\[
\lim_{h\to0} \Tr[(a^{W,h})^{\otimes p} \, \gamma_{\varepsilon(h)}^{(p)}]
=p! \, \Big(\int_{\mathbb{R}^{2d}}\frac{e^{-\beta|X|^{2}/2}}{1+e^{-\beta|X|^{2}/2}}a(X)\frac{dX}{(2\pi)^d}\Big)^{p}\,.
\]
Hence $d\nu^{(p)}(X)
=p! \, (\frac{e^{-\beta|X|^{2}/2}}{1+e^{-\beta|X|^{2}/2}} \frac{dX}{(2\pi)^d})^{\otimes p}$.
\end{proof}
 
\subsubsection{Parameter dependent Gibbs states  and Bose-Einstein condensation in the bosonic
  case}
\label{sec:parGibbs}

The Bose-Einstein condensation phenomenon occurs when $H$ has a
ground state $\ker H=\cz \psi_{0}$ and the chemical potential is 
scaled according to
$$
-\beta\mu=\frac{\varepsilon}{\nu_{C}}\quad\text{for some fixed}~\nu_{C}>0\,.
$$
An especially interesting case is when $H$ is a semiclassically
quantized symbol with semiclassical parameter $h$ related to
$\varepsilon$\,, or $\varepsilon=\varepsilon(h)$ according to our
previous notations.  The  quantum and
semiclassical parts arise simultaneously when $\varepsilon=h^{d}$\,.
Two cases will be considered: the first one concerns
 $\Z=L^{2}(\rz^{d})$ with a non degenerate bottom well hamiltonian; the second one $\Z= L^2(M)$ with the semiclassical Laplace-Beltrami operator on the compact riemannian manifold $M$\,.\\
In the first case, let $S(\langle X\rangle^	m, \frac{dX^2}{\langle X\rangle^2})$ denote the H\"ormander class of symbols satisfying $|\partial_X^{\beta}a(X)|\leq C_\beta\langle X\rangle^{m-\beta}$\,, and let $\alpha\in S(\langle X\rangle^2,\frac{dX^2}{\langle X\rangle^2})$ be elliptic in this class with a unique non degenerate minimum at $X=0$ (e.g. the symbol of the harmonic oscillator hamiltonian). We can even consider small perturbations of this situation after setting
$$
H=\alpha^{W,h}+B^{h}-\lambda_{0}(\alpha^{W,h}+B^{h})\quad,\quad
\alpha^{W,h}=\alpha(\sqrt{h}x,\sqrt{h}D_{x})\quad,\quad \varepsilon=h^{d}\,,
$$
where 
 $B^{h}=B^{h\,*}\in {\cal
  L}(L^{2}(\rz^{d}))$\,, $\|B^{h}\|=o(h)$ and 
$\lambda_{0}(\alpha^{W,h}+B^{h})=\inf \sigma(\alpha^{W,h}+B^{h})$\,. It is convenient in this
case to introduce the linear symplectic transformation $T\in Sp_{2d}(\rz)$ such that ${}^{t}X{}^{t}T^{-1}\mathrm{Hess}~
\alpha(0)T^{-1}X=\sum_{j=1}^{d}\beta_{j}X_{j}^{2}$ and to introduce 
some unitary quantization $U_{T}$ of $T$\,, i.e. a unitary operator on  $L^{2}(\rz^{d})$ such that $U_{T}^{*}b^{W}U_{T}=b(T^{-1}.)^{W}$\,.
\begin{prop}
\label{pr:WigBEC}
Under the above assumptions with dimension $d\geq 2$\,, for any
$p\in\nz$\,,  ${\cal M}^{(2)}(\gamma_{\varepsilon(h)}^{(p)}, h\in
{\cal E})=\left\{(\nu^{(p)}\,, 0\,, \gamma_{0}^{(p)})\right\}$ (see Def.~\ref{de:M2}),
where
\begin{eqnarray*}
  && \gamma^{(p)}_{0}=p!|\psi_{0}^{\otimes p}\rangle \langle
     \psi_{0}^{\otimes
     p}|\quad\text{with}~\psi_{0}(x)=U_{T}\frac{e^{-x^{2}/2}}{\pi^{d/4}}\\
&&\nu^{(p)}
=\sum_{\sigma\in \mathfrak{S}_{p}}\sigma_{*}
\left[\sum_{k=0}^{p}\frac{1}{(p-k)!}\nu_{C}^{k}\delta_{0}^{\otimes
   k}\otimes \nu(\beta,\cdot)^{\otimes p-k}\right]\,,\\
\text{with}&&d\nu(\beta,X)=\frac{e^{-\beta \alpha(X)}}{1-e^{-\beta\alpha(X)}}\frac{dX}{(2\pi)^{d}}\,.
\end{eqnarray*}
\end{prop}
The proof is, given in Section~\ref{sec:secBEC}, needs some preliminaries given in Prop.~\ref{pr:scsing} and Lemma~\ref{le:alpalp0}.

Another even simpler case, related to the example $M=\tz^{d}$ 
presented in \cite{AmNi1}, is $\Z=L^{2}(M,dv_{g}(x))$ when $(M,g)$ 
is a compact Riemannian manifold with volume $dv_{g}(x)$ and
$$
H=-h^{2}\Delta_{g}+B_{h}-\lambda_{0}(-h^{2}\Delta_{g}+B_{h})\,,
$$
where $\Delta_{g}$ is the Laplace Beltrami operator on $(M,g)$ and
$B_{h}=B_{h}^{*}\in {\cal L}(L^{2}(M))$\,,
$\|B_{h}\|=\underline{o(h^{2})}$\,. 
\begin{prop}
\label{pr:BECmfld}
Under the above assumptions with $d\geq 3$\,, for any $p\in \nz$\,,
${\cal M}^{(2)}(\gamma_{\varepsilon(h)}^{(p)}\,, h\in {\cal
  E})=\left\{(\nu^{(p)},0,\gamma_{0}^{(p)})\right\}$ where
\begin{eqnarray*}
  && \gamma^{(p)}_{0}=p!|\psi_{0}^{\otimes p}\rangle \langle
     \psi_{0}^{\otimes p}|\quad,\quad
     \psi_{0}=\frac{1}{v_{g}(M)^{1/2}}\,,\\
&& \nu^{(p)}=\sum_{\sigma\in \mathfrak{S}_{p}}\sigma_{*}\left[
\sum_{k=0}^{p}\frac{1}{(p-k)!}\nu_{C}^{k}
(\frac{1}{v_{g}(M)}dv_{g}(x)\otimes \delta_{0}(\xi))^{\otimes k}\otimes \nu(\beta)^{\otimes (p-k)}\right]\,,\\
\text{with}&&
d\nu(\beta,X)=\frac{e^{-\beta |\xi|^{2}_{g(x)}}}{1-e^{-\beta
              |\xi|^{2}_{g(x)}}}\frac{dxd\xi}{(2\pi)^{d}}\,,\\
\text{and}&&
|\xi|_{g(x)}^{2}=\sum_{i,j\leq
             d}g^{ij}(x)\xi_{i}\xi_{j}\quad\text{when}\quad
g=\sum_{i,j\leq d}g_{ij}(x)dx^{i}dx^{j}\quad,\quad (g_{ij})^{-1}=(g^{ij})\,.
\end{eqnarray*}
\end{prop}
We  shall focus on the first case which requires  a more carefull
analysis, while $\sigma(-h^{2}\Delta_{g})=h^{2}\sigma(-\Delta_{g})$
reduces even more easily the problem to the integrability of
$\frac{e^{-\beta|\xi|_{g(x)}^{2}}}{1-e^{-\beta|\xi|_{g(x)}^{2}}}$
valid when $d\geq 3$\,. The
proof of Proposition~\ref{pr:BECmfld} is left as an exercise, which
requires the adaptation of the following arguments in the case of
Proposition~\ref{pr:multivar} with the associated
Definitions~\ref{de:sepM} and \ref{de:M2}.
\subsubsection{Semiclassical asymptotics with a singularity at $X=0$}
\label{sec:scsing}
We give here a general semiclassical result in  $T^{*}\rz^{d}$\,, which
involves traces and symbols with a singularity at $X=0$\,.
\begin{prop}
 \label{pr:scsing} 
Consider the hamiltonian
 $H=\alpha^{W,h}+B_{h}-\lambda_{0}(\alpha^{W,h}+B^{h})$\,,
with $\alpha^{W,h}=\alpha(\sqrt{h}x,\sqrt{h}D_{x})$\,, $\alpha\in
S(\langle X\rangle^{2},\frac{dX^{2}}{\langle X\rangle^2})$ elliptic and
real such that $\alpha(0)=0$ is the unique non degenerate minimum,
$B_{h}=B_{h}^{*}\in \mathcal{L}(L^{2}(\rz^{d}))$, $\|B_{h}\|=o(h)$\,, and
$\lambda_{0}(\alpha^{W,h}+B_{h})=\inf \sigma(\alpha^{W,h}+B_{h})$\,.
Assume that $f\in \mathcal{C}^{\infty}((0,+\infty);\rz)$ is decreasing and satisfies 
$$
0\leq f(u)\leq Cu^{-\kappa_{\infty}}\,,\quad \lim_{u\to
  0^{+}}u^{\kappa_{0}}f(u)=f_{0}
\in \rz\,,\quad  0<\kappa_{0}<d<\kappa_{\infty}\,.
$$ 
For $c>0$\,, the operator $f(H+ch^{d/\kappa_{0}})$ is trace class with 
$$
\limsup_{h\to
  0^{+}}h^{d}\|f(H+ch^{d/\kappa_{0}})\|_{\mathcal{L}^{1}(L^{2}(\rz^{d}))}<+\infty\,.
$$
Moreover the convergence 
$$
\lim_{h\to 0}h^{d}\Tr\left[f(H+ch^{d/\kappa_{0}})a^{W,h}\right]
=\frac{f_{0}}{c^{\kappa_{0}}}a(0)+\int_{\rz^{2d}}f(\alpha(X))a(X)~\frac{dX}{(2\pi)^{d}}\,,
$$
holds for all $a\in S(1,dX^{2})$\,.
Finally, all the above estimates and convergences hold uniformly with respect to $c\in (\frac{1}{A},A)$ for any fixed $A>1$\,.
\end{prop}
The following Lemma gives in a simple way useful inequalities for our
purpose which are deduced with elementary arguments, an in a robust
way
w.r.t the perturbation $B_{h}$\,,  from more accurate and
sophisticated results on the spectrum of $\alpha^{W,h}$ (see
\cite{ChVN}\cite{DiSj} and references therein).
\begin{lem}
 \label{le:alpalp0}
Let $\alpha\in S(\langle X\rangle^{2},\frac{dX^{2}}{\langle
  X\rangle^{2}})$ be real-valued,
elliptic, which means 
$1+\alpha(X)\geq C^{-1}\langle X\rangle^{2}$\,, with a unique
non degenerate minimum at $X=0$ and set
$\alpha_{0}(X)=\frac{|X|^{2}}{2}$\,. Let $B_{h}=B_{h}^{*}\in
\mathcal{L}(L^{2}(\rz^{d}))$ be such that $\|B_{h}\|=o(h)$\,. The ordered eigenvalues are denoted by
$\lambda_{j}(\alpha^{W,h}+B_{h})$ and $\lambda_{j}(\alpha_{0}^{W,h})$ for $j\in\nz$\,.
\begin{itemize}
\item
For $j=0$\,,
$\lambda_{0}(\alpha^{W,h}+B_{h})=
\Tr\left[\mathrm{Hess}~\alpha(0)\right]h+o(h)$
and the associated spectral projection satisfies
$$
\lim_{h\to
  0}1_{\{
\lambda_{0}(\alpha^{W,h}+B_{h})\}}(\alpha^{W,h}+B_{h})=(\pi^{-d}e^{-|TX|^{2}})^{W}(x,D_{x})\,,\quad\text{in}~\mathcal{L}^{1}(L^{2}(\rz^{d}))\,,
$$
where $T\in Sp_{2d}(\rz)$ is such that ${}^{t}X{}^{t}T^{-1}\mathrm{Hess}~\alpha(0)
T^{-1}X=\sum_{j=1}^{d}\beta_{j}X_{j}^{2}$\,.
\item
There exist $h_{0}>0$ and 
$C'\geq 1$ such that 
, 
for all $j>0$ and $h\in (0,h_{0})$\,,
\begin{equation}
  \label{eq:encaa0}
C'^{-1}hd/2\leq C'^{-1}\lambda_{j}(\alpha_{0}^{W,h})\leq
\lambda_{j}(\alpha^{W,h}+B_{h})-\lambda_{0}(\alpha^{W,h}+B_{h})\leq
 C' \lambda_{j}(\alpha_{0}^{W,h})\,.
\end{equation}
\end{itemize}
\end{lem}
\begin{remark}
  Of course $\sigma(\alpha_{0}^{W,h})=\left\{h\,(d/2+|n|)\,,
    n\in\nz^{d}\right\}$ and the bounds \eqref{eq:encaa0} are actually
  written in order to use this later. But for an easy use of the
  min-max principle it is better to write the eigenvalues
  $\lambda_{j}(\alpha_{0}^{W,h})$ in the increasing order, with
  repetition according to their multiplicity.
\end{remark}
\begin{proof}[Proof of Lemma~\ref{le:alpalp0}]
We start by noting that $1+\alpha\in S(\langle X\rangle^{2},
\frac{dX^{2}}{\langle X\rangle^{2}})$ is fully elliptic in the sense
that $(1+\alpha)^{-1}\in S(\langle X\rangle^{-2},\frac{dX^{2}}{\langle
  X\rangle^{2}})$\,. 
Therefore
$$
(1+\alpha)\sharp^{W,h}\frac{1}{1+\alpha}=1+h^{2}R_{+}(h)\quad,\quad
\frac{1}{1+\alpha}\sharp^{W,h}(1+\alpha)=1+h^{2}R_{-}(h)
$$
with $R_{\pm}(h)$ uniformly bounded in $S(\langle
X\rangle^{-2},\frac{dX^{2}}{\langle X\rangle^{2}})$\,. 
The semiclassical calculus with the
metric $\frac{dX^{2}}{\langle X\rangle^{2}}$\,, then says
\begin{equation}
  \label{eq:1+a}
(1+\alpha^{W,h})^{-1}=[(1+\alpha)^{-1}]^{W,h}+\mathcal{O}(h^{2})\quad
\text{in}~S\left(\langle X\rangle^{-2}\,, \frac{dX^{2}}{\langle
    X\rangle^{2}}\right)\,. 
\end{equation}
The same of course also holds for $\tau\alpha_{0}(X)=\tau \frac{|X|^{2}}{2}$ with $\tau\in
(0,+\infty)$ fixed. Therefore $\alpha^{W,h}+B_{h}$ and
$\alpha_{0}^{W,h}$ are self-adjoint with the same domain
$D(\alpha^{W,h})=D(\alpha_{0}^{W,h})=D(\alpha_{0}^{W,1})$\,, and they have a compact resolvent.
We shall collect all the necessary
information by comparing the eigenvalues of $\alpha^{W,h}+B_{h}$
and $\alpha_{0}^{W,h}$ in the intervals  $(-\infty,2|\beta|h]$\,,
$[0,2]$ and $[1,+\infty[$\,, with $|\beta|=\sum_{j=1}^{d}\beta_{j}$\,. For the first part, we refer to the
ready-made simple statement of \cite{ChVN}-Theorem~4.5
 and complete the other parts
with simple pseudodifferential
calculus and 
the min-max principle\,.\\
\noindent\textbf{Interval $(-\infty,2|\beta|h]$:} 
By Theorem~4.5 of \cite{ChVN}\,, there exist  a family of real numbers $(\omega_{n}^{h})_{h>0,n\in\mathbb N^d}$ and, for any $t>0$, a constant
$C_{t}>0$ such that 
\begin{eqnarray*}
  &&
\sigma(\alpha^{W,h})\cap (-\infty,th]
=\left\{\omega_{n}^{h}\,, n\in\nz^{d}\right\}\cap [|\beta|h/2,th]
\end{eqnarray*}
and
\begin{eqnarray*}
\Big|\omega_{n}^{h}-\sum_{j=1}^{d}h\beta_{j}(\frac{1}{2}+n_{j})\Big|\leq
   C_{t}h^{3/2}\,.
\end{eqnarray*}
As $\|B_{h}\|=o(h)$, the min-max principle with $\alpha^{W,h}$ and $\alpha^{W,h}+B_{h}$ then gives,
\begin{equation*}
\sigma(\alpha^{W,h}+B_{h})\cap(-\infty,th]=\left\{\omega_{n}^{h}+o(h)\,,
  \, n\in\nz\right\}\cap [0,th]\,.
\end{equation*}
By choosing $t=2|\beta|$\,, the operator $\alpha^{W,h}+B_{h}$
 is non-negative with
 $\lambda_{0}(\alpha^{W,h}+B_{h})=|\beta|h/2+o(h)$ and
the spectral gap is bounded from below by
\begin{eqnarray}
\nonumber
\forall  j\in \nz\setminus\left\{0\right\}\,,
\quad  \lambda_{j}(\alpha^{W,h}+B_{h})-\lambda_{0}(\alpha^{W,h}+B_{h})
&\geq& \lambda_{1}(\alpha^{W,h}+B_{h})-\lambda_{0}(\alpha^{W,h}+B_{h})
\\
\label{eq:0th}
&\geq&\beta_{m}h+o(h)\geq \beta_{m}h/2
\,,
\end{eqnarray}
with 
$\beta_{m}=\min\left\{\beta_{1},\ldots, \beta_{d}\right\}$\,.\\
Let $T\in Sp_{2d}(\rz^{d})$ be such that ${}^{t}X{}^{t}T^{-1}
\mathrm{Hess}~\alpha(0)T^{-1}X=\sum_{j=1}^{d}\beta_{j}X_{j}^{2}$\,,
let $U_{T}$ be a unitary operator such that
$U_{T}^{*}b^{W}U_{T}=b(T^{-1}.)^{W}$ and set $\varphi_{T}(x)=(\pi)^{-d/4}
U_{T}e^{-\frac{x^{2}}{2}}$\,. We compute 
\begin{align*}
\langle \varphi_{T}\,,\, (\alpha^{W,h}+B_{h})\varphi_{T}\rangle
&=\Tr\left[U_{T}^{*}\alpha^{W,h}U_{T}|\varphi_{\Id}\rangle\langle
  \varphi_{\Id}|\right]+o(h)
  \\
&=\int_{\rz^{2d}}\alpha(\sqrt{h}T^{-1}X)e^{-|X|^{2}}~\frac{dX}{\pi^{d}}+o(h)\,.
\end{align*}
But since $\alpha(T^{-1}X)=\sum_{j=1}^{d}\beta_{j}|X_{j}|^{2}/2
+P_{3}(X)+\mathcal{O}(|X|^{4})$\,, with
$P_{3}$ a homogeneous polynomial of degree 3, we obtain 
$$
\langle \varphi_{T}\,,\, (\alpha^{W,h}+B_{h})\varphi_{T}\rangle=h|\beta|/2+o(h)=\lambda_{0}(\alpha^{W,h}+B_{h})+o(h)\,.
$$
With the spectral gap \eqref{eq:0th} this implies that the ground
state $\psi_{0}^{h}$ of $\alpha^{W,h}+B_{h}$ satisfies 
$\lim_{h\to 0}\|\psi_{0}^{h}-\varphi_{T}\|_{L^{2}}=0$
and
$$
\lim_{h\to 0}\|
1_{\{\lambda_{0}(\alpha^{W,h}+B_{h})\}}(\alpha^{W,h}+B_{h})-
\pi^{-d}(e^{-|TX|^{2}})^{W,1}\|_{\mathcal{L}^{1}}=0\,.
$$
\\
\noindent\textbf{Interval $[0,2]$:}
Our assumptions on $\alpha$ provide a constant $C_{2}\geq 1$ such that 
$C_{2}^{-1}\alpha_{0}\leq \alpha\leq C_{2}\alpha_{0}$ and therefore 
$\frac{C_{2}^{-1}\alpha_{0}}{1+C_{2}^{-1}\alpha_{0}}\leq \frac{\alpha}{1+\alpha}\leq
\frac{C_{2}\alpha_{0}}{1+C_{2}\alpha_{0}}$\,, as $x\mapsto \frac{x}{1+x}$ is increasing on $\mathbb R^*$. 
Since all those symbols belong to $S(1,dX^{2})$\,, the semiclassical
Fefferman-Phong inequality for the constant metric $dX^{2}$ (see
\cite{Hor3}-Lemma~18.6.10) says
$$
\frac{C_{2}^{-1}\alpha_{0}^{W,h}}{1+C_{2}^{-1}\alpha_{0}^{W,h}}-\mathcal{O}(h^{2})\leq 
\frac{\alpha^{W,h}}{1+\alpha^{W,h}}\leq \frac{C_{2}\alpha_{0}^{W,h}}{1+C_{2}\alpha_{0}^{W,h}}+\mathcal{O}(h^{2})\,,
$$
after using $\left(\frac{\alpha_{.}}{1+\alpha_{.}}\right)^{W,h}=\frac{\alpha_{.}^{W,h}}{1+\alpha_{.}^{W,h}}+{\cal O}(h^{2})$\,.
With
$\|(1+\alpha^{W,h})^{-1}-(1+\alpha^{W,h}+B_{h})^{-1}\|=\mathcal{O}(\|B_{h}\|)=o(h)$
and $\frac{x}{1+x}=1-\frac{1}{1+x}$\,,
we deduce
$$
\frac{C_{2}^{-1}\alpha_{0}^{W,h}}{1+C_{2}^{-1}\alpha_{0}^{W,h}}-o(h)
\leq \frac{\alpha^{W,h}+B_{h}}{1+\alpha^{W,h}+B_{h}}
\leq \frac{C_{2}\alpha_{0}^{W,h}}{1+C_{2}\alpha_{0}^{W,h}}+o(h)\,. 
$$
For  $r=2(1+C_{2})$ and $h_{0}>0$ 
small enough the above operators have
a discrete spectrum in $[0,\frac{r}{1+r}]$ with eigenvalues in this
interval, while the function $x\mapsto \frac{x}{1+x}$ increases on
$[0,+\infty)$\,. 
Hence the min-max principle implies that there exists
$C_{2}'\geq 1$
such that
\begin{multline}
  \label{eq:02}
\left(
\lambda_{j}(\alpha^{W,h}+B_{h})\leq
2\right)
\\
\Rightarrow
\left(
C_{2}'^{-1}\lambda_{j}(\alpha_{0}^{W,h})-o(h)
\leq 
\lambda_{j}(\alpha^{W,h}+B_{h})\leq C_{2}'\lambda_j(\alpha_{0}^{W,h})+o(h)
\right)
\end{multline} 
holds for all $j\in\nz$\,.
With the spectral gap 
\eqref{eq:0th} and $\lambda_{0}(\alpha^{W,h}+B_{h})=|\beta|h/2+o(h)$
we conclude that \eqref{eq:encaa0} holds when
$\lambda_{j}(\alpha^{W,h}+B_{h})\leq 2$\,.\\
\noindent\textbf{Interval $[1,+\infty)$:}
Our assumptions on $\alpha$ provide a constant $C_{1}\geq 1$ such that 
$C_{1}^{-2}\leq\left(\frac{1+\alpha_{0}}{1+\alpha}\right)^{2}\leq C_{1}^{2}$\,. With
\eqref{eq:1+a}\,, the semiclassical Garding inequality then gives for
$h_{0}$ small enough: 
$$
\max\left\{\|(1+\alpha_{0}^{W,h})(1+\alpha^{W,h})^{-1}\|
\,, \|(1+\alpha^{W,h})(1+\alpha_{0}^{W,h})^{-1}\|\right\}\leq 2C_{1}
$$
Owing to $\|B_{h}\|=o(h)$ this is also true when $\alpha^{W,h}$ is
replaced by $\alpha^{W,h}+B_{h}$\,. We obtain
\begin{multline*}
\forall \psi\in D(\alpha_{0}^{W,1})\,,
\\
(2C_{1})^{-2}\langle \psi\,,\, (1+\alpha_{0}^{W,h})^{2}\psi \rangle
\leq \langle \psi\,, (1+\alpha^{W,h}+B_{h})^{2}\psi\rangle
\leq (2C_{1})^{2}\langle \psi\,,\, (1+\alpha_{0}^{W,h})^{2}\psi\rangle\,.
\end{multline*}
and the min-max principle gives
$$
\forall j\in \nz\,, (2C_{1})^{-2}\lambda_{j}((1+\alpha_{0}^{W,h})^{2})
\leq \lambda_{j}((1+\alpha^{W,h}+B_{h})^{2})
\leq (2C_{1})^{2}\lambda_{j}((1+\alpha^{W,h}_{0})^{2})\,.
$$
By taking the square roots
$$
\forall j\in \nz\,, (2C_{1})^{-1}(1+\lambda_{j}(\alpha_{0}^{W,h}))\leq 
1+\lambda_{j}(\alpha^{W,h}+B_{h})\leq 2C_{1}(1+\lambda_{j}(\alpha^{W,h}_{0}))\,.
$$
yields  \eqref{eq:encaa0} for $\lambda_{j}(\alpha^{W,h}+B_{h})\geq 1$\,.
\end{proof}

\begin{proof}[Proposition~\ref{pr:scsing}]
With $H=\alpha^{W,h}+B_{h}-\lambda_{0}(\alpha^{W,h}+B_h)$\,,
Lemma~\ref{le:alpalp0} provides a constant $C'>0$ such that
$$
\forall j\in \nz\setminus\left\{0\right\}\,,\quad
C'^{-1}\lambda_{j}(\alpha_{0}^{W,h})\leq \lambda_{j}(H)\leq 
C'\lambda_{j}(\alpha_{0}^{W,h})
$$
while $\lambda_{0}(H)=0$ and the ground state of $H$ is the same as
the one of $\alpha^{W,h}+B_{h}$\,.\\
When the function $f$ is non-negative and decaying, we deduce
\begin{multline}
Tr\left[f(H+ch^{d/\kappa_{0}})\right]
=
f(ch^{d/\kappa_{0}})
+\sum_{j=1}^{\infty} f(\lambda_{j}(H)+ch^{d/\kappa_0})
\\
\label{eq.inegfdec}
\leq  
f(ch^{d/\kappa_{0}})
+\sum_{j=1}^{\infty} f(\lambda_{j}(H))
\leq
f(ch^{d/\kappa_{0}})
+\sum_{\substack{n\in\nz^{d}\\ n\neq 0}}f(C_{4}^{-1}h|n|)\,,
\end{multline}
with $C_{4}=C_{3} \left(1+\frac{4|\beta|}{\beta_{m}}\right)$\,;
and for $R>0$\,,
$$
\Tr\left[f(H+ch^{d/\kappa_{0}})1_{[R,+\infty)}(H)\right]
=\sum_{\lambda_{j}(H)\geq R}f(\lambda_{j}(H)+ch^{d/\kappa_{0}})\leq
                     \sum_{\substack{n\in\nz^{d}\\ h|n|\geq R/(2C_{3})}}f(C_{4}^{-1}h|n|)\,.
$$
Apply \eqref{eq.inegfdec} first, with $f=s^{-\kappa_{0}}\langle
s\rangle^{-\kappa_{\infty}+\kappa_{0}}$:
$$
h^{d}\Tr\left[f(H+ch^{d/\kappa_{0}})\right]\leq
c^{-\kappa_{0}}+Ch^{d}\sum_{n\in\nz^{d}\,,\,n\neq
  0}(h|n|)^{-\kappa_{0}}\langle h|n|\rangle^{-\kappa_{\infty}+\kappa_{0}}
$$
after splitting the sum into $\sum_{h|n|\leq 1}$ and $\sum_{h|n|\geq
  1}$ and with $\#\left\{n\in\nz^{d}\,,
  |n|=m\right\}=C_{m+d-1}^{d-1}=\mathcal{O}(m^{d-1})$\,, it becomes
\begin{align*} 
h^{d}&\Tr\left[f(H+ch^{d/\kappa_{0}})\right]
\\
&\leq c^{-\kappa_{0}}+
C'h^{d}\sum_{m=1}^{\lceil h^{-1}\rceil}h^{-\kappa_{0}}m^{d-1-\kappa_{0}}
+
C'h^{d}\sum_{m=\lfloor h^{-1}\rfloor}^{\infty}h^{-\kappa_{\infty}}m^{d-1-\kappa_{\infty}}\,.
\\
&\leq
c^{-\kappa_{0}}+C''h^{d-\kappa_{0}}\lceil h^{-1}\rceil^{d-\kappa_{0}}+
C''h^{d-\kappa_{\infty}}\lfloor h^{-1}\rfloor^{d-\kappa_{\infty}}\leq c^{-\kappa_0}+C'''\,.
\end{align*}
owing to $\kappa_{\infty}>d$ and $\kappa_{0}\in (0,d)$\,.
With a function $f(s)=s^{-\kappa_{0}}\chi(s/\delta)$ with
$0\leq\chi\leq 1$ compactly supported and decaying on $[0,+\infty)$ we
get similarly
$$
\lim_{\delta\to 0^{+}}\limsup_{h\to 0}h^{d}\Tr\left[f(H+ch^{d/\kappa_{0}})\right]-c^{-\kappa_{0}}=0\,,
$$
while with $f(s)=\langle s\rangle^{-\kappa_{\infty}}$\,, the truncated
trace
$\Tr\left[f(H+ch^{d/\kappa_{0}})1_{[\delta^{-1},+\infty)}(H)\right]$ satisfies
$$
\lim_{\delta\to 0^{+}}\limsup_{h\to 0}
h^{d}\Tr\left[f(H+ch^{d/\kappa_{0}})1_{[\delta^{-1},+\infty}(H)\right]=0\,.
$$
The comparison of $\lambda_j(H)$ with $\lambda_j(\alpha_0^{W,h})$, $j\in\nz$, stated in Lemma~\ref{le:alpalp0} does not depend on the parameter $c$. Neither do the constants $C_3$, $C_4$, $C$, $C'$, $C''$ and $C'''$  ($f$ is non-negative and decaying) depend on $c$\,. Therefore the previous asymptotic trace estimates are uniform with respect to $c\in(\frac{1}{A},A)$ for any fixed $A>1$.

Thus if $\chi\in \mathcal{C}^{\infty}_{0}(\rz)$ is a cut-off function
such that $0\leq \chi \leq 1$\,, $\chi\equiv 1$ in $(-1,1)$ and  if a general $f\in
\mathcal{C}^{\infty}((0,+\infty))$ fulfills all the assumptions of Prop.~\ref{pr:scsing}, then 
\begin{equation}
\label{eq.elimextr}
\lim_{\delta\to 0^{+}}\limsup_{h\to
  0^{+}}h^{d}\|[f(H+ch^{d/\kappa_{0}})
  1_{(0,+\infty)}(H)[\chi(\delta^{-1}H)+(1-\chi(\delta H))]\|_{\mathcal{L}^{1}}=0\,.
\end{equation}
For $g\in \mathcal{C}^{\infty}_{0}(\rz)$\,, with an almost analytic
extension $\tilde{g}\in \mathcal{C}^{\infty}_{0}(\cz)$\,, 
Helffer-Sj{\"o}strand formula
$$
g(\alpha^{W,h})=\frac{1}{2i\pi}\int_{\cz}\partial_{\bar
  z}\tilde{g}(z)(z-\alpha^{W,h})~dz\wedge d\bar z\,,
$$
combined with the semiclassical Beals criterion \cite{DiSj,HeNi,NaNi} with the constant
metric $dX^{2}$ implies that
$$
g(\alpha^{W,h})-g(\alpha)^{W,h}=h \, r(h)^{W,h}\,,
$$
with $r(h)$ uniformly bounded (with respect to $h$) in $S(1,dX^2)$\,.
Since $(1+\alpha)\in S(\langle X\rangle^{2},\frac{dX^{2}}{\langle
  X\rangle^2})$ is an invertible elliptic symbol, 
$(1+\alpha^{W,h})^{-N}-[(1+\alpha)^{-N}]^{W,h}=h^2r^{\prime}(h)^{W,h}$ with $r'(h)$ uniformly bounded in
$S(\langle X\rangle^{-2N-2},\frac{dX^{2}}{\langle
  X\rangle^{2}})\subset S(\langle X\rangle^{-2N},dX^{2})$\,.
For a function $f_{\delta}\in
\mathcal{C}^{\infty}_{0}((0,+\infty))$\,, we take $g(s)=(1+s)^{N}f_{\delta}(s)$
and write
$$
f_{\delta}(\alpha^{W,h})=g(\alpha^{W,h})(1+\alpha^{W,h})^{-N}\,,
$$
so that
\begin{multline*}
f_{\delta}(\alpha^{W,h})-f_{\delta}(\alpha)^{W,h}=
\left[g(\alpha^{W,h})-g(\alpha)^{W,h}\right](1+\alpha^{W,h})^{-N}\\
+g(\alpha)^{W,h}(1+\alpha^{W,h})^{-N}-f_{\delta}(\alpha)^{W,h}
=h\,r''(h)^{W,h}\,,
\end{multline*}
with $r''(h)$ uniformly bounded in $S(\langle X\rangle^{-2N}\,,\, dX^{2})$\,. In particular, $h^d r''(h)^{W,h}$ is uniformly bounded in 
$\mathcal{L}^{1}(L^{2}(\rz^{d}))$ if we choose  $N>d$\,.\\
Similarly, Helffer-Sj{\"o}strand formula can be used to prove
$g(H+ch^{d/\kappa_{0}})-g(\alpha^{W,h})=o(h)$ in
$\mathcal{L}(L^{2}(\rz^{d}))$\,. With  
 $h^d\left[(1+ H+ch^{d/\kappa})^{-N}-(1+\alpha^{W,h})^{-N}\right]=o(h)$ in  
 $\mathcal{L}^{1}(L^{2}(\rz^{d}))$ due to
$$
(1+H+ch^{d/\kappa_{0}})^{-1}=\left[1+(1+\alpha^{W,h})^{-1}(B_{h}+ch^{d/\kappa_{0}})\right]^{-1}(1+\alpha^{W,h})^{-1} 
$$
 the same trick as above
transforms the  $\mathcal{L}(L^{2}(\rz^{d}))$ estimate
into 
\begin{equation}
\label{eq.gL1}
h^d\left[f_{\delta}(H+ch^{d/\kappa_{0}})-f_{\delta}(\alpha^{W,h})\right] =o(h)\quad \text{in}\;\mathcal{L}^1(L^{2}(\rz^{d}))
\end{equation}
Note again that this holds uniformly with respect to $c\in (\frac{1}{A},A)$ for any fixed $A>1$\,.

Now take 
 $f_\delta(s)=(1-\chi(\delta^{-2}s))\chi(\delta^2 s)f(s)$  for which we note that the inequality
 $$
\forall s\geq 0,\quad 1-(1-\chi(\delta^{-2}s))\chi(\delta^2 s)\leq
 	\chi(\delta^{-1}s) + (1-\chi(\delta s))
 $$
 as soon as $\delta<\delta_{chi}$ implies
\begin{equation}
\label{eq.majffd}
 	\forall s\geq 0,\quad
 	0\leq f(s)-f_\delta(s)\leq f(s) [\chi(\delta^{-1}s) + (1-\chi(\delta s))]\,.
\end{equation}
In the expression $h^d\Tr[f(H+ch^{d/\kappa_0}a^{W,h}]$\,,
 decompose $f(H+ch^{d/\kappa_0})$ into
 $$
 \underbrace{f_\delta(H+ch^{d/\kappa_0})}_{(I)}
 +\underbrace{(f(H+ch^{d/\kappa_0}-f_\delta(H+ch^{d/\kappa_0}))1_{(0,+\infty)}(H)}_{(II)}
 +\underbrace{1_{\{0\}}(H)f(ch^{d/\kappa_0})}_{(III)}\,.
 $$
 We now conclude with the following steps:
\begin{itemize}
\item The estimate \eqref{eq.gL1} yields
\begin{align*}
\lim_{h\to
  0}h^{d}\Tr\left[f_{\delta}(H+ch^{d/\kappa_{0}})a^{W,h}\right]
& =\lim_{h\to 0}h^{d}\Tr\left[f_{\delta}(\alpha)^{W,h}a^{W,h}\right]
\\
& =\int_{\rz^{2d}}f_{\delta}(\alpha(X))a(X)\frac{dX}{(2\pi)^{d}}\,,
\end{align*}
which provides the contribution of $(I)$\,.
\item The upperbound \eqref{eq.majffd} combined with \eqref{eq.elimextr} leads to
$$
\lim_{\delta\to 0^+}\limsup_{h\to
    0}\left|h^{d}\Tr[[f(H+ch^{d/\kappa_0})-f_{\delta}(H+ch^{d/\kappa_{0}})]1_{(0,+\infty)}(H)a^{W,h}]\right|=0\,,
$$
which says that $(II)$ has a null contribution in the limit $\delta\to 0$\,. 
\item The contribution of $(III)$ is simply computed as
$$
h^{d}\Tr\left[f(H+ch^{d/\kappa_{0}})1_{\left\{0\right\}}(H)a^{W,h}\right]=\frac{f_{0}}{c^{\kappa_{0}}}\langle
  \psi_{0}^{h}\,, a^{W,h}\psi_{0}^{h}\rangle
$$
where $\psi_{0}^{h}$ is the ground
  state of $H+ch^{d/\kappa_{0}}$ with
  $\|\psi^{h}-\pi^{-d/4}e^{-x^{2}/2}\|\stackrel{h\to
    0}{\to}0$\,. This implies $\lim_{h\to 0}\langle \psi^{h}\,,
  a^{W,h}\psi^{h}\rangle=a(0)$\,.
\item Finally, the assumptions on $f$ ensure $f(\alpha)\in L^{1}(\rz^{2d})$ and
$$
\lim_{\delta\to 0}\int_{\rz^{2d}}f_{\delta}(\alpha(X))a(X)~dX=\int_{\rz^{2d}}f(\alpha(X))a(X)~dX\,.
$$
\end{itemize}
\end{proof}

\subsubsection{Semiclassical analysis of the reduced density matrices in the bosonic case}
\label{sec:secBEC}
\begin{proof}[Proof of Proposition~\ref{pr:WigBEC}] This will be made
  in two parts: We first compute the semiclassical measures
  $\nu^{(p)}$ and  then identify the weak-$^{*}$ limit
  $\gamma_{0}^{(p)}$\,.\\
For the first part Proposition~\ref{pr:Phia} says that it suffices to
find the limit $\Phi_{a,0}(s)$ of $\Phi_{a,h}(s)$ for $a\in {\cal
  C}^{\infty}_{0}(T^{*}\rz^{d})$\,, real-valued, and $s\in
(-r_{a},r_{a})$\,. Actually Proposition~\ref{pr:scsing} allows to
consider more generally $a\in S(1,dX^{2})$\,.
 For $a\in S(1,dX^{2})$\,, real-valued,  take $s\in\rz$\,, $|s|<r_{a}=
\frac{1}{\nu_{C}C_{a}}$\,, $4\|a^{W,h}\|\leq C_{a}$ and set
 \begin{eqnarray*}
   &&DT_{a,h}(s)=\log\Tr\left[\varrho_{\varepsilon}\Gamma(e^{\varepsilon s a})\right]
=-\Tr\left[\log (1-CB_{s})\right]+\Tr\left[\log(1-C)\right]\\
\text{and}&&\Phi_{a,h}(s)=
\Tr\left[\varrho_{\varepsilon}\Gamma(e^{\varepsilon s a^{W,h}})\right]
=\exp DT_{a,h}(s)\,,\quad \varepsilon=h^{d}\,,
 \end{eqnarray*}
with $C=e^{-\beta(H+\frac{\varepsilon}{\beta \nu_{C}})}$ and
$B_{s}=e^{\varepsilon s a^{W,h}}$\,.\\
Assume $s\in(-r_{a},r_{a})$ and compute
\begin{align*}
DT_{a,h}(s)&=\int_{0}^{1}\Tr\left[
\frac{C_{ts}\tilde{B}_{ts}}{1-C_{ts}\tilde{B}_{ts}}
\varepsilon sa^{W,h}\right]~dt
\\
&=
\int_{0}^{1}\Tr\left[\varepsilon
  sf(H+\frac{\varepsilon}{\beta}(\nu_{C}^{-1}-tsa(0)))a^{W,h}\right]~dt
\\
&\quad+
\int_{0}^{1}\Tr\left[\varepsilon s [-(1-C_{ts})^{-1}+(1-C_{ts}\tilde{B}_{ts})^{-1}]a^{W,h}\right]dt
\end{align*} 
with
$C_{ts}=e^{-\beta(H+\frac{\varepsilon}{\beta}(\nu_{C}^{-1}-tsa(0)))}$\,,
$\tilde{B}_{ts}=e^{\varepsilon ts(a-a(0))^{W,h}}$ and
$f(u)=\frac{e^{-\beta u}}{1-e^{-\beta u}}$\,.
Note that for $t\in[0,1]$ the parameter $\frac{1}{\beta}(\nu_C^{-1}-tsa(0))$ remains in a compact subset of $(0,+\infty)$\,.
Prop.~\ref{pr:scsing} implies for all $t\in [0,1]$:
\begin{multline*}
\lim_{h\to 0}\Tr\left[\varepsilon
  sf(H+\frac{\varepsilon}{\beta}(\nu_{C}^{-1}-tsa(0)))a^{W,h}\right]
\\
=\frac{\nu_{C}sa(0)}{1-t\nu_{C}s a(0)}+s\int_{\rz^{2d}}\frac{e^{-\beta
    \alpha(X)}}{1-e^{-\beta \alpha(X)}}a(X)~\frac{dX}{(2\pi)^{d}}\,.
\end{multline*}
With the uniform control with respect to  $\frac{1}{\beta}(\nu_C^{-1}-tsa(0))=c\in[\frac{1}{A},A]$
in Proposition~\ref{pr:scsing}, we obtain for the first term 
\begin{multline*}
\lim_{h\to 0}\int_{0}^{1}\Tr\left[\varepsilon
  sf(H+\frac{\varepsilon}{\beta}(\nu_{C}^{-1}-ts
  a(0)))a^{W,h}\right]~dt
  \\
=-\log(1-s\nu_{C}a(0))
+s\int_{\rz^{2d}}\frac{e^{-\beta
    \alpha(X)}}{1-e^{-\beta \alpha(X)}}a(X)~\frac{dX}{(2\pi)^{d}}\,.
\end{multline*}
For the remainder term, introduce $\Pi_{0}^{h}=|\psi_{0}^{h}\rangle
\langle \psi_{0}^{h}|$\,, where $\psi_{0}^{h}=U_{T}(\pi^{-d/4}e^{-\frac{x^{2}}{2}})+o(h^{0})$ is the ground state of
$H$\,, and write 
\begin{align*}
  (&1-C_{ts}\tilde{B}_{ts})
  \\
  &=1-C_{ts} -C_{ts}(\tilde{B}_{ts}-1)
=(1-C_{ts})\left[1+\frac{C_{ts}}{1-C_{ts}}(1-\tilde{B}_{ts})\right]
\\
&=
(1-C_{ts})
\left[
1+\frac{C_{ts}}{1-C_{ts}}\Pi_{0}^{h}(1-\tilde{B}_{ts})
+\frac{C_{ts}}{1-C_{ts}}(1-\Pi_{0}^{h})(1-\tilde{B}_{ts})
\right]
\\
&=(1-C_{ts})\left[
1+\underbrace{f\left(\frac{\varepsilon}{\beta}(\nu_{C}^{-1}-tsa(0))\right)\Pi_{0}^{h}(1-\tilde{B}_{ts})}_{I}
+\underbrace{\frac{C_{ts}}{1-C_{ts}}(1-\Pi_{0}^{h})(1-\tilde{B}_{ts})}_{II}
\right]\,.
\end{align*} 
We know
$$
\varepsilon\times
f\left(\frac{\varepsilon}{\beta}(\nu_{C}^{-1}-tsa(0))\right)=\frac{1}{\nu_{C}^{-1}-tsa(0)}+O(\varepsilon)=\frac{1}{\nu_{C}^{-1}-tsa(0)}+o(h) \,.
$$
We write 
$$
\varepsilon^{-1}(1-\tilde{B}_{ts})\psi_{0}^{h}
=
-\int_{0}^{1}e^{\varepsilon uts (a-a(0))^{W,h}}ts (a-a(0))^{W,h}\psi_{0}^{h}~du\,,
$$
where $\psi_{0}^{h}=\pi^{-d/4}
U_{T}e^{-\frac{x^{2}}{2}}+o(h^{0})$ and $a(X)-a(0)\leq C\min\{1,|X|\}$ for some $C>0$ 
imply $\lim_{h\to
  0}\|(a-a(0))^{W,h}\psi_{0}^{h}\|_{L^{2}(\rz^{d})}=0$\,.
Therefore the second term in the above bracket satisfies
$$
I=f\left(\frac{\varepsilon}{\beta}(\nu_{C}^{-1}-tsa(0))\right)\Pi_{0}^{h}(1-\tilde{B}_{ts})=o(h^{0})\quad\text{in}
~\mathcal{L}^{1}(L^{2}(\rz^{d}))\,.
$$
Note that we have also proved
$$
(1-\tilde{B}_{ts})\Pi_{0}^{h}-\Pi_{0}^{h}(1-\tilde{B}_{ts})=o(\varepsilon)\quad\text{in}~{\cal L}(L^{2}(\rz^{d})) \,.
$$
By using
$$
\|1-\tilde{B}_{ts}\|=\mathcal{O}(\varepsilon)\quad,\quad
\left\|\frac{C_{ts}}{1-C_{ts}}(1-\Pi_{0}^{h})\right\|=\mathcal{O}(\frac{1}{h})
\,,
$$
and
\begin{align*}
\lim_{h\to 0}\|\varepsilon \frac{C_{ts}}{1-C_{ts}}\|_{{\cal L}^{1}}
&=
\lim_{h\to 0}\Tr\left[\varepsilon
  f(H+\frac{\varepsilon}{\beta}(\nu_{C}^{-1}-tsa(0)))\right]
\\
&=\frac{\nu_{C}}{1-t\nu_{C}s a(0)}+s\int_{\rz^{2d}}\frac{e^{-\beta
    \alpha(X)}}{1-e^{-\beta \alpha(X)}}~\frac{dX}{(2\pi)^{d}}\,,
\end{align*}
the third term in the above bracket 
satisfies
\begin{eqnarray*}
  &&\|II\|_{{\cal L}^{1}}={\cal O}(1)
\quad,\quad
\|II\|={\cal O}(\varepsilon h^{-1})=o(h^{0})\,,\\
&&
   \|II-\frac{C_{ts}}{1-C_{ts}}(1-\Pi_{0}^{h})(1-\tilde{B}_{ts})(1-\Pi_{0}^{h})
\|_{{\cal L}^{1}}=o(h^{0})\,.
\end{eqnarray*}
Again all these estimates are uniform with respect to $t\in[0,1]$ owing to the uniformity of the estimates in Proposition~ \ref{pr:scsing} with respect to $c=\frac{1}{\beta}(\nu_C^{-1}-tsa(0))$\,.
By expanding the Neumann series
$(1+I+II)^{-1}=\sum_{k=0}^{\infty}(-1)^{k}(I+II)^{k}$ we deduce
$$
 \left[1+\frac{C_{ts}}{1-C_{ts}}(1-\tilde{B}_{ts})\right]^{-1}
=
1-\frac{C_{ts}}{1-C_{ts}}(1-\Pi_{0}^{h})(1-\tilde{B}_{ts})(1-\Pi_{0}^{h})
+R_{h}
$$
with $\|R_{h}\|_{{\cal L}^{1}}=o(h^{0})$\,.
With $\|\varepsilon(1-C_{ts})^{-1}\|={\cal O}(1)$ we finally obtain
\begin{multline*} 
\varepsilon s [(1-C_{ts})^{-1}-(1-C_{ts}\tilde{B}_{ts})^{-1}]
\\
=\frac{s\varepsilon  C_{ts}}{1-C_{ts}}(1-\Pi_{0}^{h})
(1-\tilde{B}_{ts})(1-\Pi_{0}^{h})(1-C_{ts})^{-1}+R'_{h}\,,\qquad
 \|R'_{h}\|_{{\cal L}^{1}}=o(h^0)\,,
\end{multline*}
while $\|\varepsilon C_{ts}(1-C_{ts})^{-1}\|_{{\cal L}^{1}}={\cal
  O}(1)$\,, $\|1-\tilde{B}\|={\cal O}(\varepsilon)$ and $\|(1-\Pi_{0}^{h})(1-C_{ts})^{-1}\|={\cal O}(h^{-1})$\,.\\
With $4\|a^{W,h}\|\leq C_{a}$\,, the remainder term tends to $0$ as
$h\to 0$ and we have proved
\begin{align*}
\forall s\in (-r_{a},r_{a})\,,\quad \lim_{h\to 0}\Phi_{a,h}(s)
&=
\Phi_{a,0}(s)
\\
&=\frac{1}{1-s\nu_C a(0)}\exp\left[s\int_{\rz^{2d}}\frac{e^{-\beta\alpha(X)}}{1-e^{-\beta
    \alpha(X)}}a(X)~\frac{dX}{(2\pi)^{d}}\right]\,.
\end{align*}
By expanding the generating function  according to
Proposition~\ref{pr:Phia}, we obtain
$$
\lim_{h\to 0}\Tr\left[\varrho_{\varepsilon}((a^{W,h})^{\otimes
    p})^{Wick}\right]
=\sum_{k=0}^{p}\frac{1}{(p-k)!}\nu_{C}^{k}a(0)^{k}\int_{\rz^{2d(p-k)}}a^{\otimes (p-k)}
d\nu(\beta)^{\otimes p-k}\,.
$$
with $d\nu(\beta) = \frac{e^{-\beta \alpha(X)}}{1-e^{-\beta \alpha (X)}}\frac{dX}{(2\pi)^{d}}$\,.  The possibility to take $a\in
S(1,dX^{2})$ contains the fact that  our quantization is adapted to
all the $\gamma^{(p)}_{h}$\,.\\
Now in order to identify the weak$^{*}$ limits of the
$\gamma^{(p)}_{h}$ we compute the Wigner measure associated with
$\varrho_{\varepsilon(h)}$\,. Remember (see \eqref{eq:Weyl} and \eqref{eq:WignerMeasure})
$$
\Tr\left[\varrho_{\varepsilon}W(\sqrt{2}\pi f)\right]=\exp\left[
-\frac{\varepsilon\pi^{2}}{2}\langle f\,,\, \frac{1+ e^{-\beta (H+\frac{\varepsilon}{\beta
       \nu_{C}})}}{1- e^{-\beta (H+\frac{\varepsilon}{\beta
       \nu_{C}})}} f\rangle\,.
\right]
$$
By using the orthonormal basis of eigenvectors $(\psi_{j}^{h})_{j\in
  \nz}$ of $H$ with associated eigenvalues $\lambda_{j}^{h}$\,,
$\lambda_{0}^{h}=0$\,, $\lambda_{j}^{h}\geq ch$ for $j>0$\,, we obtain
$$
\log\left(\Tr\left[\varrho_{\varepsilon(h)}W(\sqrt{2}\pi
    f)\right]\right)
=-\pi^{2}\nu_{C}|\langle f,\psi_{0}^{h}\rangle|^{2}+{\cal
  O}(\varepsilon h^{-1})\,.
$$
With
$\|\psi_{0}^{h}-\psi_{0}\|_{L^{2}}=o(h)$\,, $\psi_{0}(x)=\pi^{-d/4}U_{T}e^{-\frac{x^{2}}{2}}$\,,
we obtain after decomposing
$f=f_{0}\psi_{0}\oplus^{\perp}f'$ 
$$
\int_{L^{2}}e^{2i\pi \Real \langle f,z\rangle}~d\mu(z)
=\lim_{h\to 0}\Tr\left[\varrho_{\varepsilon(h)}W(\sqrt{2}\pi f)\right]
=e^{-\pi^{2}\nu_{C}|f_{0}|^{2}}\,.
$$
We deduce, like in \cite[Section 7.5]{AmNi1} or \cite[Section 4.4]{AmNi3},
\begin{eqnarray*}
  &&
{\cal M}(\varrho_{\varepsilon(h)}, h\in {\cal
  E})=\left\{\left(\frac{e^{-\frac{|z_{0}|^{2}}{\nu_{C}}}}{\pi
    \nu_{C}}L(dz_{0})\right)\otimes \delta_{0}(z')\right\}
\quad (z=z_{0}\psi_{0}\oplus^{\perp}z')\,,
\\
\text{and}&&
\gamma_{0}^{(p)}=p!\nu_{C}^{p}|\psi_{0}^{\otimes p}\rangle\langle
             \psi_{0}^{\otimes p}|\,, \forall p\in\nz\,.
\end{eqnarray*}
The fact that $\nu_{(I)}^{(p)}\equiv 0$ for all $p\in \nz$\,, now comes
from
$$
\Tr\left[\gamma_{0}^{(p)}\right]=p!\nu_{C}^{p}=\nu^{(p)}(\left\{0\right\})\,.
$$
\end{proof}

\appendix
\section{Multiscale Measures}\label{sec:multiscale-measures}
We now recall facts about multiscale measures, introduced in \cite{FeGe,Fer}. 
For this we need a new class of symbols. Let $D',D'',D'''\in\nz$ be such
that $D'+D''+D'''= D$ and set
$F=\left\{X=(x',x'',x''',\xi',\xi'',\xi''')\in\rz^{2D}\,,\, x'=0\,,\, x''=\xi''=0\right\}$\,. The class 
of symbols $S^{(2)}_{F}$ is defined as the set of $(X,Y)\to a(X,Y)
\in {\cal C}^{\infty}(\rz^{2D}\times \rz^{D'+2D''})$ , (note that $\rz^{D'+2D''}\cong F^\perp$, hence the notation $S^{(2)}_{F}$) 
such that 
\begin{itemize}
\item there exists $C>0$ such that: $\forall Y\in \rz^{D'+2D''}$\,,
$a(\cdot,Y)\in {\cal C}^{\infty}_{0}(B(0,C))$\,;
\item there exists  a function $a_{\infty}\in {\cal
  C}^{\infty}_{0}(\rz^{2D}\times \sz^{D'+2D''-1})$ 
such that
$a(X,R\omega)\stackrel{R\to \infty}{\to}a_{\infty}(X,\omega)$ in
${\cal C}^{\infty}(\rz^{2D}\times \sz^{D'+2D''-1})$\,.
\end{itemize}
Those symbols are quantized according to 
$$
a^{(2),h}=a_{h}^{W,h}\quad, \quad
a_{h}(X)=a(X,\frac{x'}{h^{1/2}},\frac{X''}{h^{1/2}})\quad X=(x',x'',x''',\xi',\xi'',\xi''')\,.
$$
Theorem~0.1 in \cite{Fer}, 
which also considers the case when
$(\frac{x'}{h^{1/2}},\frac{X''}{h^{1/2}})$ is replaced by $(\frac{x'}{h^{s}},\frac{X''}{h^{s}})$\,, $s<\frac{1}{2}$, says the following.
\begin{prop}
 \label{pr:multisc}
Let $(\gamma_{h})_{h\in {\cal E}}$ be a bounded family of non-negative
trace-class operators on $L^{2}(\rz^{2D})$ with
$\lim_{h\to 0}\Tr\left[\gamma_{h}\right]<+\infty$\,. There exist ${\cal E}'\subset {\cal E}$\,, $0\in \overline{{\cal E}'}$\,, with 
${\cal M}(\gamma_{h}, h\in {\cal E}') =\left\{\nu\right\}$\,, 
 a non-negative measure
$\nu_{(I)}$ on $F\times
\sz^{D'+2D''-1}$ and a ${\cal L}^{1}(L^{2}(\rz^{2D''}))$-measure $m$ on
$F\times{\rz^{D'}}$\,,  such that the convergence
\begin{multline*}
\lim_{h\in {\cal E}'\,,\, h\to 0}\Tr\left[\gamma_{h}a^{(2),h}\right]
=\int_{\rz^{2D}\setminus F}
a_{\infty}(X,\frac{(x',X'')}{|(x',X'')|})~d\nu(X)
\\
+\int_{F\times
  \sz^{D'+2D''-1}}a_{\infty}(X,\omega)~d\nu_{(I)}(X,\omega)
+\Tr\left[\int_{F\times\rz^{D'}}a(X,x',z,D_{z})dm(X,x')\right]\,
\end{multline*}
holds for all $a\in S^{(2)}_{F}$\,.
\end{prop}

\begin{remark}
With this scaling and when 
$a^{W,h}=a^{W}(x,hD_{x})=a(x,hD_{x})+O(h)$\,, $t=0$\,, 
Fermanian checked in \cite{Fer} the equivariance by the
  semiclassical Egorov theorem. Hence, this construction is naturally
  extended  to the
  case when $T^{*}\rz^{D}$ is replaced by $T^{*}M$ and $F$ is replaced
by a submanifold of $T^{*}\rz^{D}$ on which the symplectic form has constant rank.
\end{remark}
In Prop.~\ref{pr:multisc_simplified} we use the simple case of the above result when
$D'=D'''=0$ and $D''=D$\,. Note that in this case 
$F\times \rz^{D'}=\left\{0\right\}$ 
and the trace-class-valued measure is nothing but a trace-class operator $\gamma_{0}$\,.

\section{Mean-field Wigner measures in the bosonic case and condition (PI)}
\label{sec:PI}
The bosonic mean-field analysis is like a semiclassical analysis in infinite dimension. Let $\Z$ be a separable
 complex Hilbert space and $\Gamma_{+}(\Z)$ be the associated bosonic
 Fock space.  With the scaled CCR relations
$$
[a_{+}(g),a_{+}^{*}(f)]=\varepsilon\langle g,f\rangle\quad,\quad
[a_{+}(g),a_{+}(f)]=[a_{+}^{*}(g),a_{+}^{*}(f)]=0
$$
and after setting 
\begin{equation}\label{eq:Weyl}
\Phi(f)=\frac{a_+(f)+a_+^{*}(f)}{\sqrt{2}}\quad,\quad W(f)=e^{i\Phi(h)}\,,
\end{equation}
mean-field Wigner measures where introduced in \cite{AmNi1}. Actually
the parameter $\varepsilon^{-1}$ represents the typical number of
particles. 
Let $(\varrho_{\varepsilon})_{\varepsilon \in {\cal E}}$\,, $0\in \overline{\mathcal{E}}$\,, be a family of normal states (normalized non-negative trace-class operators) in $\Gamma_{+}(\Z)$\,.
Under the sole uniform estimate
$\Tr\left[\varrho_{\varepsilon}(1+\mathbf{N})^{\delta}\right]\leq
C_{\delta}$ for some $\delta>0$\,, Wigner measures are defined
as Borel probability measures on $\Z$ and characterized by their
characteristic function as follows: $\mu\in {\cal
  M}(\varrho_{\varepsilon}, \varepsilon\in {\cal E})$\,, iff there
exists ${\cal E}'\subset{\cal E}$\,, $0\in \overline{{\cal E}'}$\,,
such that
\begin{equation}\label{eq:WignerMeasure}
\forall f\in \Z\,,\quad \lim_{\substack{\varepsilon\in {\cal E}'\\
    \varepsilon\to 0}}\Tr\left[\varrho_{\varepsilon}W(\sqrt{2}\pi
    f)\right]
=\int_{\Z}e^{2i\pi \Real\langle f,z\rangle}~d\mu(z)\,.
\end{equation}
Assuming $\Tr\left[\varrho_{\varepsilon}\mathbf{N}_{+}^{k}\right]\leq
C^{k}$ for all $k\in \nz$ (or as we do in Hypothesis~\ref{hyp:borne},
$\Tr\left[\varrho_{\varepsilon}e^{c\mathbf{N}_{+}}\right]\leq C$)\,,
${\cal M}(\varrho_{\varepsilon}\,, \varepsilon\in {\cal
  E})=\left\{\mu\right\}$ implies that
\begin{equation}
  \label{eq:bLinf}  
\lim_{\substack{\varepsilon\to 0}}
\Tr\left[\varrho_{\varepsilon}
\tilde{b}^{Wick}\right]
=\int_{\Z}\langle z^{\otimes p}\,,\, \tilde{b}z^{\otimes p}\rangle~d\mu(z)\,
\end{equation}
holds for all \underline{compact} $\tilde{b}\in {\cal L}^{\infty}({\cal S}_{+}^{p}\Z^{\otimes p})$\,. 
In particular with the definition of non normalized reduced density
matrices we obtain
$$
\forall p\in \nz,\quad \mathop{w^{*}-lim}_{\varepsilon\to
  0}\gamma_{\varepsilon}^{(p)}=\gamma_{0}^{(p)}=\int_{\Z}|z^{\otimes
  p}\rangle \langle z^{\otimes p}|~d\mu(z)\,.
$$
This $w^{*}$-limit can be transformed to a
$\|~\|_{{\cal L}^{1}}$ iff the restriction to compact $\tilde{b}$ in 
(\ref{eq:bLinf}) can be removed. 
It actually suffices to check that \eqref{eq:bLinf} holds for $\tilde{b}\in {\cal L}^{\infty}({\cal S}_{+}^{p}\Z^{\otimes p})$ and $\tilde{b}=\Id_{{\cal S}^{p}_{+}\Z^{\otimes p}}$\,, as
shows the following result.
\begin{prop}
  \label{pr.eqPPI}
For a family $(\varrho_{\varepsilon})_{\varepsilon\in
  \mathcal E}$ in  $\L^{1}(\H)$\,, $0\in\overline{\mathcal E}$\,, such that
$\varrho_{\varepsilon}\geq 0$, $\Tr[\varrho_{\varepsilon}]=1$, $\mathcal{M}(\varrho_{\varepsilon},
\varepsilon\in \mathcal E)=\left\{\mu\right\}$, the conditions
$(PI)$ and $(P)$ are equivalent:
\begin{multline*}
\left(
(PI):\quad
\forall \alpha\in\nz,~
  \lim_{\varepsilon\to 0}\Tr\left[\varrho_{\varepsilon}{\bf N}^{\alpha}\right]
=\int_\Z |z|^{2\alpha}~d\mu(z)<\infty\right)
\\
\Leftrightarrow
\left(
(P):\quad
\forall b\in \P_{alg}(\Z),~
\lim_{\varepsilon\to
  0}\Tr\left[\varrho_{\varepsilon}b^{Wick}\right]=\int_\Z b~d\mu\right)\,,
\end{multline*}
where ${\cal P}_{p,q}(\Z)=\left\{b:\Z\ni z\mapsto b(z)=\langle z^{\otimes q}\,,\,
  \tilde{b}z^{\otimes p}\rangle\in \cz\,,\; \tilde{b}\in {\cal
    L}({\cal S}_{+}^{p}\Z^{\otimes p};{\cal S}_{+}^q\Z^{\otimes
    q})\right\}$\,, and 
${\cal P}_{alg}(\Z)=\oplus_{p,q\in\nz}^{alg}{\cal
  P}_{p,q }(\Z)$\,.
\end{prop}
We give below the proof, which rectifies a minor mistake in \cite{AmNi3}.

\begin{proof}
For $\alpha\in \nz^{*}$, $(|z|^{2\alpha})^{Wick}={\bf N}({\bf N}-\varepsilon)\ldots({\bf
  N}-(\alpha-1)\varepsilon)$. 
Hence the condition $(PI)$
is equivalent to
$$(PI)':\qquad
\forall \alpha\in\nz,\quad
\lim_{\varepsilon\to 0}\Tr\left[\varrho_{\varepsilon}(|z|^{2\alpha})^{Wick}\right]=\int_\Z |z|^{2\alpha}~d\mu(z)<\infty \,.
$$
Hence the condition $(PI)$ is a particular case of $(P)$ and it is sufficient to prove
$(PI)'\Rightarrow (P)$. From now, assume $(PI)'$\,. 

We want to prove $(P)$ for a general $b\in {\cal P}_{alg}(\Z)=\oplus_{p,q\in\nz}^{alg}{\cal
  P}_{p,q }(\Z)$\,. 
Let us first consider the ``diagonal'' case $b\in \P_{p,p}(\Z)$\,,  $p\in\nz^{*}$\,. 
Using the decomposition
$\tilde{b}=\tilde{b}_{R,+}-\tilde{b}_{R,-}+i\tilde{b}_{I,+}-i\tilde{b}_{I,-}$
with all the $\tilde{b}_{\bullet}\geq 0$\, we can assume  $\tilde{b}\geq 0$. For such a $\tilde{b}$\,, there exists a non-decreasing sequence $(\tilde{b}_{n})_{n\geq0}$ of non-negative compact operators in 
$\L^{\infty}({\cal S}_{+}^{p}\Z)^{\otimes p}$ such that 
$\lim_{n\to\infty}\tilde{b}_{n}=\tilde{b}$ in the weak operator topology. Recall from \cite[Prop.~2.9]{AmNi3} that the convergence in the $(P)$ condition always holds when the kernel $\tilde{b}$ is compact, thus
\[
\forall n\in\nz\,,\quad \int_\Z  b_{n} \, \, d\mu
=\lim_{\varepsilon\to0}\Tr[\varrho_{\varepsilon} \, b_{n}^{Wick}]
\leq\liminf_{\varepsilon\to0}\Tr[\varrho_{\varepsilon}\,b^{Wick}]\,.
\]
Using $b_n(z)=\langle z^{\otimes p},\tilde{b}_n z^{\otimes p}\rangle \to \langle z^{\otimes p},\tilde{b} z^{\otimes p}\rangle=b(z)$ as $n\to\infty$ and Fatou's lemma yield
\begin{equation}\label{eq:int_b_leq_liminf_Tr_rhob}
\int_\Z b\,d\mu
\leq
\liminf_{\varepsilon\to0}\Tr[\varrho_{\varepsilon}b^{Wick}] \, .
\end{equation}
The same arguments with $\tilde{b}$ replaced by
$|b|_{\P_{p,p}}\Id_{{\cal S}_{+}^{p}\Z^{\otimes p}}-\tilde{b}\geq0$
provide
\[
\liminf_{\varepsilon\to0}\Tr[\varrho_{\varepsilon}(|\tilde{b}|_{{\cal
    S}_{+}^{p}\mathcal{Z}^{\otimes
    p}}|z|^{2p}-b(z))^{Wick}]\geq\int(|\tilde{b}|_{{\cal
    S}_{+}^{p}\mathcal{Z}^{\otimes p}}|z|^{2p}-b(z))d\mu(z)\,.
\]
With $(PI)'$ condition, the $|z|^{2p}$ terms can be removed on both sides and thus
\begin{equation}\label{eq:limsup_leq_int_b}
\limsup_{\varepsilon\to0}\Tr[\varrho_{\varepsilon}b^{Wick}]\leq\int_\Z  b \, d\mu \, .
\end{equation}
The inequalities \eqref{eq:int_b_leq_liminf_Tr_rhob} and \eqref{eq:limsup_leq_int_b}
show that the convergence in the $(P)$ condition holds for all $b\in \P_{p,p}(\Z)$ such that $\tilde{b}\geq 0$, and hence for all $b\in \P_{p,p}(\Z)$.

We now consider the general case $b\in \P_{p,q}(\Z)$. There exists a sequence of compact operators $\tilde{b}_{n}\in
\L^{\infty}({\cal S}_{°}^{p}\Z^{\otimes p},{\cal S}_{+}^{q}\Z^{\otimes
q})$ such that:
\begin{eqnarray*}
&&\forall n\in\nz\,, \quad 
|b_{n}|_{\P_{p,q}}=|\tilde{b}_{n}|_{\L({\cal S}_{+}^{p}\Z^{\otimes
    p},{\cal S}_{+}^{q}\Z^{\otimes q})}
\leq|\tilde{b}|_{\L({\cal S}_{+}^{p}\Z^{\otimes p},{\cal
    S}_{+}^{q}\Z^{\otimes q})}=|b|_{\P_{p,q}}\\
\text{and}&&   \forall z\in\Z\,,\quad
\lim_{n\to\infty}b_{n}(z)=\lim_{n\to\infty}\langle z^{\otimes q}\,,\,\tilde{b}_{n}z^{\otimes p}\rangle =\langle z^{\otimes q}\,,\,\tilde{b}z^{\otimes p}\rangle=b(z)\,.
\end{eqnarray*}
For any fixed $n\in\nz$\,,
\begin{multline}\label{eq:Tr_rho_b-int_b}
\limsup_{\varepsilon\to  0} \left|\Tr\left[ \varrho_{\varepsilon}b^{Wick}\right]-\int_{\Z}b(z)~d\mu(z)\right|
\\
\leq \limsup_{\varepsilon\to  0}\left|\Tr\left[ \varrho_{\varepsilon}(b^{Wick}-b_n^{Wick})\right]\right|
+\limsup_{\varepsilon\to 0}\left| \Tr\left[ \varrho_{\varepsilon}b_{n}^{Wick}\right]
-\int_{\Z}b_{n}~d\mu\right|
+\int_{\Z}|b_n-b|~d\mu \,,
\end{multline}
where the second term of the right-hand side vanishes because $\tilde{b}_n$ is a fixed compact operator.
Using the Cauchy-Schwarz inequality with $\Tr[\varrho_\varepsilon]=1$ gives
$$
\left|\Tr\left[\varrho_{\varepsilon}(b^{Wick}-b_{n}^{Wick})\right]\right|
\leq
\Tr\left[\varrho_{\varepsilon}(b^{Wick}-b_{n}^{Wick})(b^{Wick,*}-b_{n}^{Wick,*})\right]^{1/2}\,.
$$
From the proved result when $p=q$\,, we deduce:
\begin{equation}\label{eq:control_Tr_rho_b-bb_leq_int_b-bn_square}
\limsup_{\varepsilon\to  0}
\left|\Tr\left[\varrho_{\varepsilon}(b^{Wick}-b_{n}^{Wick})\right]\right|
\leq
\left[\int_{\Z}
\left|b-b_{n}\right|^{2}
~d\mu(z)\right]^{1/2}\,.
\end{equation}
With $\int_{\Z}|z|^{r(p+q)}~d\mu(z)<\infty$ and
\begin{equation*}
  \forall n\in\nz\,,\,\forall z\in\Z\,,\quad |b(z)-b_{n}(z)|^{r}\leq
  (2|b|_{\P_{p,q}})^{r}|z|^{r(p+q)}\,,
\end{equation*}
Lebesgue's convergence theorem yields
\begin{equation}\label{eq:convergence_int_b-bn}
\lim_{n\to\infty}\int_{\Z}|b-b_{n}|^{r}~d\mu=0\,
\end{equation}
for $r\in\{1,2\}$\,.
Combining \eqref{eq:Tr_rho_b-int_b}, \eqref{eq:control_Tr_rho_b-bb_leq_int_b-bn_square} and \eqref{eq:convergence_int_b-bn} proves $(P)$ for any $b\in \P_{p,q}(\Z)$\,.
\end{proof}

\section{The Composition Formula of  Wick Quantized Operators}\label{sec:Proof_Composition_Wick}

We give an algebraic proof for the composition formula~\eqref{eq:Composition_Wick}
of two Wick quantized operators on a
finite or infinite dimensional separable complex Hilbert space $\Z$\,. This
proof holds in both the bosonic and fermionic cases. It uses
only the definition of the Wick quantization, and it involves neither
creation and annihilation operators, nor the canonical commutation
or anticommutation relations.

We note $\llbracket m,n\rrbracket:=\{m,\dots,n\}$ for $m\leq n\in\mathbb{N}$.
The action of the symmetric group  $\mathfrak{S}_{\llbracket1,n\rrbracket}$
on product vectors in $\hil^{\otimes n}$, $\sigma\cdot (z_{1}\otimes\dots\otimes z_{n})=z_{\sigma_{1}}\otimes\dots\otimes z_{\sigma_{n}}$\,,
$z_{j}\in\hil$\,, is extended to $\hil^{\otimes n}$ by linearity and density. With
this notation $\mathcal{S}_{\pm}^{n}=\frac{1}{n!}\sum_{\mathfrak{S}_{\llbracket 1, n\rrbracket}}s_{\pm}(\sigma)\,\sigma\cdot$\,.

We begin with a preliminary lemma on a special set of permutations.
\begin{lem}
\label{lem:Sigma(k,p,q,K)}Let $k,p,q,K\in\mathbb{N}$ such that $k\in\llbracket\max\{0,p+q-K\},\min\{p,q\}\rrbracket$\,,
and
\[
\mathfrak{S}(k):=\Big\{\sigma\in\mathfrak{S}_{\llbracket1,K\rrbracket}\;\Big|\;\card\big(\sigma(\llbracket p-k+1,p-k+q\rrbracket)\cap\llbracket1,p\rrbracket\big)=k\Big\}\,.
\]

\begin{enumerate}
\item The cardinal of $\mathfrak{S}(k)$ is $\card\mathfrak{S}(k)=\binom{q}{k}\binom{p}{k}k!\frac{(K-q)!\,(K-p)!}{(K-(q+p-k))!}$\,.
\item Any permutation $\sigma\in\mathfrak{S}(k)$ can be factorized as $\sigma=\sigma^{(1)}\sigma^{(2)}\sigma^{(3)}\sigma^{(4)}$,
where $\sigma^{(1)}\in\mathfrak{S}_{\llbracket1,p\rrbracket}$\,, $\sigma^{(2)}\in\mathfrak{S}_{\llbracket p+1,K\rrbracket}$\,,
$\sigma^{(3)}\in\mathfrak{S}_{\llbracket p-k+1,p-k+q\rrbracket}$\,, $\sigma^{(4)}\in\mathfrak{S}_{\llbracket1,K\rrbracket\setminus
\llbracket p-k+1,p-k+q\rrbracket}$\,.
\end{enumerate}
\end{lem}
Note that:
\begin{itemize}
\item There is no uniqueness of such a decomposition.
\item For $A\subset B$ an element of $\mathfrak{S}_A$ is identified with the corresponding element of $\mathfrak{S}_B$ which is the identity on $B\setminus A$\,.
\item The permutations $\sigma^{(1)}$ and $\sigma^{(2)}$ commute, and so do  $\sigma^{(3)}$ and $\sigma^{(4)}$.
\end{itemize}
\begin{proof}
For Point 1: We count the number of permutation in $\mathfrak{S}(k)$.
We first choose $k$ integers out of $\llbracket p-k+1,p-k+q\rrbracket$
and $k$ integers out of $\llbracket1,p\rrbracket$\,. 
There is $\binom{q}{k}\binom{p}{k}$
such possible choices and $k!$ possible permutations for each of these choices. Then the remaing $q-k$ integers of $\llbracket p-k+1,p-k+q\rrbracket$
have to be sent in $\llbracket p+1,K\rrbracket$\,. 
There is $(q-k)!\binom{K-p}{q-k}$ possibilities for that. In the same way we have $(p-k)!\binom{K-q}{p-k}$
possibilities for the remaining integers of $\llbracket1,p\rrbracket$
that come from $\llbracket 1,K\rrbracket\setminus\llbracket p-k+1,p-k+q\rrbracket$\,.
Finally the $K-k-(q-k)-(p-k)$ remaing integers on both sides can
be permuted in $(K-q-p+k)!$ different ways. So that
\[
\card\mathfrak{S}(k)=\binom{q}{k}\binom{p}{k}k!(q-k)!\binom{K-p}{q-k}(p-k)!\binom{K-q}{p-k}(K-q-p+k)!
\]
and this gives the result.

For Point 2:
Let $A=\sigma^{-1}(\llbracket1,p\rrbracket)\cap\llbracket p-k+1,p-k+q\rrbracket$~.
There exists $\sigma^{(3)}\in\mathfrak{S}_{\llbracket p-k+1,p-k+q\rrbracket}$
such that $\sigma^{(3)}(A)=\llbracket p-k+1,p\rrbracket$~. Then
\[
\sigma\,\sigma^{(3)-1}(\llbracket p-k+1,p\rrbracket)=\sigma(A)\subseteq\llbracket1,p\rrbracket\,.
\]
Hence there exists $\sigma^{(1)}\in\mathfrak{S}_{\llbracket1,p\rrbracket}$
such that $\sigma^{(1)}(j)=\sigma\,\sigma^{(3)-1}(j)$ on $\llbracket p-k+1,p\rrbracket$~.
And, similarly, there exists $\sigma^{(2)}\in\mathfrak{S}_{\llbracket p+1,K\rrbracket}$
such that $\sigma^{(2)}(j)=\sigma\,\sigma^{(3)-1}(j)$ on $\llbracket p+1,p-k+q\rrbracket$~.
Note that $\sigma^{(1)}$ and $\sigma^{(2)}$ commute. 
Finally, we set $\sigma^{(4)}=\sigma^{(2)-1}\sigma^{(1)-1}\sigma\sigma^{(3)-1}$~.
By construction, $\sigma^{(4)}(j)=j$ for $j\in\llbracket p-k+1,p-k+q\rrbracket$~,
hence $\sigma^{(4)}\in\mathfrak{S}_{\llbracket1,K\rrbracket\setminus\llbracket p-k+1,p-k+q\rrbracket}$
and $\sigma=\sigma^{(1)}\sigma^{(2)}\sigma^{(3)}\sigma^{(4)}$ (as
$\sigma^{(4)}$ and $\sigma^{(3)}$ commute).
\end{proof}

Notation: On  $\mathcal{L}(\hil^{\otimes p};\hil^{\otimes q})$\,, the equivalence relation $\cong$ is defined by
$$
A\cong B
\quad
\Leftrightarrow
\quad
\mathcal{S}_{\pm}^{q} \, A \, \mathcal{S}_{\pm}^{p,*}
=\mathcal{S}_{\pm}^{q} \, B \, \mathcal{S}_{\pm}^{p,*}\,.
$$

\begin{lem}
\label{lem:p_b_I_p_b_I_p}Let $\tilde{b}_{j} \in \mathcal{L}(\mathcal{S}_{\pm}^{p_j}\hil^{\otimes p_j};\mathcal{S}_{\pm}^{ q_j}\hil^{\otimes q_j})$
and $n_{j}$ such that $n_{1}+p_{1}=n_{2}+q_{2}=:K$. Then
$$
(\tilde{b}_{1} \otimes \Id^{\otimes n_{1}}) \mathcal{S}_{\pm}^{K,*}
\mathcal{S}_{\pm}^{K} (\tilde{b}_{2} \otimes \Id^{\otimes n_{2}})
\cong \sum_{k}(\pm1)^{(p_{2}+q_{2})(k-p_{1})}\:\frac{n_{2}!\, n_{1}!}{K'!\, K!\, k!}\,(\tilde{b}_{1}\sharp^{k}\tilde{b}_{2})\otimes \Id^{\otimes K^{\prime}}\,,
$$
where $k\in\llbracket\max\{0,p_{1}+q_{2}-K\},\min\{p_{1},q_{2}\}\rrbracket$\,,
and $K^{\prime}=K-q_{2}-p_{1}+k$\,.
\end{lem}
\begin{proof}
Using the partition $\mathfrak{S}_{\llbracket1,K\rrbracket}=\bigsqcup_{k}\tilde{\mathfrak{S}}(k)$
in subsets
\[
\tilde{\mathfrak{S}}(k):=\Big\{\sigma\in\mathfrak{S}_{\llbracket1,K\rrbracket}\;\Big|\;\card\big(\sigma(\llbracket1,q_{2}\rrbracket)\cap\llbracket1,p_{1}\rrbracket\big)=k\Big\}
\]
for $k\in\llbracket\max\{0,p_{1}+q_{2}-K\},\min\{p_{1},q_{2}\}\rrbracket$\,,
yields 
$$
(\tilde{b}_{1}\otimes \Id^{\otimes n_{1}}) \mathcal{S}_{\pm}^{K,*} \mathcal{S}_{\pm}^{K} (\tilde{b}_{2}\otimes \Id^{\otimes n_{2}}) 
\\
=\frac{1}{K!} \sum_{k} \sum_{\tilde{\sigma}\in\tilde{\mathfrak{S}}(k)} 
 (\tilde{b}_{1}\otimes \Id^{\otimes n_{1}}) s_{\pm}(\tilde{\sigma})\tilde{\sigma}\cdot\, (\tilde{b}_{2}\otimes \Id^{\otimes n_{2}}) 
\,.
$$
We fix $k$ and $\tilde{\sigma}\in\tilde{\mathfrak{S}}(k)$. A cyclic
permutation $\tau_r :=(1\,2\,3\,\cdots\, r)$ acting on $\mathcal{Z}^{\otimes r}$
defines the shift operator $\tau_{r}\cdot=(1\,2\,3\,\cdots\, r)\cdot$ and
then $\sigma:=\tilde{\sigma}\, \tau_{K}^{k-p_{1}}$ is in $\mathfrak{S}(k)$
(with $p=p_{1}$ and $q=q_{2}$) and 
\begin{align*}
(\tilde{b}_{1}
 & \otimes \Id^{\otimes n_{1}}) \,
 s_{\pm}(\tilde{\sigma}) \, \tilde{\sigma} \tau_{K}^{k-p_{1}}
 \tau_{K}^{p_{1}-k}\cdot\, (\tilde{b}_{2}\otimes \Id^{\otimes n_{2}})\, \tau_{p_{2}+n_{2}}^{k-p_{1}} 
 \tau_{p_{2}+n_{2}}^{p_{1}-k}\cdot\,
\\
 & \cong (\tilde{b}_{1}\otimes \Id^{\otimes n_{1}}) \, s_{\pm}(\sigma) \, \sigma\cdot\,(\pm1)^{K(k-p_{1})} (\Id^{\otimes p_{1}-k}\otimes\tilde{b}_{2}\otimes \Id^{\otimes K^{\prime}}) (\pm1)^{(p_{2}+n_{2})(k-p_{1})}
\\
 & \cong (\pm1)^{(K+p_{2}+n_{2})(k-p_{1})}  (\tilde{b}_{1}\otimes \Id^{\otimes n_{1}}) s_{\pm}(\sigma)\sigma\cdot\, (\Id^{\otimes p_{1}-k}\otimes\tilde{b}_{2}\otimes \Id^{\otimes K^{\prime}}) 
\end{align*}
holds for operators in $\mathcal{L}(\hil^{\otimes q_1+n_1};\hil^{\otimes p_2+n_2})$\,. We used 
\begin{eqnarray*}
&&s_\pm(\sigma)=s_\pm(\tilde{\sigma})s_\pm(\tau_{K}^{k-p_1})=s_\pm(\tilde{\sigma})(\pm 1)^{K(k-p_1)}\\
\text{and}&&
(\tau_{p_2+n_2}^{p_1-k}\cdot)\circ\mathcal{S}_\pm^{p_2+n_2}=(\pm 1)^{(p_2+n_2)(p_1-k)}\mathcal{S}_\pm^{p_2+n_2}\,.
\end{eqnarray*}
Owing to the factorization $\sigma=\sigma^{(1)}\sigma^{(2)}\sigma^{(3)}\sigma^{(4)}$
of Lemma~\ref{lem:Sigma(k,p,q,K)} with $\sigma^{(i)}\sigma^{(i+1)}=\sigma^{(i+1)}\sigma^{(i)}$ for $i\in\{1,3\}$\,, we get
\begin{align*}
(\tilde{b}_{1}
& \otimes \Id^{\otimes n_{1}}) s_{\pm}(\sigma)\sigma\cdot\, (\Id^{\otimes p_{1}-k} \otimes \tilde{b}_{2} \otimes \Id^{\otimes K^{\prime}}) 
\\
 & \cong (\tilde{b}_{1}\otimes \Id^{\otimes n_{1}}) s_{\pm}(\sigma)(\sigma^{(1)}\sigma^{(2)}\sigma^{(3)}\sigma^{(4)})\cdot\, (\Id^{\otimes p_{1}-k} \otimes \tilde{b}_{2} \otimes  \Id^{\otimes K^{\prime}})
\\
& \cong s_{\pm}(\sigma) ((b_{1}\, \sigma^{(1)}\cdot\,)\otimes \Id^{\otimes n_{1}}\sigma^{(2)}\cdot\,) \,
\sigma^{(4)}\cdot\,
(\Id^{\otimes p_{1}-k}\otimes(\sigma^{(3)}\cdot\,\tilde{b}_{2})\otimes \Id^{\otimes K^{\prime}})
\\
& \cong s_{\pm}(\sigma) (\tilde{b}_{1}s_{\pm}(\sigma^{(1)})\otimes s_{\pm}(\sigma^{(2)}) \Id^{\otimes n_{1}})
s_{\pm}(\sigma^{(4)})
(\Id^{\otimes p_{1}-k}\otimes s_{\pm}(\sigma^{(3)})\tilde{b}_{2}\otimes \Id^{\otimes K^{\prime}})
\\
& \cong (\tilde{b}_{1} \otimes   \Id^{\otimes n_{1}})\, (\Id^{\otimes p_{1}-k} \otimes \tilde{b}_{2} \otimes \Id^{\otimes K^{\prime}})
\\
& \cong [(\tilde{b}_{1}\otimes \Id^{\otimes q_{2}-k})(\Id^{\otimes p_{1}-k}\otimes\tilde{b}_{2})]\otimes \Id^{\otimes K^{\prime}}
\\
& \cong \big(\tfrac{p_{1}!}{(p_{1}-k)!}\tfrac{q_{2}!}{(q_{2}-k)!}\big)^{-1}\,(\tilde{b}_{1}\sharp^{k}\tilde{b}_{2}) \otimes \Id^{\otimes K^{\prime}}.
\end{align*}
We conclude with the first statement of Lemma \ref{lem:Sigma(k,p,q,K)} which counts the terms
in $\sum_{\tilde\sigma\in
\tilde{\mathfrak{S}}(k)}$ 
because 
$\card(\tilde{\mathfrak{S}}(k))=\card(\mathfrak{S}(k))$\,.
\end{proof}

\begin{proof}[Proof of Proposition~\ref{pro:composition_wick}] 
For $n_{1},n_{2}$ such that $n_{1}+p_{1}=n_{2}+q_{2}=:K$, using
Lemma~\ref{lem:p_b_I_p_b_I_p},
\begin{align*}
& \varepsilon^{-\frac{p_{1}+q_{1}+p_{2}+q_{2}}{2}}\times\tilde{b}_{1}^{Wick}\left.\tilde{b}_{2}^{Wick}\right|_{\mathcal{S}_\pm^{n_2+p_2}\hil^{\otimes n_{2}+p_{2}}} \\
& =  \tfrac{\sqrt{K!(n_{1}+q_{1})!}}{n_{1}!}\, \tfrac{\sqrt{(n_{2}+p_{2})!K!}}{n_{2}!}\, 
\mathcal{S}_\pm^{q_1+n_1} 
(\tilde{b}_{1} \otimes \Id^{\otimes n_{1}}) \mathcal{S}_\pm^{p_1+n_1,*}
\mathcal{S}_\pm^{p_2+q_2} 
(\tilde{b}_{2} \otimes \Id^{\otimes n_{2}}) \mathcal{S}_\pm^{p_2+n_2,*}
\\
 & =\sum_{k}(\pm 1)^{(p_{2}+q_{2})(k-p_{1})}\tfrac{\sqrt{(n_{1}+q_{1})!(n_{2}+p_{2})!}}{n_{1}!n_{2}!}\, K!\tfrac{n_{2}!\, n_{1}!}{K'!\, K!\, k!}\,
\mathcal{S}_\pm^{q_1+n_1}
\big( (\tilde{b}_{1}\sharp^{k}\tilde{b}_{2})\otimes \Id^{\otimes K^{\prime}} \big)
\mathcal{S}_\pm^{p_2+n_2,*}
\\
 & =\sum_{k}(\pm 1)^{(p_{2}+q_{2})(k-p_{1})}\tfrac{\sqrt{(q_{2}+q_{1}-k+K^{\prime})!(p_{2}+p_{1}-k+K^{\prime})!}}{K^{\prime}!\, k!}\,
\mathcal{S}_\pm^{q_1+n_1} 
\big((\tilde{b}_{1}\sharp^{k}\tilde{b}_{2})\otimes \Id^{\otimes K^{\prime}} \big)
\mathcal{S}_\pm^{p_2+n_2,*} 
\end{align*}
where $K^{\prime}:=K-q_{2}-p_{1}+k$\,.\\
With $p_2+n_2=p_2+p_1-k+K'$ and $q_1+n_1=q_2+q_1-k+K'$\,, we thus obtain the equality of operators
$$
\tilde{b}_{1}^{Wick}\tilde{b}_{2}^{Wick} =\sum_{k}(\pm 1)^{(p_{2}+q_{2})(k-p_{1})}\,\frac{\varepsilon^{k}}{k!}\,(\tilde{b}_{1}\sharp^{k}\tilde{b}_{2})^{Wick}
$$
restricted to $\mathcal{S}_\pm^{n_2+p_2}\hil^{\otimes n_{2}+p_{2}}$\,.
\end{proof}

\section{A general formula for $\Tr\left[\Gamma_{\pm}(C)\right]$}
\label{sec:genlembos}
The following result about traces of the second quantized operator
 $\Gamma_{\pm}(C)$
 is often presented  for self-adjoint trace-class operators,
 although it is valid without self-adjointness. We recall here the general version for the
 sake of completeness. It relies on a simple holomorphy argument and can be compared with Lidskii's Theorem which
 says that for any trace-class operator $T$\,, 
 $\Tr[T]=\sum_{\lambda\in \sigma(T)}\lambda$\,.
\begin{lem}
\label{le:trbos}
For any trace-class operator $C\in \mathcal{L}^{1}(\Z)$  (which is assumed to
be a strict contraction in the bosonic case, $\pm =+$)\,, its second
quantized version
 $\Gamma_{\pm}(C)$ is trace-class in $\Gamma_{\pm}(\Z)$ and
$$
\Tr\left[\Gamma_{\pm}(C)\right]=\exp\left(\mp \Tr\left[\log(1\mp C)\right]\right)\,.
$$
\end{lem}
\begin{proof}
  When $C=C^{*}\in \mathcal{L}^{1}(\Z)$ 
  using an orthonormal basis of eigenvectors $(e_{n})_{n\in\nz}$ in $\Z$ with the corresponding eigenvalues $(\lambda_n)_{n\in\nz}$\,, and
  $\Gamma_{\pm}(\Z)=\otimes_{n\in\nz}\Gamma_{\pm}(\cz e_{n})$\,, we
  obtain
  \begin{itemize}
  \item in the bosonic case with $\|C\|<1$\,,
    \begin{align*}
\Tr\left[\Gamma_{+}(C)\right]&=\prod_{n\in\nz}\Tr\left[\Gamma_{+}(\lambda_{n}\Id_{\cz})\right]=\prod_{n\in\nz}\frac{1}{1-\lambda_{n}}
=\exp\left(-\sum_{n\in\nz}\log(1-\lambda_{n})\right)\\
&=\exp\left(-\Tr\left[\log(1-C)\right]\right)\,,
    \end{align*}
  \item in the fermionic case\,,
\begin{align*}
\Tr\left[\Gamma_{-}(C)\right]
&=\prod_{n\in\nz}\Tr\left[\Gamma_{-}(\lambda_{n}\Id_{\cz})\right]=\prod_{n\in\nz}{(1+\lambda_{n})}
=\exp\left(+\sum_{n\in\nz}\log(1+\lambda_{n})\right)
\\
&=\exp\left(\Tr\left[\log(1+C)\right]\right)\,.
\end{align*}
\end{itemize}
The functoriality of $\Gamma_{\pm}$ for
 the polar decomposition $C=U|C|$\,, reads
$\Gamma_{\pm}(C)=\Gamma_{\pm}(U)\Gamma_{\pm}(|C|)$\,,
while $\|C\|<1\Leftrightarrow \||C|\|< 1$ in the bosonic case. Hence
$\Gamma_{\pm}(C)$ is trace-class when $C\in {\cal L}^{1}(\Z)$ (and
$\|C\|<1$ in the bosonic case).\\ 
Set ${\cal C}={\cal L}^{1}(\Z)$ in the fermionic case and ${\cal
  C}={\cal L}^{1}(\Z)\cap \left\{C\in {\cal L}(\Z)\,, \|C\|<1\right\}$
in the bosonic case. In both  cases $\mathcal C$ is an open convex set, on which the
two sides of the equality are holomorphic functions. Actually the
holomorphy of the left-hand side comes from series expansion
$$
\Tr\left[\Gamma_{\pm}(C)\right]=\sum_{n=0}^{\infty}\Tr\left[{\cal S}_{\pm}^{n}C^{\otimes
  n}{\cal S}_{\pm}^{n,*}\right]\,,
$$
which converges uniformly in $B(C_{0},\delta_{C_{0}})=\left\{C\in
  {\cal L}^{1}(\Z)\,, \|C-C_{0}\|_{{\cal L}^{1}(\Z)}<\delta_{C_{0}}\right\}$ for
$\delta_{C_{0}}>0$ small enough,  for any $C_{0}\in {\cal
  L}^{1}(\Z)$ (satisfying additionally $\|C_{0}\|< 1$ in the bosonic case).
Actually the estimate $\|C\|_{\mathcal L^1(\Z)}\leq A$ (and $\|C\| \leq \varrho$ with $\varrho<1$ in the bosonic case) imply $\|\,|C|\,\|_{\mathcal L^1(\Z)}\leq A$ (and $\|\,|C|\,\| \leq \varrho$ in the bosonic case). 
Now the inequality 
$$
|\Tr\left[{\cal S}_{\pm}^{n}C^{\otimes n}{\cal S}_{\pm}^{n,*}\right]|\leq \Tr\left[{\cal S}_{\pm}^{n}|C|^{\otimes n}{\cal S}_{\pm}^{n,*}\right]\,,
$$
and the formula in the self-adjoint case with 
\begin{eqnarray*}
&&\sum_{n=0}^\infty \Tr\left[{\cal S}_{-}^{n}|C|^{\otimes n}{\cal S}_{-}^{n,*}\right]\leq \exp(A)
\qquad \text{(fermions)}
\\
\text{resp.}
&&\sum_{n=0}^\infty \Tr\left[{\cal S}_{+}^{n}|C|^{\otimes n}{\cal S}_{+}^{n,*}\right]\leq \exp(\frac{A}{1-\varrho}) 
\qquad
\text{(bosons)} \,,
\end{eqnarray*}
ensure the uniform convergence of the series\,.
 
For any $C\in {\cal C}$\,, $C$ and $\Real C=\frac{C+C^{*}}{2}$ belong
to ${\cal C}$ so that $C(s)=\Real C + i s\Imag C$ belong to ${\cal C}$
when $s\in \omega_0=(-\delta,\delta)+i(-\delta,\delta)$ and when $s\in \omega_1=(1-\delta,1+\delta)+i(-\delta,\delta)$ for $\delta>0$ small enough. By convexity of $\mathcal C$\,, $C(s)\in {\cal C}$ for all $s\in\omega=(-\delta,
1+\delta)+i(-\delta,\delta)$\,. When $s\in i(-\delta,\delta)$\,, $C(s)$
 is self-adjoint and the equality holds. The holomorphy of both sides 
 w.r.t $s\in \omega$ implies that the equality holds true for all
 $s\in \omega$ in particular when $s=1$\,.
\end{proof}

\subsection*{Acknowledgements}
The work of S.B. is supported by the Basque Government through the BERC 2014-2017 program, and by the Spanish Ministry of Economy and Competitiveness MINECO (BCAM Severo Ochoa accreditation SEV-2013-0323, MTM2014-53850), and the European Union's Horizon 2020 research and innovation programme under the Marie Sklodowska-Curie grant agreement No 660021.

\bibliographystyle{amsalpha}

\end{document}